\documentclass[letterpaper,11pt]{article}
\newif\ifdraft
\draftfalse
\usepackage[utf8]{inputenc}
\usepackage{microtype}
\usepackage[letterpaper,margin=1in]{geometry}
\usepackage[numbers,sort&compress]{natbib}
\usepackage{xcolor,xspace}
\usepackage{amsmath}
\usepackage{amssymb}
\usepackage{amsthm}
\usepackage{enumitem}
\usepackage{textgreek}
\usepackage{graphicx}
\usepackage{framed}

\usepackage[small]{caption}
\usepackage{booktabs}
\usepackage[unicode]{hyperref}
\usepackage{cleveref}
\usepackage{doi}
\usepackage{thmtools,thm-restate}

\theoremstyle{plain}
\newtheorem{theorem}{Theorem}[section]
\newtheorem{lemma}[theorem]{Lemma}

\newtheorem{corollary}[theorem]{Corollary}
\newtheorem{claim}[theorem]{Claim}

\newtheorem*{conjecture*}{Conjecture}
\newtheorem{observation}[theorem]{Observation}

\theoremstyle{definition}
\newtheorem{definition}[theorem]{Definition}
\newtheorem{example}[theorem]{Example}

\theoremstyle{remark}
\newtheorem{remark}[theorem]{Remark}
\newtheorem*{remark*}{Remark}

\newcommand{\namedref}[2]{\hyperref[#2]{#1~\ref*{#2}}}

\newcommand{\lparent}{\mathsf {P}}
\newcommand{\lleft}{\mathsf {L}}
\newcommand{\lright}{\mathsf {R}}
\newcommand{\lup}{\mathsf {U}}
\newcommand{\ldown}{\mathsf {D}}
\newcommand{\llch}{\ensuremath{\mathsf {Ch_L}}}
\newcommand{\lrch}{\ensuremath{\mathsf {Ch_R}}}

\newcommand{\lerror}{\mathsf {Error}}

\newcommand{\lin}{\mathsf {in}}
\newcommand{\lout}{\mathsf {out}}
\newcommand{\lgrid}{\mathsf {grid}}
\newcommand{\lgridbot}{\mathsf {bottomGrid}}
\newcommand{\lgridside}{\mathsf {sideGrid}}
\newcommand{\lcolumntrees}{\mathsf {colTree}}
\newcommand{\ltoptree}{\mathsf {topTree}}
\newcommand{\ltreelike}{\mathsf {tree}}
\newcommand{\lstruct}{\mathsf {struct}}
\newcommand{\lproof}{\mathsf {proof}}
\newcommand{\lpointer}{\mathsf {pointer}}

\newcommand{\wild}{\ensuremath{\star}\xspace}

\newcommand{\lerr}{\mathsf {bad}}

\newcommand{\logstar}{\log^{*}}
\newcommand{\eps}{\varepsilon}
\newcommand{\lovasz}{Lov\'{a}sz\xspace}

\newcommand{\LCL}{{\upshape\sffamily LCL}\xspace}
\newcommand{\LOCAL}{{\upshape\sffamily LOCAL}\xspace}
\newcommand{\CONGEST}{{\upshape\sffamily CONGEST}\xspace}
\newcommand{\rake}{\ensuremath{\mathrm{Rake}}}
\newcommand{\compress}{\ensuremath{\mathrm{Compress}}}
\newcommand{\NN}{\ensuremath{\mathbb{N}}}
\newcommand{\uin}{U^{\mathrm{in}}}
\newcommand{\uout}{U^{\mathrm{out}}}
\newcommand{\din}{D^{\mathrm{in}}}
\newcommand{\dout}{D^{\mathrm{out}}}

\newcommand{\etype}{label-set}

\newcommand{\kc}[1]{\ifdraft \textcolor[rgb]{0.9,0.0,0.0}{K: #1}\fi}

\DeclareMathOperator{\poly}{poly}
\DeclareMathOperator{\polylog}{polylog}

\newenvironment{myabstract}
{\list{}{\listparindent 1.5em%
        \itemindent    \listparindent
        \leftmargin    1cm
        \rightmargin   1cm
        \parsep        0pt}%
    \item\relax}
{\endlist}

\newenvironment{mycover}
{\list{}{\listparindent 0pt
        \itemindent    \listparindent
        \leftmargin    1cm
        \rightmargin   1cm
        \parsep        0pt}%
    \raggedright
    \item\relax}
{\endlist}

\newcommand{\myaff}[1]{\,$\cdot$\, {\small #1}\par\smallskip}

\hypersetup{
    colorlinks=true,
    linkcolor=black,
    citecolor=black,
    filecolor=black,
    urlcolor=[HTML]{0088CC},
    pdftitle={Locally Checkable Labelings with Small Messages},
    pdfauthor={A. Balliu, K. Censor-Hillel, Y. Maus, D. Olivetti, J. Suomela}
}

\begin{document}

\begin{mycover}
    {\huge\bfseries Locally Checkable Labelings\\
    with Small Messages \par}
    \bigskip
    \bigskip

    \textbf{Alkida Balliu}
    \myaff{University of Freiburg, Germany}
	\textbf{Keren Censor-Hillel}
	\myaff{Technion, Israel}
	\textbf{Yannic Maus}
	\myaff{Technion, Israel}
	\textbf{Dennis Olivetti}
	\myaff{University of Freiburg, Germany}
	\textbf{Jukka Suomela}
	\myaff{Aalto University, Finland}

\end{mycover}
\bigskip

\begin{myabstract}
\noindent\textbf{Abstract.}
A rich line of work has been addressing the computational complexity of locally checkable labelings (\LCL{}s), illustrating the landscape of possible complexities. In this paper, we study the landscape of \LCL complexities under bandwidth restrictions. Our main results are twofold. First, we show that on trees, the \CONGEST complexity of an \LCL problem is asymptotically equal to its complexity in the \LOCAL model. An analog statement for general (non-\LCL) problems is known to be false. Second, we show that for general graphs this equivalence does not hold, by providing an \LCL problem for which we show that it can be solved in $O(\log n)$ rounds in the \LOCAL model, but requires $\tilde{\Omega}(n^{1/2})$ rounds in the \CONGEST model. 
\end{myabstract}

%\tableofcontents

%!TEX root = congest-lcls.tex

\section{Introduction}\label{sec:introduction}

Two standard models of computing that have been already used for decades to study distributed graph algorithms are the \LOCAL model and the \CONGEST model \cite{peleg00}. In the \LOCAL model, each node in the network can send \emph{arbitrarily large messages} to each neighbor in each round, while in the \CONGEST model the nodes can only send \emph{small messages} (we will define the models in Section~\ref{ssec:models}). In general, being able to send arbitrarily large messages can help a lot: there are graph problems that are trivial to solve in the \LOCAL model and very challenging in the \CONGEST model, and this also holds in trees.

Nevertheless, we show that there is a broad family of graph problems---\emph{locally checkable labelings} or \LCL{}s in short---in which the two models of computing have exactly the same expressive power \emph{in trees} (up to constant factors): if a locally checkable labeling problem $\Pi$ can be solved in trees in $T(n)$ communication rounds in the \LOCAL model, it can be solved in $O(T(n))$ rounds also in the \CONGEST model. We also show that this is no longer the case if we switch from trees to general graphs:

\begin{center}
\begin{tabular}{@{}lcc@{}}
& \emph{\LCL{} problems} & \emph{General problems} \\
& \emph{(our work)} & \emph{(prior work)} \\
\midrule[\heavyrulewidth]
\emph{Trees:}          & \CONGEST $=$ \LOCAL   & \CONGEST $\ne$ \LOCAL \\[2pt]
\emph{General graphs:} & \CONGEST $\ne$ \LOCAL & \CONGEST $\ne$ \LOCAL \\
\bottomrule[\heavyrulewidth]
\end{tabular}
\end{center}

\paragraph{Locally Checkable Labelings.}

The study of the distributed computational complexity of locally checkable labelings (\LCL{}s) in the \LOCAL model was initiated by Naor and Stockmeyer \cite{naor95} in the 1990s, but this line of research really took off only in the recent years \cite{balliu19lcl-decidability,BBHORS19MMlowerBound,BBOS18almostGlobal,BBOS20paddedLCL,BFHKLRSU16,BHKLOS18lclComplexity,BHOS19,binary_lcls,Brandt19RE,Brandt2017,Chang20,ChangKP19,CP19,Olivetti2019REtor,balliu21rooted-trees,BBKOmis,BBO20rs}.

\LCL{}s are a family of graph problems: an \LCL problem $\Pi$ is defined by listing a \emph{finite set of valid labeled local neighborhoods}. This means that $\Pi$ is defined on graphs of some finite maximum degree $\Delta$, and the task is to label the vertices and/or edges with labels from some finite set so that the labeling satisfies some set of local constraints (see Section~\ref{ssec:lcl-definition} for the precise definition).

A simple example of an \LCL problem is the task of coloring a graph of maximum degree $\Delta$ with $\Delta+1$ colors (here valid local neighborhoods are all properly colored local neighborhoods). \LCL{}s are a broad family of problems, and they contain many key problems studied in the field of distributed graph algorithms, including graph coloring, maximal independent set and maximal matching.

\paragraph{Classification of \LCL problems.}

One of the key questions related to the \LOCAL model has been this: given an arbitrary \LCL problem $\Pi$, what can we say about its computational complexity in the \LOCAL model (i.e., how many rounds are needed to solve $\Pi$)? It turns out that we can say quite a lot. There are infinitely many distinct complexity classes, but there are also some wide \emph{gaps} between the classes---for example, if $\Pi$ can be solved with a deterministic algorithm in $o(\log n)$ rounds, it can also be solved in $O(\log^* n)$ rounds \cite{ChangKP19}.

\begin{figure}
    \centering
    \includegraphics[page=1]{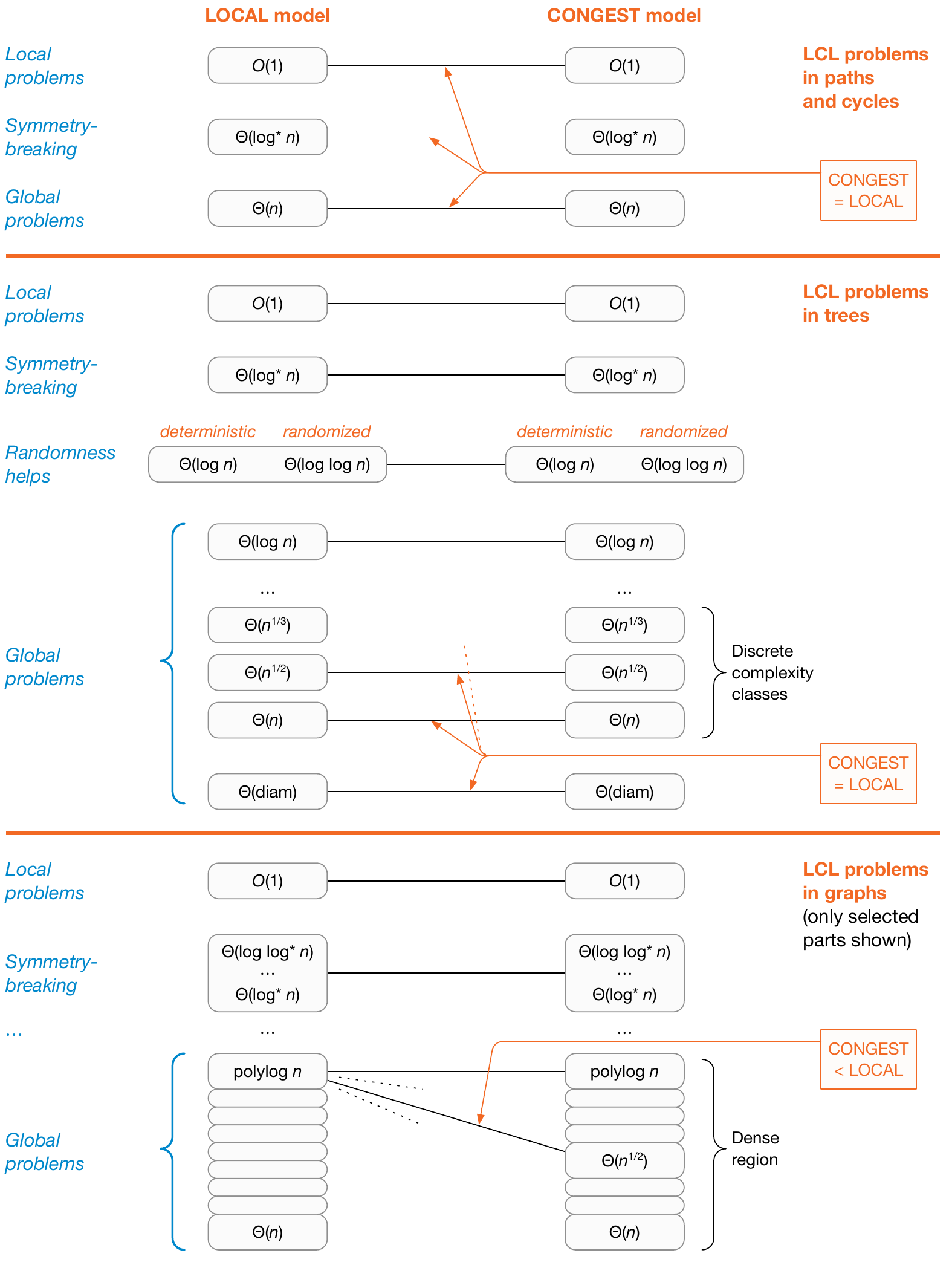}
    \caption{The landscape of \LCL problems in the \LOCAL and \CONGEST models. \kc{I still feel we will be short on space, and could have a smaller figure if we flip it horizontally.}}\label{fig:overview}
\end{figure}

Furthermore, some parts of the classification are \emph{decidable}: for example, in the case of rooted trees, we can feed the description of $\Pi$ to a computer program that can determine the complexity of $\Pi$ in the \LOCAL model \cite{balliu21rooted-trees}. The left part of Figure~\ref{fig:overview} gives a glimpse of what is known about the landscape of possible complexities of \LCL problems in the \LOCAL model.

However, this entire line of research has been largely confined to the \LOCAL model. Little is known about the general structure of the landscape of \LCL problems in the \CONGEST model. In some simple settings (in particular, paths, cycles, and rooted trees) it is known that the complexity classes are the same between the two models \cite{balliu21rooted-trees,lcls_on_paths_and_cycles}, but this has been a straightforward byproduct of work that has aimed at classifying the problems in the \LOCAL model. What happens in the more general case has been wide open---the most interesting case for us is \LCL problems in (unrooted) trees.

\paragraph{Prior work on \LCL{}s in trees.}

In the case of trees, \LCL problems are known to exhibit a broad variety of different complexities in the \LOCAL model. For example, for every $k = 1, 2, 3, \dotsc$ there exists an \LCL problem whose complexity in the \LOCAL model is exactly $\Theta(n^{1/k})$ \cite{CP19}. There are also problems in which randomness helps exponentially: for example, the sinkless orientation problem belongs to the class of problems that requires $\Theta(\log n)$ rounds for deterministic algorithms and only $\Theta(\log \log n)$ rounds for randomized algorithms in the \LOCAL model \cite{Brandt2017,ChangKP19}. Until very recently, one open question related to \LCL{}s in trees remained: whether there are any problems in the region between $\omega(1)$ and $o(\log^* n)$; there is now a (currently unpublished) result \cite{brandt21trees} that shows that no such problems exist, and this completed the classification of \LCL{}s in trees in the \LOCAL model.

\paragraph{Prior work separating \CONGEST and \LOCAL.}

In the \LOCAL model, all natural graph problems can be trivially solved in $O(n)$ rounds and also in $O(D)$ rounds, where $D$ is the diameter of the input graph: in $O(D)$ rounds all nodes can gather the full information on the entire input graph.

However, there are many natural problems that do not admit $O(D)$-round algorithms in the \CONGEST model. Some of the best-known examples include the task of finding an (approximate) minimum spanning tree, which requires $\tilde\Omega(\sqrt{n} + D)$ rounds \cite{Sarma2012,Peleg2000,Elkin2006}, and the task of computing the diameter, which requires $\tilde\Omega(n)$ rounds \cite{FrischknechtHW12,AbboudCK16}. There are also natural problems that do not even admit $O(n)$-round algorithms in the \CONGEST model. For example, finding an exact minimum vertex cover or dominating set requires $\tilde\Omega(n^2)$ rounds \cite{Censor-HillelKP17,BachrachCDELP19}. 

Moreover, separations also hold in some cases where the \LOCAL complexity is constant: One family of such problems is that of detecting subgraphs, for which an extreme example is that for any $k$ there exists a subgraph of diameter 3 and size $O(k)$, which requires $\Omega(n^{2-1/k})$ rounds to detect in \CONGEST, even when the network diameter is also 3 \cite{FischerGKO18}. Another example is spanner approximations, for which there is a constant-round $O(n^{\epsilon})$-approximation algorithm in the \LOCAL model \cite{BarenboimEG18}, but $\tilde{\Omega}(n^{1/2-\epsilon/2})$ rounds are needed in the \CONGEST model (and even $\tilde{\Omega}(n^{1-\epsilon/2})$ rounds for deterministic algorithms) \cite{Censor-HillelD18}. Separations hold also in trees: all-pairs shortest-paths on a star can be solved in 2  \LOCAL rounds, but requires $\tilde{\Omega}(n)$ \CONGEST rounds \cite{Nanongkai14,Censor-HillelKP17}.

While these lower bound results do not have direct implications in the context of \LCL problems, they show that there are many natural settings in which the \LOCAL model is much stronger than the \CONGEST model.

\subsection{Our Contributions}

We show that
\begin{framed}
    \centering
    \LOCAL = \CONGEST for \LCL problems in trees.
\end{framed}
\noindent Not only do we have the same round complexity classes, but every \LCL problem has the same asymptotic complexity in the two models.
In particular, our result implies that \emph{all} prior results related to \LCL{}s in trees hold also in the \CONGEST model. For example all decidability results for \LCL{}s in trees in the \LOCAL model hold also in the \CONGEST model; this includes the decidable gaps in \cite{CP19} and \cite{Chang20}. We also show that the equivalence holds not only if we study the complexity as a function of $n$, but also for problems with complexity~$\Theta(D)$.

Given the above equality, one could conjecture that the \LOCAL model does not have any advantage over the \CONGEST model for any \LCL problem. We show that this is not the case: as soon as we step outside the family of trees, we can construct an example of an \LCL problem that is solvable in $\polylog n$ rounds in the \LOCAL model but requires $\Omega(\sqrt{n})$ rounds in the \CONGEST model. In summary, we show that
\begin{framed}
    \centering
    \LOCAL $\ne$ \CONGEST for \LCL problems in general graphs.
\end{framed}
\noindent Here ``general graphs'' refers to general \emph{bounded-degree} graphs, as each \LCL problem comes with some finite maximum degree $\Delta$. We summarize our main results in Figure~\ref{fig:overview}.

\paragraph{Open questions.}

The main open question after the present work is how wide the gap between \CONGEST and \LOCAL can be made in general graphs. More concretely, it is an open question whether there exists an \LCL problem that is solvable in $O(\log n)$ rounds in the \LOCAL model but requires $\Omega(n)$ rounds in the \CONGEST model.

\subsection{Road Map, Key Techniques, and New Ideas}

We prove the equivalence of \LOCAL and \CONGEST in trees in Sections \ref{sec:tinyRuntimes}--\ref{sec:treeSublog}, and we show the separation between \LOCAL and \CONGEST in general graphs in Section~\ref{sec:separation}. We start in Section~\ref{sec:tinyRuntimes} with some basic facts in the $O(\log n)$ regime that directly follow from prior work. The key new ideas are in Sections \ref{sec:treeSuperlog}--\ref{sec:separation}.

\paragraph{Equivalence in trees: superlogarithmic region (Section~\ref{sec:treeSuperlog}).}

The first major challenge is to prove that any \LCL with some sufficiently high complexity $T(n)$ in the \LOCAL model in trees has exactly the same asymptotic complexity $\Theta(T(n))$ also in the \CONGEST model. A natural idea would be to show that a given \LOCAL model algorithm $A$ can be simulated efficiently in the \CONGEST model. However, this is \emph{not possible} in general---the proof has to rely somehow on the property that $A$ solves an \LCL problem.

Instead of a direct simulation approach, we use as a starting point prior \emph{gap results} related to \LCL{}s in the \LOCAL model. A typical gap result in trees can be phrased as follows, for some $T_2 \ll T_1$:
\begin{itemize}[noitemsep]
    \item Given: a $T_1(n)$-round algorithm $A_1$ that solves some \LCL $\Pi$ in trees in the \LOCAL model.
    \item We can construct: a $T_2(n)$-round algorithm $A_2$ that solves it in the \LOCAL model.
\end{itemize}
We \emph{amplify} the gap results so that we arrive at theorems of the following form---note that we not only speed up algorithms but also switch to a weaker model:
\begin{itemize}[noitemsep]
    \item Given: a $T_1(n)$-round algorithm $A_1$ that solves some \LCL $\Pi$ in trees in the \textbf{\LOCAL} model.
    \item We can construct: a $T_2(n)$-round algorithm $A'_2$ that solves it in the \textbf{\CONGEST} model.
\end{itemize}
At first, the entire approach may seem hopeless: clearly $A'_2$ has to somehow depend on $A_1$, but how could a fast \CONGEST-model algorithm $A'_2$ possibly make any use of a slow \LOCAL-model algorithm $A_1$ as a black box? Any attempts of simulating (even a small number of rounds) of $A_1$ seem to be doomed, as $A_1$ might use large messages.

We build on the strategy that Chang and Pettie \cite{CP19} and Chang \cite{Chang20} developed in the context of the \LOCAL model. In their proofs, $A_2$ does not make direct use of $A_1$ as a black box, but mere \emph{existence} of $A_1$ guarantees that $\Pi$ is sufficiently well-behaved in the sense that long path-like structures can be labeled without long-distance coordination.

This observation served as a starting point in \cite{CP19,Chang20} for the development of an efficient \LOCAL model algorithm $A_2$ that finds a solution that is possibly different from the solution produced by $A_1$, but it nevertheless satisfies the constraints of problem $\Pi$. However, $A_2$ obtained with this strategy abuses the full power of large messages in the \LOCAL model.

Hence while our aim is at understanding the landscape of computational complexity, we arrive at a \emph{concrete algorithm design challenge}: we need to design a \CONGEST-model algorithm $A'_2$ that solves essentially the same task as the \LOCAL-model algorithm $A_2$ from prior work, with the same asymptotic round complexity. We present the new \CONGEST-model algorithm $A'_2$ in Section~\ref{sec:treeSuperlog} (more precisely, we develop a family of such algorithms, one for each gap in the complexity landscape).

\paragraph{Equivalence in trees: sublogarithmic region (Section~\ref{sec:treeSublog}).}

The preliminary observations in Section~\ref{sec:tinyRuntimes} cover the lowest parts of the complexity spectrum, and Section~\ref{sec:treeSuperlog} covers the higher parts. To complete the proof of the equivalence of \CONGEST and \LOCAL in trees we still need to show the following result:
\begin{itemize}[noitemsep]
    \item Given: a randomized $o(\log n)$-round algorithm $A_1$ that solves some \LCL $\Pi$ in trees in the \LOCAL model.
    \item We can construct: a randomized $O(\log \log n)$-round algorithm $A'_2$ that solves it in the \CONGEST model.
\end{itemize}
If we only needed to construct a \LOCAL-model algorithm, we could directly apply the strategy from prior work \cite{CP19,CHLPU18}: replace $A_1$ with a faster algorithm $A_0$ that has got a higher failure probability, use the Lov{\'a}sz local lemma (LLL) to show that $A_0$ nevertheless succeeds at least for some assignment of random bits, and then plug in an efficient distributed LLL algorithm \cite{CHLPU18} to find such an assignment of random bits.

However, there is one key component missing if we try to do the same in the \CONGEST model: a sufficiently fast LLL algorithm. Hence we again arrive at a concrete algorithm design challenge: we need to develop an efficient \CONGEST-model algorithm that solves (at least) the specific LLL instances that capture the task of finding a good random bit assignments for $A_0$.

We present our new algorithm in Section~\ref{sec:treeSublog}. We make use of the shattering framework of \cite{FGLLL17}, but one of the key new ideas is that we can use the equivalence results for the superlogarithmic region that we already proved in Section~\ref{sec:treeSuperlog} as a tool to design fast \CONGEST model algorithms also in the sublogarithmic region.

\paragraph{Separation in general graphs (Section~\ref{sec:separation}).}

Our third major contribution is the separation result between \CONGEST and \LOCAL for general graphs---and as we are interested in proving such a separation for \LCL{}s, we need to prove the separation for bounded-degree graphs.

Our separation result is constructive---we show how to design an \LCL problem $\Pi$ with the following properties:
\begin{enumerate}[label=(\alph*), noitemsep]
    \item There is a deterministic algorithm that solves $\Pi$ in the \LOCAL model in $O(\log n)$ rounds.
    \item Any algorithm (deterministic or randomized) that solves $\Pi$ in the \CONGEST model requires $\Omega(\sqrt{n}/\log^2 n)$ rounds.
\end{enumerate}
To define $\Pi$, we first construct a graph family $\mathcal{G}$ and an \LCL problem $\Pi^{\mathrm{real}}$ such that $\Pi^{\mathrm{real}}$ would satisfy properties (a) and (b) \emph{if we promised that the input comes from family $\mathcal{G}$}. Here we use a bounded-degree version of the lower-bound construction by \cite{Sarma2012}: the graph has a small diameter, making all problems easy in \LOCAL, but all short paths from one end to the other pass through the top of the structure, making it hard to pass a large amount of information across the graph in the \CONGEST model.

However, the existence of \LCL problems that have a specific complexity given some arbitrary promise is not yet interesting---in particular, \LCL{}s with an arbitrary promise do not lead to any meaningful complexity classes or useful structural theorems. Hence the key challenge is \emph{eliminating the promise} related to the structure of the input graph. To do that, we introduce the following \LCL problems:
\begin{itemize}[noitemsep]
    \item $\Pi^{\mathrm{proof}}$ is a distributed proof for the fact $G \in \mathcal{G}$. That is, for every $G \in \mathcal{G}$, there exists a feasible solution $X$ to $\Pi^{\mathrm{proof}}$, and for every $G \notin \mathcal{G}$, there is no solution to $\Pi^{\mathrm{proof}}$. This problem can be hard to solve.
    \item $\Pi^{\mathrm{bad}}$ is a distributed proof for the fact that a given labeling $X$ is \emph{not} a valid solution to $\Pi^{\mathrm{proof}}$. This problem has to be sufficiently easy to solve in the \LOCAL model whenever $X$ is an invalid solution (and impossible to solve whenever $X$ is a valid solution).
\end{itemize}
Finally, the \LCL problem $\Pi$ captures the following task:
\begin{itemize}[noitemsep]
    \item Given a graph $G$ and a labeling $X$, solve either $\Pi^{\mathrm{real}}$ or $\Pi^{\mathrm{bad}}$.
\end{itemize}
Now given an arbitrary graph $G$ (that may or may not be from $\mathcal{G}$) and an arbitrary labeling $X$ (that may or may not be a solution to $\Pi^{\mathrm{proof}}$), we can solve $\Pi$ efficiently in the \LOCAL model as follows:
\begin{itemize}[noitemsep]
    \item If $X$ is not a valid solution to $\Pi^{\mathrm{proof}}$, we will detect it and we can solve $\Pi^{\mathrm{bad}}$.
    \item Otherwise $X$ proves that we must have $G \in \mathcal{G}$, and hence we can solve $\Pi^{\mathrm{real}}$.
\end{itemize}
Particular care is needed to ensure to that even for an adversarial $G$ and $X$, at least one of $\Pi^{\mathrm{real}}$ and $\Pi^{\mathrm{bad}}$ is always sufficiently easy to solve in the \LOCAL model. A similar high-level strategy has been used in prior work to e.g.\ construct \LCL problems with a particular complexity in the \LOCAL model, but to our knowledge the specific constructions of $\mathcal{G}$, $\Pi^{\mathrm{real}}$, $\Pi^{\mathrm{proof}}$, and $\Pi^{\mathrm{bad}}$ are all new---we give the details of the construction in Section~\ref{sec:separation}.

\section{Preliminaries \& Definitions}
\label{sec:definitions}

\paragraph{Formal \LCL{} definition.}\label{ssec:lcl-definition}
An \LCL{} problem $\Pi$ is a tuple $(\Sigma_{\mathrm{in}},\Sigma_{\mathrm{out}},C,r)$ satisfying the following.
\begin{itemize}
	\item Both $\Sigma_{\mathrm{in}}$ and $\Sigma_{\mathrm{out}}$ are constant-size sets of labels;
	\item The parameter $r$ is an arbitrary constant, called \emph{checkability radius} of $\Pi$;
	\item $C$ is a finite set of pairs $(H = (V^H,E^H),v)$, where:
	\begin{itemize}
		\item $H$ is a graph, $v$ is a node of $H$, and the radius of $v$ in $H$ is at most $r$;
		\item Every pair $(v,e) \in V^H \times E^H$ is labeled with a label in  $\Sigma_{\mathrm{in}}$ and a label in $\Sigma_{\mathrm{out}}$.
	\end{itemize}
\end{itemize}
Solving a problem $\Pi$ means that we are given a graph where every node-edge pair is labeled with a label in $\Sigma_{\mathrm{in}}$, and we need to assign a label in $\Sigma_{\mathrm{out}}$ to each node-edge pair, such that every $r$-radius ball around each node is isomorphic to a (labeled) graph contained in $C$. We may use the term \emph{half-edge} to refer to a node-edge pair.

\paragraph{\LOCAL model.}\label{ssec:models}
In the \LOCAL model \cite{linial87}, we have a connected input graph $G$ with $n$ nodes, which communicate unbounded-size messages in synchronous rounds according to the links that connect them, initially known only to their endpoints. A trivial and well-known observation is that a $T$-round algorithm can be seen as a mapping from radius-$T$ neighborhoods to local outputs.

\paragraph{\CONGEST model.}

The \CONGEST model is in all other aspects identical to the \LOCAL model, but we limit the message size: in a network with $n$ nodes, the size of each message is limited to at most $O(\log n)$ bits \cite{peleg00}.

\paragraph{Randomized algorithms.}
We start by formally defining what is a randomized algorithm in our context. We consider randomized Monte Carlo algorithms, that is, the bound on their running time holds deterministically, but they are only required to produce a valid solution with high probability of success.

\begin{definition}[Randomized Algorithm]
	A \emph{randomized algorithm} $\mathcal{A}$ run with parameter $n$, written as $\mathcal{A}(n)$, (known to all nodes) has runtime $t_{\mathcal{A}}(n)$ and is correct with probability $1/n$ on any graph with at most $n$ nodes. There are no unique IDs. Further, we assume that there is a finite upper bound $h_{\mathcal{A}}(n)\in\NN$ on the number of random bits that a node uses on a graph with at most $n$ vertices.

	The \emph{local failure probability} of a randomized algorithm $\mathcal{A}$ at a node $v$ when solving an \LCL{} is the probability that the \LCL{} constraint of $v$ is violated when executing $\mathcal{A}$.
\end{definition}

The assumption that the number of random bits used by a node is bounded by some (arbitrarily fast growing) function $h(n)$ is made in other gap results in the \LOCAL model as well (see e.g.\ ``the necessity of graph shattering'' in \cite{ChangKP19}). Our results do not care about the growth rate of $h(n)$, e.g., it could be doubly exponential in $n$ or even growing faster. Its growth rate only increases the leading constant in our runtime.

The assumption that randomized algorithms are not provided with unique IDs is made to keep our proofs simpler, but it is not a restriction. In fact, any randomized algorithm can, in $0$ rounds, generate an ID assignment, where IDs are unique with high probability. Hence, any algorithm that requires unique IDs can be converted into an algorithm that does not require them, by first generating them and then executing the original algorithm. The ID generation phase may fail, and the algorithm may not even be able to detect it and try to recover from it, but this failure probability can be made arbitrarily small, by making the ID space large enough. Hence, we observe the following.
\begin{observation}
	\label{obs:IdorNoId}
	Let $c\geq 1$ be a constant. For any randomized Monte Carlo algorithm with failure probability  at most $1/n^c$ on any graph with at most $n$ nodes that relies on IDs from an ID space of size $n^{c+2}$ there is a randomized Monte Carlo algorithm with failure probability $2/n^c$ that does not use unique IDs.
\end{observation}

\paragraph{Deterministic algorithms.} 
\begin{definition}[Deterministic Algorithm]
	A \emph{deterministic algorithm} $\mathcal{A}$ run with parameter $n$, written as $\mathcal{A}(n)$, (known to all nodes) has runtime $t_{\mathcal{A}}(n)$ and is always correct on any graph with at most $n$ nodes. We assume that vertices are equipped with unique IDs from a space of size $\mathcal{S}$. We require that a single ID can be sent in a \CONGEST message. If the parameter $S$ is omitted, then it is assumed to be $n^c$, for some constant $c \ge 1$.
\end{definition}

\paragraph{\boldmath Lying about $n$.}
Algorithms are defined such that a parameter $n$ is provided to them, where $n$ represents an upper bound on the number of nodes of the graph. We can nevertheless try to run an algorithm $\mathcal{A}$ with a parameter that violates this promise, that is, we can try to run $\mathcal{A}(n_0)$ on a graph of size $n > n_0$. In this case, we define $\mathcal{A}(n_0)$ as follows:
\begin{itemize}
	\item If the algorithm reaches a state that it would have never reached if the promise were satisfied, then the algorithm must still terminate in $t_{\mathcal{A}}(n_0)$, but it can produce any output.
	\item Otherwise, it behaves exactly as it would behave while running in a graph of size $n_0$.
\end{itemize}
Every time we will lie about $n$, we will make sure that we never satisfy the first case. Examples of unreachable states are the following:
\begin{itemize}
	\item For deterministic algorithms that assume that IDs are in $\{1,\ldots, n^2\}$, if we lie about $n$ and run $\mathcal{A}(N)$ where $N = \log n$, the algorithm can notice that IDs are exponentially larger than what they should be. Hence, a necessary condition in order to lie about $n$ for deterministic algorithms is to compute a new ID assignment, or to use an algorithm that tolerates a larger ID space.
	\item Both deterministic algorithms and randomized algorithms expect to see a node of degree $\le 2$ within their $O(\log n)$ neighborhood (there are no trees containing subtrees of radius $\omega(\log n)$ that only contain nodes of degree $\ge 3$).  If we lie about $n$ and run $\mathcal{A}(N)$ where $N$ is e.g.\ $\log n$, then an algorithm may find no nodes of degree $\le 2$ within its running time. Hence, without additional assumptions on the structure of the graph, we cannot lie about $n$ for algorithms running in $\Omega(\log n)$. 
\end{itemize}

\section{Warm-Up: The \texorpdfstring{\boldmath $O(\log n)$}{O(log n)} Region}
\label{sec:tinyRuntimes}
As a warm-up, we consider the regime of sublogarithmic complexities. In this region, we can use simple observations to show that gaps known for \LCL{} complexities in the \LOCAL model directly extend to the \CONGEST model as well. In the following, we assume that the size of the ID space $\mathcal{S}$ is polynomial in $n$.

We start by noticing that a constant time (possibly randomized) algorithm for the \LOCAL model implies a constant time deterministic algorithm for the CONGEST model as well.
\begin{theorem}\label{thm:constant}
	Let $\Pi$ be an \LCL{} problem. Assume that there is a randomized algorithm for the \LOCAL model that solves $\Pi$ in $O(1)$ rounds with failure probability at most $1/n$. Then, there is a deterministic algorithm for the \CONGEST model that solves $\Pi$ in $O(1)$ rounds.
\end{theorem}
\begin{proof}
	It is known that, the existence of a randomized $O(1)$-round algorithm solving $\Pi$ in the \LOCAL model with failure probability at most $1/n$ implies the existence of a deterministic algorithm solving $\Pi$ in the \LOCAL model in $O(1)$ rounds \cite{naor95,CP19}.
	
	Also, it is known that any algorithm $\mathcal{A}$ running in $T$ rounds in the \LOCAL model  can be normalized, obtaining a new algorithm $\mathcal{A}'$ that works as follows: first gather the $T$-radius ball neighborhood, and then, without additional communication, produce an output.	Since $\Delta=O(1)$, algorithm $\mathcal{A}'$ can be simulated in the \CONGEST model.
\end{proof}

We can use \Cref{thm:constant} to show that, if we have algorithms that lie inside know complexity gaps of the \LOCAL model, we can obtain fast algorithms that work in the \CONGEST model as well.
\begin{corollary}\label{cor:subloglogstar}
	Let $\Pi$ be an \LCL{} problem. Assume that there is a randomized algorithm for the \LOCAL model that solves $\Pi$ in $o(\log \log^* n)$ rounds with failure probability at most $1/n$. Then, there is a deterministic algorithm for the \CONGEST model that solves $\Pi$ in $O(1)$ rounds.
\end{corollary}
\begin{proof}
	In the \LOCAL model, it is known that an $o(\log \log^* n)$-round randomized algorithm implies an $O(1)$-round deterministic algorithm \cite{naor95,CP19}. Then, by applying Theorem \ref{thm:constant} the claim follows.
\end{proof}
While \Cref{cor:subloglogstar} holds for any graph topology, in trees, paths and cycles we obtain a better result.

\begin{corollary}
		Let $\Pi$ be an \LCL{} problem on trees, paths, or cycles. Assume that there is a randomized algorithm for the \LOCAL model that solves $\Pi$ in $o(\log^* n)$ rounds with failure probability at most $1/n$. Then, there is a deterministic algorithm for the \CONGEST model that solves $\Pi$ in $O(1)$ rounds.
\end{corollary}
\begin{proof}
	In the \LOCAL model, it is known that, on trees, paths, or cycles, an $o(\log^* n)$-round randomized algorithm implies an $O(1)$-round deterministic algorithm \cite{naor95,brandt21trees}. Then, by applying Theorem \ref{thm:constant} the claim follows.
\end{proof}

We now show that similar results hold even in the case where the obtained algorithm does not run in constant time.
\begin{theorem}\label{thm:sublogn}
	Let $\Pi$ be an \LCL{} problem. Assume that there is a deterministic algorithm for the \LOCAL model that solves $\Pi$ in $o(\log n)$ rounds, or a randomized algorithm for the \LOCAL model that solves $\Pi$ in $o(\log \log n)$ rounds with failure probability at most $1/n$. Then, there is a deterministic algorithm for the \CONGEST model that solves $\Pi$ in $O(\log^* n)$ rounds.
\end{theorem}
\begin{proof}
	It is known that for solving \LCL{}s in the \LOCAL model, any deterministic $o(\log n)$-round algorithm or randomized $o(\log \log n)$-round algorithm can be converted into a deterministic $O(\log^* n)$-round algorithm \cite{ChangKP19}.
	We exploit the fact that the speedup result of \cite{ChangKP19} produces algorithms that are structured in a normal form. In particular, all problems solvable in $O(\log^* n)$ can also be solved as follows:
	\begin{enumerate}
		\item Find a distance-$k$ $O(\Delta^{2k})$-coloring, for some constant $k$.
		\item Run a deterministic $k$ rounds algorithm.
	\end{enumerate}
	The first step can be implemented in the \CONGEST model by using e.g.\ Linial's coloring algorithm \cite{Linial92}.
	Then, similarly as discussed in the proof of \Cref{thm:constant}, any $T$ rounds algorithm can be normalized into an algorithm that first gathers a $T$-radius neighborhood and then produces an output without additional communication. Hence, also the second step can be implemented in the \CONGEST model in $O(k) = O(1)$ rounds.
\end{proof}
We use this result to show stronger results for paths and cycles.
\begin{corollary}
	Let $\Pi$ be an \LCL{} problem on paths or cycles.  Assume that there is a randomized algorithm for the \LOCAL model that solves $\Pi$ in $o(n)$ rounds with failure probability at most $1/n$. Then, there is a deterministic algorithm for the \CONGEST model that solves $\Pi$ in $O(\log^* n)$ rounds.
\end{corollary}
\begin{proof}
	In the \LOCAL model, it is known that, on paths and cycles, an $o(n)$-round randomized algorithm implies an $O(\log^* n)$-round deterministic algorithm \cite{ChangKP19}.
	The claim follows by applying Theorem \ref{thm:sublogn}.
\end{proof}

\section{Trees: The \texorpdfstring{\boldmath $\Omega(\log n)$}{\textOmega(log n)} Region}
\label{sec:treeSuperlog}

In this section we prove that in the regime of complexities that are at least logarithmic, the asymptotic  complexity to solve any \LCL{} on trees is the same in the \LOCAL and in the \CONGEST model, when expressed as a function of $n$. Combined with the results proved in \Cref{sec:tinyRuntimes} and \Cref{sec:treeSublog}, which hold in the sublogarithmic region, we obtain that the asymptotic complexity of any \LCL{} on trees is identical in \LOCAL and \CONGEST.

\paragraph{Polynomial and subpolynomial gaps on trees.}
Informally, the following theorem states that there are no \LCL{}s on trees with complexity between $\omega(\log n)$ and $n^{o(1)}$, and for any constant integer $k \ge 1$, between $\omega(n^{1/(k+1)})$ and $o(n^{1/k})$, in both the \CONGEST and \LOCAL models, and that the complexity of any problem is the same in both models.

\begin{theorem}[superlogarithmic gaps]\label{thm:gapCONGESTrakeCompress-treePolyGap}
	Let $T^{\mathrm{slow}} = \{n^{o(1)}\} \cup \{o(n^{1/k}) ~|~ k \in \NN^+\}$.
	For $k \ge 1$, let $f(o(n^{1/k})) := O(n^{1/(k+1)})$. Also, let $f(n^{o(1)}) := O(\log n)$.
	
	Let $T \in T^{\mathrm{slow}}$. 
	Let $\Pi$ be any \LCL{} problem on trees that can be solved with a $T$-round randomized \LOCAL algorithm  that succeeds with probability at least $1-1/n$ on graphs of at most $n$ nodes.  The problem $\Pi$ can be solved with a deterministic  $f(T)$-round \CONGEST algorithm.
	
	Given the description of $\Pi$ it is decidable whether there is an $f(T)$-round  deterministic \CONGEST algorithm, and if it is the case then it can be obtained  from the description of $\Pi$.
\end{theorem}
\begin{remark}
		The deterministic complexity in \Cref{thm:gapCONGESTrakeCompress-treePolyGap} suppresses an $O(\logstar |\mathcal{S}|)$ dependency on the size $|\mathcal{S}|$ of the ID space. Also, there is an absolute constant $l_{\Pi}=O(1)$ such that the \CONGEST algorithm works even if, instead of unique IDs, a distance-$l_{\Pi}$ input coloring from a color space of size $|\mathcal{S}|$ is provided.
\end{remark}

Note that gap theorems do not hold if one has promises on the input of the \LCL{}. For example, consider a path, where some nodes are marked and others are unmarked, and the problem requires to $2$-color unmarked nodes. If we have the promise that there is at least one marked node every $\sqrt{n}$ steps, then we obtain a problem with complexity $\Theta(\sqrt{n})$, that does not exist on paths for \LCL{}s without promises on inputs.

Since \Cref{thm:gapCONGESTrakeCompress-treePolyGap} shows how to construct \CONGEST algorithms starting from \LOCAL ones,  the existence of \CONGEST problems with complexity $\Theta(n^{1/k})$ follows implicitly from the existence of these complexities in the \LOCAL model. In particular, Chang and Pettie devised a series of problems that they name $2\frac{1}{2}$-coloring \cite{CP19}. These problems are parameterized by an integer constant $k>1$ and have complexity $\Theta(n^{1/k})$ on trees.

\paragraph{Diameter time algorithms.}
Additionally, we prove that a randomized diameter time \LOCAL algorithm is asymptotically not more powerful than a deterministic diameter time \CONGEST algorithm, when solving \LCL{}s on trees. This result can be seen as an orthogonal result to the remaining results that we prove for \LCL{}s on trees, because the runtime is not expressed as a function of $n$, but as a function of a different parameter, that is, the diameter of the graph. While the result might be of independent interest, it mainly deals as a warm-up to explain the proof of the technically more involved \Cref{thm:gapCONGESTrakeCompress-treePolyGap}. 
\begin{theorem}[diameter algorithms]\label{thm:treediam}
	Let $\Pi$ be an \LCL{} problem on trees that can be solved with a randomized \LOCAL algorithm running in $O(D)$ rounds that succeeds with high probability, where $D$ is the diameter of the tree. The problem $\Pi$ can be solved with a deterministic \CONGEST algorithm running in $O(D)$ rounds.
	The \CONGEST algorithm does not require unique IDs but a means to break symmetry between adjacent nodes, that can be given by unique IDs, an arbitrary input coloring, or an arbitrary orientation of the edges. 	
	%	Further, it is decidable, given an \LCL{} problem $\Pi$ whether an $O(D)$ round algorithm exists, that is, it is decidable whether an \LCL{} problem $\Pi$ can be solved on any tree.
\end{theorem}

Any solvable \LCL{} problem on trees can trivially be solved in the \LOCAL model in $O(D)$ rounds by gathering the whole tree topology at a leader node, who then locally computes a solution and distributes it to all nodes. Using pipelining, the same algorithm can be simulated in the \CONGEST model in $O(D+n\cdot\Delta)=O(n)$ rounds, but this running time can still be much larger than the $O(D)$ running time obtainable in the \LOCAL model. On a high level, we  show that for \LCL{}s on trees, it is not required to gather the whole topology on a single node and brute force a solution---gathering the whole information at a single node has an $\Omega(n)$ lower bound even if the diameter is small.

\paragraph{Black-white formalism.}
In order to keep our proofs simple, we consider a simplified variant of \LCL{}s, called \emph{LCLs in the black-white formalism}. The main purpose of this formalism is to reduce the radius required to verify if a solution is correct.  We will later show that, on trees, the black-white formalism is in some sense equivalent to the standard \LCL{} definition.

A problem $\Pi$ is a tuple $(\Sigma_{\mathrm{in}},\Sigma_{\mathrm{out}},C_W,C_B)$ satisfying the following.
\begin{itemize}
	\item Both $\Sigma_{\mathrm{in}}$ and $\Sigma_{\mathrm{out}}$ are constant-size sets of labels;
	\item $C_B$ and $C_W$ are sets of multisets of pairs of labels, where each pair $(i,o)$ satisfies $i \in \Sigma_{\mathrm{in}}$ and $o \in \Sigma_{\mathrm{out}}$.
\end{itemize}
Solving a problem $\Pi$ means that we are given a bipartite two-colored graph where every edge is labeled with a label in $\Sigma_{\mathrm{in}}$, and we need to assign a label in $\Sigma_{\mathrm{out}}$ to each edge, such that for every black (resp. white) node, the multiset of pairs of input and output labels assigned to the incident edges is in $C_B$ (resp. $C_W$).

\paragraph{Node-edge formalism.}
The black-white formalism allows us to define problems also on graphs that are not bipartite and two-colored, as follows. Given a graph $G$, we define a bipartite graph $H$, where white nodes correspond to nodes of $G$, and black nodes correspond to edges of $G$. A labeling of edges of $H$ corresponds to a labeling of node-edge pairs of $G$. The constraints $C_W$ of white nodes of $H$ correspond to node constraints of $G$, and the constraints $C_B$ of black nodes of $H$ corresponds to edge constraints of $G$.

\paragraph{Equivalence on trees.}
Clearly, any problem that can be defined with the node-edge formalism can be also expressed as a standard \LCL{}. We now show that the node-edge formalism and the standard \LCL{} formalism are in some sense equivalent, if we restrict to trees.
\begin{claim}
	\label{claim:LCLequivalence}
	For any \LCL{} problem $\Pi$ with checkability radius $r$ we can define an \emph{equivalent} node-edge checkable problem $\Pi'$. In other words, given a solution for $\Pi'$, we can find, in $O(r)$ rounds, a solution for $\Pi$, and vice versa.
\end{claim}
\begin{proof}
	Given $\Pi = (\Sigma_{\mathrm{in}},\Sigma_{\mathrm{out}},C,r)$, we define $\Pi' = (\Sigma'_{\mathrm{in}},\Sigma'_{\mathrm{out}},C'_W,C'_B)$ as follows.
	\begin{itemize}
		\item $\Sigma'_{\mathrm{in}} = \Sigma_{\mathrm{in}}$;
		\item $\Sigma'_{\mathrm{out}}$ contains all the triples $(H,v,j)$ such that there exists a pair $(H,v) \in C$ and $1 \le j \le \delta(v)$, where $\delta(v)$ is the degree of $v$;
		\item $C'_W$ contains all the sets $\{(i_j,(H,v,j)) ~|~ 1 \le j\le \delta(v)\}$ satisfying that $(H,v,j) \in \Sigma'_{\mathrm{out}}$, $i_j \in \Sigma'_{\mathrm{in}}$, and the $j$-th port of $v$ in $H$ has input label $i_j$;
		\item $C'_E$ contains all the multisets $\{(i,(H,v,j)),(i',(H',v',j'))\}$ satisfying that the neighbor $u$ of $v$ reachable on $H$ through port $j$ has the same $(r-1)$-radius neighborhood of $v'$ in $H'$, the neighbor $u'$ of $v'$ reachable on $H'$ through port $j'$ has the same  $(r-1)$-radius neighborhood of $v$ in $H$, the $j$-th port of $v$ in $H$ has input label $i$, and the $j'$-th port of $v'$ in $H'$ has input label $i'$.
	\end{itemize}
	We now prove that, on trees, given a solution for $\Pi$ we can find, in constant time, a solution for $\Pi'$, and vice versa. In order to solve $\Pi'$ given a solution for $\Pi$, each node can spend $r$ rounds to gather its $r$-radius neighbor. Note that such neighborhood must be contained in $C$.
	Hence, there exists some pair $(H,v) \in C$ that corresponds to the neighborhood of the node. Each node outputs, on each port $j$, the triple $(H,v,j)$. The constraints $C_N$ are clearly satisfied, since the same $(H,v)$ pair is given for each port. The constraints $C_E$ are also satisfied, since the $(H,v)$ pairs given in output by the nodes come from the same global assignment.
	
	In order to solve $\Pi$ given a solution for $\Pi'$, each node $u$ of a graph $G$ can, in $0$ rounds, map the labeling $(H,v,j)$ of each of its ports $j'$ into the labeling of the $j$-th port of $v$ in $H$ (note that, in general, $j$ can differ from $j'$, that is, the $j$-th port of $v$ in $H$ may be different from the port of $u$ in $G$, but this is not an issue, as the constraints of an \LCL cannot refer to a specific port numbering). The solution is correct since, in order for the solution of $\Pi'$ to be locally correct everywhere, it must hold that the $(H,v)$ pairs given in output by the nodes must exactly encode the output of the nodes in their radius-$r$ neighborhood. This is not necessarily true in general graphs, but only on trees, and for example it is not possible to encode triangle-detection in this formalism.
\end{proof}
By \Cref{claim:LCLequivalence}, on trees, all \LCL{}s can be converted into an equivalent node-edge checkable \LCL{}, and note that the node-edge formalism is a special case of the black-white formalism where black nodes have degree $2$. To make our proofs easier to read, in the rest of the section we prove our results in the black-white formalism (where the degree of black nodes is $2$), but via \Cref{claim:LCLequivalence} all results also hold for the standard definition of \LCL{}s. We start by proving \Cref{thm:treediam}.

\begin{proof}[Proof of \Cref{thm:treediam}]
Recall that in the black and white formalism nodes are properly $2$-colored, input and output labels are on edges (that is, there is only one label for each edge, and not one for each node-edge pair), and the correctness of a solution can be checked independently by black and white nodes by just inspecting the labeling of their incident edges. Assume we are given a tree. We apply the following algorithm, that is split into $3$ phases.
\begin{enumerate}
	\item{\textbf{Rooting the tree.}} By iteratively 'removing' nodes with degree one from the tree, nodes can produce an orientation that roots the tree in  $O(D)$ \CONGEST rounds. This operation can be performed even if, instead of IDs, nodes are provided with an arbitrary edge orientation. In the same number of rounds, nodes can know their distance from the (computed) root in the tree. We say that nodes with the same distance to the root are in the same \emph{layer}, where leaves are in layer $1$, and the root is in layer $L=O(D)$.
	\item{\textbf{Propagate \etype s up.}} We process nodes layer by layer, starting from layer $1$, that is, from the leaves. Each leaf $u$ of the tree tells its parent $v$ which labels for the edge $\{u,v\}$ would make them happy. The set of these labels is what we call a \emph{\etype{}}.  A leaf $u$ is \emph{happy} with a label if the label satisfies $u$'s \LCL{} constraints in $\Pi$. Then, in round $i$, each node $v$ in layer $i$ receives from all children the sets of labels that make them happy, and tells to its parent $w$ which labels for the edge $\{v,w\}$ would make it happy, that is, it also sends a \emph{\etype}. Here $v$ is \emph{happy} with a label for the edge $\{v,w\}$ if it can label all the edges $\{\{v,u\} \mid u\text{ is a child of $v$}\}$ to its children such that all of its children are happy. In other words, it must hold that for any element in the \etype{} sent by $v$ to $w$, there must exist a choice in the sets previously sent by the children of $v$ to $v$, such that the constraints of $\Pi$ are satisfied at $v$. In $O(D)$ rounds, this propagation of sets of \etype s reaches the root of the tree.
	\item{\textbf{Propagate final labels down.}} We will later prove that, if $\Pi$ is solvable, then the root can pick labels that satisfy its own \LCL{} constraints and makes all of its children happy. Then, layer by layer, in $O(D)$ iterations, the vertices pick labels for all of their edges to their children such that their own \LCL{} constraints are satisfied and all their children are happy.  More formally, from phase $2$ we know that, for any choice made by the parent, there always exists a choice in  the \etype s  previously sent by the children, such that the \LCL{} constraints are satisfied on the node.
\end{enumerate} 
This algorithm can be implemented in $O(D)$ rounds in \CONGEST as the \etype s that are propagated in the second phase are subsets of the finite alphabet $\Sigma$. In the third phase, each vertex only has to send one final label per outgoing edge to its children.

Let $L_{v,u}$ be the \etype{} received by node $v$ from its child $u$. We now prove that, if $\Pi$ is solvable, then there is a choice of labels, in the \etype s received by the root, that makes the root happy. Since $\Pi$ is solvable, then there exists an assignment $\phi$ of labels to the edges of the tree such that the constraints of $\Pi$ are satisfied on all nodes.  We prove, by induction on the layer number $j$, that every edge $\{u,v\}$ such that $u$ is in layer $j$ and $v$ is the parent of $u$, satisfies that $\phi(\{u,v\}) \in L_{v,u}$. For $j=1$ the claim trivially holds, since leaves send to their parent the set of all labels that make them happy. For $j>1$, consider some node $v$ in layer $j$. Let $u_1,\ldots,u_d$ be its children. By the induction hypothesis it holds that $\phi(\{v,u_i\}) \in L_{v,u_i}$. Let $p$ be the parent of $v$ (if it exists). Since $\phi$ is a valid labeling, and since it is true that if the parent labels the edge $\{v,p\}$ with label $\phi(\{v,p\})$ then $v$ can pick something from its sets $L_{v,u_i}$ and be happy, then the algorithm puts $\phi(\{v,p\})$ into the set $L_{p,v}$.
Hence, every set $L_{r,u_i}$ received by the root $r$ contains the label $\phi(\{r,u_i\})$, implying that there is a valid choice for the root.
\end{proof}

%Decidability: You only operate on sets of labels. You begin with all labels that a degree 1 vertex is happy with. And then you need to see what's compatible with what for each number of $k\in [\Delta-1]$ sets that you have already produced? You keep doing this until no new sets appear. If you ever get the empty set you loose.
%One more thing for the root: actually you always take $k \in [\Delta]$ of them, and see if there is a "good choice" for a node of degree k if not, then we can make a bad "root". In this case you add nothing to the set of sets. things should get more complicated with inputs

While this process is extremely simple, its runtime of $O(D)$ rounds is rather slow. In order to obtain the $O(n^{1/k})$ and $O(\log n)$ \CONGEST algorithms required for proving \Cref{thm:gapCONGESTrakeCompress-treePolyGap}, we need a more sophisticated approach.
 As done by prior work in \cite{CP19,Chang20}, we  decompose the tree into \emph{fewer} layers, and show that, the mere existence of a fast algorithm for the problem implies that, similar to the algorithm of \Cref{thm:treediam}, it is sufficient to propagate constant sized \etype s between the layers; thus we obtain a complexity that only depends on the number of layers and their diameter. The rest of the section is split into two parts. In \Cref{ssec:treeDecomp} we restate results from \cite{CP19,Chang20} on decompositions of trees for the \LOCAL model to show that they can immediately be implemented in the \CONGEST model with the same guarantees. In \Cref{ssec:suplogarithmicGaps} we prove \Cref{thm:gapCONGESTrakeCompress-treePolyGap}.

\subsection{Generalized Tree Decomposition}
\label{ssec:treeDecomp}

In this section we present a modified version of the generalized \rake{} \& \compress{} algorithm from \cite{CP19,Chang20}, which is a generalization of the Miller and Reif tree contraction procedure \cite{Miller1985}.

\paragraph{Notations.} For a graph $G=(V,E)$ and two vertices $v,u$, let $\mathrm{dist}(v,u)$ be the hop distance of a shortest path between $v$ and $u$. For a set $M\subseteq V$, we have $\mathrm{dist}(v,M) = \min_{u\in M} \mathrm{dist}(v,u)$.
For $\alpha, \beta\geq 1$, an $(\alpha, \beta)$-ruling set of a graph $G$ is a subset $M\subseteq V$  of the nodes such that $\mathrm{dist}(v, M\setminus \{v\})\geq \alpha$ for all $v\in M$ and $\mathrm{dist}(v,M)\leq \beta$ for all $v\in V$. If $\alpha$ and $\beta$ are constants and $\beta \ge \alpha-1$, then in constant degree graphs $(\alpha,\beta)$-ruling sets can be computed in $O(\logstar |\mathcal{S}|)$ \CONGEST rounds, where $S$ is the size of the ID space \cite{Linial92}. Actually, it is sufficient to have an input  distance coloring with $S$ colors, that is, a coloring in which each color only appears at most once in each $O(1)$-hop neighborhood, where the constant depends on $\alpha$ and $\beta$.

\paragraph{Modified generalized tree decomposition.} The following tree decomposition algorithm depends on two parameters, $\gamma$ (potentially a function of $n$) and $l$ (a constant), and computes a partition of the vertex set of a tree graph into layers satisfying several properties that are stated in \Cref{lem:rakeCompress}.

The algorithm proceeds in iterations, in each of which one layer is produced and removed from the graph. At the core of each iteration are two operations that we define next. The \rake{} operation removes nodes of degree $1$ (if two adjacent nodes have degree $1$ the \rake{} operation only removes one of them by exploiting an arbitrary symmetry breaking between the nodes, e.g., an arbitrary orientation of the edge). The \compress{} operation removes all nodes that have degree $2$ and are contained in a path of degree-$2$ nodes of length at least $l$, in the graph induced by the remaining nodes. In both operations, all degrees are with respect to the graph induced by the nodes that have not yet been removed.

The decomposition algorithm begins with the full tree $G=(V,E)$ (no nodes have been removed initially). In iteration $i=1,2,\ldots$ do the following.
\begin{enumerate}
	\item Do $\gamma$ \rake{} operations.
	\item Do one \compress{} operation.
\end{enumerate}
We denote by $V^R_i$ the set of nodes removed during a \rake{} operation in iteration $i$, and by $V^C_i$ the set of nodes removed during a \compress{} operation during iteration $i$. We call $V^R_i$ and $V^C_i$ \emph{rake layers} and \emph{compress layers}, respectively, and their nodes are called \emph{rake nodes} and \emph{compress nodes}. The modification we introduce over the decomposition in \cite{Chang20} is that, in addition to the layers, we split each rake layer $V^R_i$ into $\gamma$ \emph{sublayers} $V^R_{i,j}$, for $1\le j \le \gamma$, where nodes belong to  $V^R_{i,j}$ if they have been removed during the $j$th \rake{} operation of iteration $i$. For simplicity, we sometimes refer to a compress layer $V^C_i$  both as a layer and as a sublayer of its own layer.
We assign a total order on all sublayers based on when they have been created, and we define a sublayer $V$ to be \emph{higher} than another sublayer $V'$ if $V$ is created after $V'$.

\textbf{The post-processing step.} In this step, some  compress nodes in long paths  are promoted to a higher layer. More precisely, some nodes of layer $V^C_i$, are moved into the rake sublayer $V^R_{i+1,1}$, such that all paths remaining in compress layers have length between $l$ and $2l$, and no two nodes that are adjacent to each other are both promoted, and nodes that are endpoints of a path (before the promotion) are never promoted. The essential ingredient to implement this step is a modified $(l+1,l)$-ruling set algorithm on subgraphs with maximum degree $2$. The property that no adjacent nodes and no endpoints are promoted is used in the proof of \Cref{lem:rakeCompress}.

The following is an analog of what is needed in \cite{Chang20}, which we prove to hold also for our modified decomposition and in \CONGEST.
\begin{lemma}[Modified generalized tree decomposition]
	\label{lem:rakeCompress}
Let $S$ be the size of the ID space or the number of colors in a $2l$-distance coloring. Let $\gamma=O(n^{1/k})$ where $k\in \NN$ and $l=O(1)$ be parameters.
	The modified generalized tree decomposition can be implemented in $O(\logstar |\mathcal{S}|+ k\gamma)$ rounds in the \CONGEST model and yields a decomposition into $2k-1$ layers $V_1^R = (V^R_{1,1},\ldots,V^R_{1,\gamma}), \ldots, V_{k}^R  = (V^R_{k,1},\ldots,V^R_{k,\gamma})$, $V_1^C, \ldots, V_{k-1}^C$ such that the following hold.
	\begin{enumerate}
		\item Compress layers: The connected components of each $G[V_i^C]$ are paths of length in $[l,2l]$, the endpoints have a neighbor in a higher layer, and all other nodes do not have neighbors in a higher layer.
		\item Rake layers: The diameter of the connected components in $G[V_i^R]$ is $O(\gamma)$, and at most one node has a neighbor in a layer above.
		\item The connected components of each sublayer $G[V^R_{i,j}]$ consist of isolated nodes.
	\end{enumerate}

	We obtain the same guarantees in $O(\logstar |\mathcal{S}| + \log n)$ \CONGEST rounds if $\gamma=1$ and $l=O(1)$. In this case, the number of layers is $O(\log n)$. 
	If $\gamma=D$, the decomposition only consists of one rake layer with $D$ sublayers, can be computed in $O(\gamma)$ \CONGEST rounds and unique IDs can be replaced with the weaker assumption of an arbitrary orientation of the edges. 	
\end{lemma}
\begin{proof}
	The first two items are proven in \cite{Chang20}, and do not change by our additional sublayer notation. For the last item, consider some sublayer $V_{i,j}^R$ for $1\leq i\leq k$ and $1\leq j\leq \gamma$. The nodes that are added to the sublayer by a \rake{} operation are an independent set. For $i>1$ and $j=1$, consider the set of nodes  $P_{i-1}^C$ that are promoted to rake layer $V^R_{i,1}$ from the compress layer $V^C_{i-1}$. nodes in $P_{i-1}^C$ form an independent set and no vertex in $P_{i-1}^C$ can be adjacent to a vertex in $V^R_{i,1}\setminus P_{i-1}^C$ as the endpoints of paths are never promoted.

	Next, we argue about the \CONGEST implementation. 
The \rake{} operation requires nodes to know their degree in the graph induced by nodes that are not already removed, which can be done with no overhead in \CONGEST.  To decide whether to be removed in a \compress{} iteration, a node needs to know whether it belongs to a path of degree-$2$ nodes of length at least $l$ in the graph induced by remaining nodes. This can be done in $O(l)=O(1)$ rounds in \CONGEST. The post-processing step requires to compute an $(l+1,l)$-ruling set in paths, and nodes can locally detect if they are part of the subgraph on which this algorithm is executed, which can be done in $O(l+\logstar |\mathcal{S}|)$ rounds (given unique IDs from an ID space of size $|\mathcal{S}|$ or a large enough distance coloring with $S$ colors). This is executed simultaneously on all layers.
\end{proof}

\subsection{Superlogarithmic Gaps}
\label{ssec:suplogarithmicGaps}

In this section we  prove \Cref{thm:gapCONGESTrakeCompress-treePolyGap}. We start by defining the notion of \emph{class}, an object essential to show how to obtain a faster variant of the $O(D)$-round algorithm presented before, while keeping it bandwidth efficient. Recall that we are considering the black and white formalism, where a solution is given by labeling each edge. Consider a tree $G$, and a subtree $H$ of $G$ connected to $G \setminus H$ through a set of incoming edges and a set of outgoing edges. Informally, a class captures how, in a correct label assignment, the labeling for the outgoing edges can depend on the labeling of the incoming edges.

\begin{definition}[\etype{}, class, maximal class, independent class]
	Let $\Pi$ be an \LCL{} problem given in the black and white formalism. Consider a connected subgraph $H\subseteq G$ and a partition of $E(H, G\setminus H)$ into incoming edges $F_{\mathrm{in}}$ and outgoing edges $F_{\mathrm{out}}$.
	Further, consider sets of labels $\mathcal{L}_{\mathrm{in}}=(L_{e})_{e\in F_{\mathrm{in}}}$ (that is, we are given a set of labels $L_e$ for each incoming edge $e$).
	The \emph{\etype} of an edge $e$ is the set of labels $L_e$ associated to the edge $e$.
	Then, a \emph{feasible labeling} of $E(H)$ and $F_{\mathrm{out}}$ with regard to  $\mathcal{L}_{\mathrm{in}}$ is a tuple 	$\big(L_{\mathrm{out}}, L_{\mathrm{in}}, L_H\big)$ where:
	\begin{itemize}
		\item $L_{\mathrm{in}}$ is a labeling $(l_e)_{e\in F_{\mathrm{in}}}$ of $F_{\mathrm{in}}$ satisfying $l_e\in (\mathcal{L}_{\mathrm{in}})_e$ for all $e \in F_{\mathrm{in}}$,
		\item $L_{\mathrm{out}}$ is a labeling of ${F_{\mathrm{out}}}$,
		\item $L_H$ is a labeling of  $E(H)$,
		\item the output labeling of the edges incident to nodes of $H$ given by the $L_{\mathrm{out}}, L_{\mathrm{in}},$ and $L_H$, and the provided input labeling for the edges incident to nodes of $H$, is such that all node constraints of each node $v\in V(H)$ are satisfied.
	\end{itemize}
	Also, given $\Pi, H, F_{\mathrm{in}}, F_{\mathrm{out}}$, and  $\mathcal{L}_{\mathrm{in}}$, we define the following:
	\begin{itemize}
		\item a \emph{class}  is a set of feasible labelings,
		\item a \emph{maximal class}  is the unique inclusion maximal class, that is, it is the set of all feasible labelings,
		\item an \emph{independent class} is a class $A$ such that
		for any $\big(L_{\mathrm{out}}, L_{\mathrm{in}}, L_H\big)\in A$ and $\big(L'_{\mathrm{out}}, L'_{\mathrm{in}}, L'_H\big)\in A$ the following holds. Let $L''_{\mathrm{out}}$ be an arbitrary combination of $L_{\mathrm{out}}$ and $L'_{\mathrm{out}}$, that is, $L''_{\mathrm{out}} = (l_e)_{e\in F_{\mathrm{out}}}$ where $l_e \in \{ (L_{\mathrm{out}})_e, (L'_{\mathrm{out}})_e \}$. There must exist some $L''_H$ and $L''_{\mathrm{in}}$ satisfying $\big(L''_{\mathrm{out}}, L''_{\mathrm{in}}, L''_H\big)\in A$.
	\end{itemize}
	Note that the maximal class with regard to some given $\Pi, H, F_{\mathrm{in}}, F_{\mathrm{out}},$ and $\mathcal{L}_{\mathrm{in}}$, is unique.
\end{definition}

In other words, given a subgraph $H$ with multiple incoming and outgoing edges, where to each incoming edge is associated a set of labels (the \etype{} of the edge), a class represents valid assignments of labels for the outgoing edges, that can be completed into valid assignments for the whole subgraph, such that the labels assigned to the incoming edges belong to their \etype{}. An independent class captures a stronger notion, where, even if labels for the outgoing edges are chosen without coordination, it is still possible to complete the rest of the subgraph.

The above definition is mainly used in only two kinds of subgraphs. In particular, we apply the tree decomposition defined in \Cref{ssec:treeDecomp} and decompose the given tree in many subtrees, that could be of two kinds: either isolated nodes, that are the ones removed with a \rake{} operation (or nodes that have been removed with a \compress{} operation and then promoted into a higher rake layer), or short paths, that are the ones obtained by using a ruling set algorithm to break the (possibly long) paths removed with a \compress{} operation into paths of constant length. We propagate \emph{\etype s} through the layers of the tree decomposition. More in detail, we consider the two following cases for propagating \etype s:
\begin{itemize}
	\item \textbf{Case 1 (isolated nodes, outgoing edge):} the graph $H$ consists of a single node $v$ with a single outgoing edge $e$ ($F_{\mathrm{out}} = \{e\}$) and a constant number of incoming edges $F_{\mathrm{in}}$ (potentially zero).
	Let $B$ be the maximal class of $H$ with regard to $\Pi, H, F_{\mathrm{in}}, F_{\mathrm{out}},$ and $\mathcal{L}_{\mathrm{in}}$. Then, omitting these dependencies, we denote $g(v)=\bigcup_{(L_{\mathrm{out}}, L_{\mathrm{in}}, L_H)\in B}\{(L_{\mathrm{out}})_e\}$. We have $g(v)\subseteq \Sigma$. Later, we will use the value of $g(v)$ as the \etype{} that we associate to the edge $e$.
	\begin{observation}
		\label{obs:computingRakeType}
		Each node $v$ can compute $g(v)$, given the value of $g(u)$ for each incoming edge $\{v,u\}$.
	\end{observation}
	\item \textbf{Case 2 (short paths):} the graph $H$ is a path of constant length. The endpoints of the path are $v_1$ and $v_2$, $F_{\mathrm{out}} = \{e_1,e_2\}$, and edge $e_i$ ($i \in \{1,2\}$) is incident to $v_i$. Let $B$ be the maximal class of $H$.
	We assume to be given a function $f_{\Pi,\gamma}$, that depends only on $\Pi$ and some parameter $\gamma$, that maps a class $B$ into an independent class $B'$. For $i \in \{1,2\}$, let $g(v_i) = \bigcup_{(L_{\mathrm{out}}, L_{\mathrm{in}}, L_H)\in f_{\Pi}(B)}\{(L_{\mathrm{out}})_{e_i}\}$. We have $g(v_i)\subseteq \Sigma$.  	Later, we will use the value of $g(v_i)$ as the \etype{} that we associate to the edge $e_i$.
	\begin{observation}
		\label{obs:computingCompressType}
		The values of $g(v_i)$, for $i\in\{1, 2\}$, can be computed given all the information associated to the path and its incident edges, that is, the input provided for all edges incident to its nodes, and the \etype s of the incoming edges.
	\end{observation}

\end{itemize}

While the definitions of $g(v)\subseteq \Sigma$ and $g(v_i)\subseteq \Sigma$ have a technical appearance and are crucial to our \CONGEST implementation, their intuition is easier to grasp. In the first case we have a node with possibly many incoming edges on which a \etype{} is already assigned, and $g(v)$ denotes the \etype{} obtained for its outgoing edge. This is essentially the same operation performed on each node of the tree in the diameter algorithm of \Cref{thm:treediam}. In fact, the layers of a rooted tree can be seen as \emph{rake} layers of a tree decomposition.
The second case is what allows us to handle \emph{compress} layers. In this case, we have a path of constant length with many incoming edges attached to them on which a \etype{} is already assigned, and we use some given function $f_{\Pi,\gamma}$ to compute the \etype{} of the two outgoing edges as a function of the incoming \etype s.

In short, following the same scheme used in \cite{Chang20}, we use the generalized \rake{} \& \compress{} algorithm (\Cref{ssec:treeDecomp}) to decompose the graph into few layers consisting of isolated nodes (case 1) and short paths (case 2). Then, we begin to propagate \etype s (that are sets of labels) from the bottom layer to the top layer. Inside the paths, we use a given function $f_{\Pi,\gamma}$ to map a maximal class computed from the incoming requirements into an independent class. From this independent class we can compute the \emph{\etype} of each edge going to higher layers, that is, we can associate a set of labels to each outgoing edge, such that for any choice over these sets, the solution inside the path can be completed. Once we are at the top, that is, once we reach a layer with no outgoing edges, we select a labeling from the class, propagate labels down and complete the labeling in each layer until we reach the bottom of the tree. The definition of independent class guarantees that we do not have any issue when we complete a solution inside the paths, no matter what are the two choices obtained for the outgoing edges from the layers above. Also, note that the function $f_{\Pi,\gamma}$ cannot be an arbitrary function mapping classes to independent classes (since it restricts the possible choices, such a function may give no choice for the layers above), but, as shown in \cite{Chang20}, a \emph{good} function $f_{\Pi,\gamma}$ always exists, and can be computed directly from the description of $\Pi$ (see \Cref{lem:ChangFunction}).

We next present the formal process for \Cref{thm:gapCONGESTrakeCompress-treePolyGap}. The proof and necessary modifications for \Cref{thm:treediam} follow at the end of the section.

\paragraph{\boldmath Algorithm (for all of the complexities above $\Omega(\log n)$).}
The following process depends on a given problem $\Pi$, and some parameters $\gamma$, that is possibly non-constant, and $l=l_{\Pi,\gamma}=O(1)$. Also, this process depends on a given function $f_{\Pi,\gamma}$, that takes in input a class and produces an independent class. In the proofs of \Cref{thm:gapCONGESTrakeCompress-treePolyGap} we choose both parameters such that this process produces a proper solution.

In the algorithm of \Cref{thm:treediam}, we start by rooting the tree, that is equivalent to decomposing the tree into at most $D$ layers, in which each vertex has at most one neighbor in a layer above.
Then, \etype s are propagated from lower layers to higher layers, until the root, a vertex (or layer) with no outgoing edges, picks a label and final labels are propagated downwards in decreasing layer order. Similar, in the following algorithm, we first decompose the tree into layers, then we propagate \etype s from lower layers to higher layers, we fix arbitrary labels when we reach vertices that are not adjacent to a higher layer, and then we propagate final labels down going through the layers in the opposite direction. The actual implementation does not have to be synchronized among all vertices in a layer. Instead, vertices can start their propagation once all of their respective neighbors have provided the necessary information.

\begin{enumerate}
	\item \textbf{Generalized tree decomposition: }Apply the generalized tree decomposition algorithm with parameters $(\gamma,l)$ to decompose the graph into $L$ layers $V_1^R = (V^R_{1,1},\ldots,V^R_{1,\gamma}), \ldots, V_{k}^R  = (V^R_{k,1},\ldots,V^R_{k,\gamma})$, $V_1^C, \ldots, V_{k-1}^C$. Note that for $\gamma = n^{1/k}$ and $l = O(1)$, $L =  2k-1$, while for $\gamma = 1$ and $l = O(1)$, $L = O(\log n)$.
	\begin{itemize}
		\item Each rake sublayer $V^R_{i,j}$ has at most one outgoing edge to a higher sublayer, and has an arbitrary amount of incoming edges from lower sublayers.
		\item Each compress layer is a path of length $l\leq x \leq 2\cdot l$, only the endpoints have one outgoing edge to higher sublayers, and each node can have an arbitrary (constant) number of incoming edges from lower sublayers.
	\end{itemize}
	\item \textbf{Propagate \etype s up} (intuitively, iterate through the layers and sublayers in increasing order)
	\begin{itemize}
		\item rake sublayers (nodes with an outgoing edge). Each connected component of a rake sublayer is composed of isolated nodes. Let $v$ be such a node,  let $\{v,w\}$ be its single outgoing edge, and let $Z$ be the union of the neighbors of $v$ in lower sublayers. Node $v$ waits until each node $u\in Z$ has determined $g(u)\subseteq \Sigma$, then it gathers $g(u)$ for all $u\in Z$. Then, $v$ computes $g(v)$ and sends it to $v$. This is possible due to \Cref{obs:computingRakeType}.
		This step takes $O(1)$ rounds per sublayer.
		\item Compress layers. Each such a layer forms a path. Let $v_1,\ldots,v_x$, $x\in [l,2l]$ be one arbitrary such path. Let $Z_i$ be the set of neighbors of $v_i$ in lower sublayers. Nodes in the path wait until, for all $i \in [x]$, all nodes $u \in Z_i$ have determined $g(u)$. Then, let $v_1$ gather all the information associated to the path and its incident edges, that is, the input provided for all edges incident to its nodes, and the \etype{} of the incoming edges, that is, $g(u)$ for all $u\in Z_i$, for all $i\in [x]$.
		From this information, due to \Cref{obs:computingCompressType}, $v_1$ is able to compute $g(v_1)$ and $g(v_x)$. Finally, $g(v_1)$ and $g(v_x)$ is broadcast to all the nodes of the path. The \etype{} of the edge outgoing from $v_1$ (resp. $v_x$) becomes $g(v_1)$ (resp. $g(v_x)$), and these \etype s are sent to neighbors of higher layers.
		 This takes $O(1)$ rounds per layer.
	\end{itemize}
	\item \textbf{Rake sublayers (nodes with no outgoing edges).} A special case of the previous one is given by rake sublayers with no neighbors in higher sublayers. In this case, after computing its class $B_v$, node $v$ selects an arbitrary element $((),L_{\mathrm{in}},L_H)$ from its class $B_v$, assigns to its incident edges the labeling given by $L_{\mathrm{in}}$ and $L_H$, and sends this choice to its incoming neighbors. This takes $O(1)$ rounds.
	\item \textbf{Propagate final labels down:} (intuitively, iterate through the layers and sublayers in decreasing order).
	\begin{itemize}
		\item Rake sublayers. Each node $v$, as soon as its single parent $u$ has committed to a final label for the edge $\{u,v\}$, picks a consistent labeling from its class $B_v$ for its remaining edges $F^v_{\mathrm{in}}$, and propagates its choice to its incoming neighbors.
		This takes $O(1)$ rounds per sublayer.
		\item Compress layers. Let $v_1,\ldots,v_x$, $x\in [l,2l]$ be a compress layer.
		As soon as the outgoing edges of $v_1$ and $v_x$ have received a label, nodes can complete the labeling for the edges of the path and all incoming edges by picking an arbitrary labeling, among those consistent with the assignment of the outgoing edges, from the independent class $B'$ of the path previously computed during the second phase. This takes $O(1)$ rounds per layer.
	\end{itemize}
\end{enumerate}

Chang in \cite{CP19,Chang20} characterized the set of parameters for which the above process works. In particular, it is shown that for all $\Pi$ and a suitable choice of parameters $\gamma$, $l$ and $f_{\Pi,\gamma}$, it cannot happen that $B_v$ becomes empty during step 3. The algorithms used in \cite{CP19,Chang20} are not exactly the same as the one that we described: instead of propagating \etype s to the layers above, they propagate virtual trees. While it makes no difference for the algorithms of \cite{CP19,Chang20}, since they are designed for the \LOCAL model, we claim that these algorithms can actually use less information, and provide a proof for completeness.

\begin{lemma}[\mbox{\cite{CP19,Chang20,Chang21}}]
	\label{lem:ChangFunction}
	Given a parameter $\gamma\in \{1\}\cup \{n^{1/k} \mid k\in \NN \}$ and an \LCL{} problem $\Pi$ solvable in the \LOCAL model in $O(\max\{\gamma,\log n\})$ rounds, there is a function $f_{\Pi,\gamma}$ and a constant $l_{\Pi,\gamma}$ such that the above process (that depends on $\Pi$, $\gamma$, $f_{\Pi,\gamma}$ and $l_{\Pi,\gamma}$) works,  that is, it never computes an empty class, and in particular in step 3 each vertex $v$ can select an element from its class $B_v$. If such a function $f_{\Pi,\gamma}$ and a constant $l_{\Pi,\gamma}$ exists, then it can also be computed given $\gamma$ and the description of $\Pi$. Moreover, given an arbitrary problem $\Pi$ and $\gamma$, it is decidable whether working $f_{\Pi,\gamma}$ and $l_{\Pi,\gamma}$ exist.
\end{lemma}

\begin{proof}
 	We show that the function $f$ described in \cite{CP19,Chang20} can be used to implement our function $f_{\Pi,\gamma}$. First of all, \cite{CP19,Chang20} consider \LCL{}s in their general form, while we consider \LCL{}s in the black and white formalism. Since the black and white formalism with black degree two is a special case of the general case we can also apply their results in this special setting. The advantage of the black and white formalism is that, in many cases, it reduces the distance of dependencies between labels: if two subtrees $T_1$ and $T_2$ are connected through an edge, and the labeling on that edge is fixed, then the labeling in $T_1$ and $T_2$ can be completed independently. This is the intuitive reason why we can propagate \etype s. In contrast , \cite{CP19,Chang20} propagates much larger structures. Still, via \Cref{claim:LCLequivalence} all our results in the black and white formalism apply in the standard \LCL{} setting.  Hence, our setting can be seen as a special and simplified case of the setting studied in \cite{CP19,Chang20}. 
 	
	We first describe how the definitions of \emph{type} and \emph{class} in \cite{CP19,Chang20} relate to our definitions of a maximal class and \etype s.
	\begin{itemize}
		\item \textbf{Type:} 
		given a subgraph $H$ with two outgoing edges, a \emph{type} describes how a labeling can be fixed in the $r$-radius neighborhood within $H$ (for some constant $r$ large enough to stop any kind of dependency between the labeling outside $H$ and inside $H$) of the two outgoing edges, such that the labeling of $H$ and lower layers connected to it can be completed. Types are actually defined as an equivalence class, where two subgraphs $H_1$ and $H_2$ are in the same equivalence class if the possible valid labelings in the neighborhood of the outgoing edges of $H_1$ are exactly the same as the ones for $H_2$, where a labeling is valid if it can be completed on the rest of the subgraph and in lower layers. The notion of type corresponds to our notion of \emph{maximal class}, in the case of short paths. In the black-white formalism, in order to stop any kind of dependency between the labeling outside $H$ and inside $H$, it is not necessary to fix the labeling in a whole $r$-radius neighborhood, but it is enough to fix a labeling on a single edge. Hence, two subgraphs have the same type if and only if their maximal classes allow the same label assignments for the outgoing edges.
		 
		\item \textbf{Class:} given a subgraph $H$ with one outgoing edge, a \emph{class} describes how a labeling can be fixed in the $r$-radius neighborhood within $H$ of the outgoing edge, such that $H$ and lower layers connected to it can be completed. Similar to the case of types, classes are actually defined as equivalence classes. The notion of a class corresponds to our notions of  \emph{maximal class} and \emph{\etype}, in the case of isolated nodes. In particular, the \etype{} of the outgoing edge of a subgraph is the same as the class of that subgraph, while a maximal class can be seen as a \etype{} augmented with additional information for the nodes on how to complete a labeling.
	\end{itemize}
	More formally, we claim that types and classes of \cite{CP19,Chang20} are equivalent to the set of valid labelings allowed by a maximal class (note that, in the case of a single outgoing edge, we defined this object as \etype{}), assuming that the \etype s assigned to the incoming edges are maximal, that is, if it is not possible to assign a label not contained in the \etype{} and complete the labeling in the layers below. 
	In fact, types and classes of \cite{CP19,Chang20} are defined as follows. Suppose we are given two partially labeled subgraphs $H$ and $H'$, connected to the rest of the graph $G$ through a single edge in the case of classes, and two edges in the case of types. Suppose that $G \setminus H$ is completely labeled, including the edge between $G$ and $H$. Let $G'$ be the graph obtained by replacing $H$ with $H'$. $H$ and $H'$ are of the same type (or class) when the partial labeling of $G$ can be completed into a proper complete labeling if and only if the same is possible in $G'$. Clearly, this definition matches our definition of maximal class, assuming that the \etype s assigned to the incoming edges are maximal.

	We now summarize how the algorithms of \cite{CP19,Chang20} work. They work similarly as our algorithms, but with a notable exception. Instead of propagating types and classes to the layers above, they propagate virtual trees. In particular, every compress layer is handled as follows:
	\begin{enumerate}
		\item Replace the path with a canonical path that preserves its type (with a representative from the equivalence class).
		\item Use the function $f$ to fix the labeling in the middle of the path. The length of the labeled region is chosen such that, after fixing the labeling, there is no dependency between the labeling of the two sides. \label{chang:fixmiddle}
		\item Replace the path with a much longer path that still preserves its type.
		\item Replace the path with two copies of it, such that each endpoint of the original path is connected to only one (different) copy.
	\end{enumerate}
	After this process, each endpoint of a compress layer becomes connected to a single virtual tree (that is not connected to the other endpoint due to step 2 and 4), and this whole tree is propagated to higher layers. 
	After having handled all layers, including propagation through rake layers, we reach vertices with no neighbors in layers above. Let $W$ denote these vertices. Since every compress path has been replaced by two independent trees every node $v\in W$ only has virtual trees connected to it that are not connected to other nodes in $W$. There is one virtual tree at each neighbor of $v$ with an outgoing edge to $v$. Then, nodes in $W$ choose a labeling solely as a function of the class of the virtual trees rooted at them, and this choice is propagated down.
	
	We claim that it is not needed to propagate whole virtual trees, but only their classes, which in our notation corresponds to \etype s. Intuitively, the application of the function $f$ in the middle of the path does nothing else than stopping any possible dependency between the two endpoints of the path, and in the black-white formalism, in order to achieve this goal, it is enough to fix the labeling of a single edge. More in detail, let $v_1,\ldots, u,w,\ldots,v_x$ be the path on which $f$ has been applied to fix the label $\ell$ of its middle edge $e=\{u,w\}$.
	We fix the \etype{} of $e$ to be $\{\ell\}$. Then, we orient the edges of $v_1,\ldots, u$ towards $v_1$, and the edges of $w,\ldots,v_x$ towards $v_x$. Now, nodes in each path have only one outgoing edge, and can hence behave as \emph{rake} nodes. In this way, we obtain a set of possible labelings for each outgoing edge of the compress layer, such that for any choice of labels over the two sets, the labeling inside the path can be completed. Note also that the obtained set is the \emph{maximal} set satisfying this condition, for the same reason explained in the proof of \Cref{thm:treediam}. Since also normal rake layers propagate maximal sets, we obtain that our notion of \etype s and maximal classes is indeed equivalent to the notion of classes and types of \cite{CP19,Chang20}. Hence, since all operations done on the compress layers by \cite{CP19,Chang20} preserve the equivalence class of their types, we obtain that the \etype s computed for the two outgoing edges are exactly the same sets of labels that are allowed for the edges outgoing from the two virtual trees computed by \cite{CP19,Chang20}.
\end{proof}

\begin{example}
	Consider the $2$-coloring problem, that is a problem that requires $\Omega(n)$ rounds. We explain why there cannot be a function $f_{\Pi,\gamma}$ for $\gamma = \sqrt{n}$. In the above process we would have three layers of nodes, $V^R_1, V^C_1, V^R_2$. We would start propagating classes in the first rake layer $V^R_1$, and the classes of all such nodes would contain both allowed colors, since by fixing the color of a node we can always complete a $2$-coloring correctly towards the leaves. Then, we would compute classes in the compress layer $V^C_1$, and then convert these classes into independent classes. This would force us to assign a specific color to each endpoint of each path, since otherwise it would not hold that for any choice of output for the endpoints the solution inside the path is completable. After this, we would propagate classes in the rake layer $V^R_2$. Consider now a rake node in $V^R_2$ with two different neighbors in $V^C_1$ of opposite classes, and note that this situation may arise for any possible way of converting classes into independent classes. This node would obtain an empty class. 
\end{example}
\begin{proof}[Proof of \Cref{thm:gapCONGESTrakeCompress-treePolyGap}]
	In the \LOCAL model, it is already known that the gaps and the decidability results mentioned in the theorem statements hold \cite{Chang20}. Hence, it is enough to prove that the \LOCAL model does not add any power compared to the \CONGEST model. Assume we are given an \LCL{} problem $\Pi$ and a constant $\gamma \in \{1\}\cup\{n^{1/k}\mid k\in\NN\}$ such that there is an $O(\max\{\gamma,\log n\})$ rounds \LOCAL algorithm for $\Pi$.
	Due to \Cref{lem:ChangFunction}, the function $f_{\Pi,\gamma}$ and the constant $l = l_{\mathrm{\Pi,\gamma}}$ needed for the above process exist. In particular, all computed classes are non-empty and each vertex in step 3 can select a label. We now prove the necessary time and bandwidth bounds for the algorithm.
	
	All steps of the procedure can be implemented in the \CONGEST model in the claimed number of rounds. In fact, the \rake{} \& \compress{} procedure can be executed in the \CONGEST model due to \Cref{lem:rakeCompress}. Step 3 and Step 4 only require each node to send one label from $\Sigma$ on each incident edge. Step 2, for rake sublayers, only requires each node $u$ to send one label from $g(u) \subseteq \Sigma$. Step 2, for compress layers, only requires $O(1)$ rounds, since the information required to apply the function $f_{\Pi,\gamma}$ in a path $v_1,\ldots,v_{x}$ only depends on the values $g(u)\subseteq \Sigma$, $u\in Z_i$, and the inputs given to all edges incident to nodes of the path. Since the length of the path is constant, and since there are only a constant number of possible values for $g(u)$, then all this information can be gathered $O(1)$ rounds. All of these steps work with the claimed conditions on the ID space of size $\mathcal{S}$ or a distance-$2l$ input $S$-coloring.
	
	Hence, since each step requires $O(1)$ rounds, we only need to bound the total number of propagation steps that the algorithm needs to perform. Note that a sublayer $V$ is ready to compute a \etype{} for its outgoing edges when all incoming edges from lower layers already computed a \etype{}. Thus, the total number of propagation steps is upper bounded by the total number of sublayers. If $\gamma=1$, the tree decomposition algorithm guarantees that the total number of layers is $L=O(\log n)$. Otherwise, for $\gamma = n^{1/k}$, the total number of layers is $L=2k-1$. The total number of sublayers is $O(L \gamma)$, implying that for $\gamma=1$ we obtain an $O(\log n)$ \CONGEST algorithm, while for $\gamma = n^{1/k}$ and constant $k$ we obtain an $O(n^{1/k})$ \CONGEST algorithm.
\end{proof}

\section{Trees: \texorpdfstring{\boldmath $o(\log n)$}{o(log n)} Randomized Implies \texorpdfstring{\boldmath $O(\log\log n)$}{O(log log n)}}
\label{sec:treeSublog}
In this section we show that, on trees, any randomized algorithm solving an \LCL{} problem $\Pi$ in $o(\log n)$ rounds can be transformed into a randomized algorithm that solves $\Pi$ in $O(\log\log n)$ rounds. This implies that in the \CONGEST model there is no \LCL{} problem in trees with a randomized complexity that lies between $\omega(\log \log n)$ and $o(\log n)$. Moreover, we show that it is not necessary to start from an algorithm for the \CONGEST model, but that a LOCAL model one is sufficient. More formally, we will prove the following theorem.
In this section we show that, on trees, any randomized algorithm solving an \LCL{} problem $\Pi$ in $o(\log n)$ rounds can be transformed into a randomized algorithm that solves $\Pi$ in $O(\log\log n)$ rounds. This implies that in the \CONGEST model there is no \LCL{} problem in trees with a randomized complexity that lies between $\omega(\log \log n)$ and $o(\log n)$. Moreover, we show that it is not necessary to start from an algorithm for the \CONGEST model, but that a LOCAL model one is sufficient. More formally, we will prove the following theorem.

\begin{theorem}[sublogarithmic gap]
	\label{thm:treesrandSublogGap}
	Let $c\geq 1$ be a constant. Given any \LCL{} problem $\Pi$, if there exists a randomized algorithm for the \LOCAL model that solves $\Pi$ on trees in $o(\log n)$ rounds with failure probability at most $1/n$, then there exists a randomized algorithm for the \CONGEST model  that solves $\Pi$ on trees in $O(\log \log n)$ rounds with failure probability at most $1/n^c$.
\end{theorem}
In \Cref{ssec:highLevelsubloggap} we present the high level idea of the proof of \Cref{thm:treesrandSublogGap}. The proof itself is split into three subsections, for which a road map appears at the end of \Cref{ssec:highLevelsubloggap}.
\subsection{Proof Idea for Theorem~\ref{thm:treesrandSublogGap}}
\label{ssec:highLevelsubloggap}

In the \LOCAL model it is known that, on trees, there are no \LCL{} problems with randomized complexity between $\omega(\log \log n)$ and $o(\log n)$ \cite{CP19,CHLPU18}. At a high level, we follow a similar approach in our proof. However, while some parts of the proof directly work in the \CONGEST model, there are some challenges that need to be tackled in order to obtain an algorithm that runs in $O(\log \log n)$ that is actually bandwidth efficient. We now provide the high level idea of our approach.

\paragraph{The standard approach: expressing the problem as an LLL instance.} As in the \LOCAL model case, the high level idea is to prove that if a randomized algorithm for an \LCL{} problem $\Pi$ runs in $o(\log n)$ rounds, then we can make it run faster at the cost of increasing its failure probability. In this way, we can obtain a constant time algorithm $\mathcal{A}_0$ at the cost of a very large failure probability. This partially gives what we want: we need a fast algorithm with small failure probability, and now we have a very fast algorithm with large failure probability. One way to fix the failure probability issue is to derandomize the algorithm $\mathcal{A}_0$, that is, to find a random bit assignment satisfying that if we run the algorithm with this specific assignment of random bits then the algorithm does not fail. Ironically, we use a randomized algorithm to find such a random bit assignment. 

\begin{lemma}[informal version of \Cref{lem:fakeN}]
	For any problem $\Pi$ solvable in $o(\log n)$ rounds with a randomized \LOCAL algorithm having failure probability at most $1/n$, there exists a constant time \LOCAL algorithm $\mathcal{A}_0$ that solves $\Pi$ with constant local failure probability $p$.
\end{lemma}

It turns out that the problem itself of finding a good assignment of random bits such that the constant time algorithm $\mathcal{A}_0$ does not fail can be formulated as a \lovasz Local Lemma  (LLL) instance.
In an LLL instance there are random variables and a set of \emph{bad events} that depend on these variables. The famous \lovasz Local Lemma \cite{LLL73} states that if the probability of each bad event is upper bounded by $p$, each bad event only shares variables with $d$ other events and the \emph{LLL criterion} $epd<1$ holds, then there exists an assignment to the variables that avoids all bad events (a more formal treatment of the \lovasz Local Lemma follows at the beginning of \Cref{ssec:solvingPiPrime}). In our setting, the random variables are given by the random bits used by the vertices and each vertex $v$ has a bad event  $\mathcal{E}_v$ that holds if the random bits are such that $v$'s constraints in $\Pi$ are violated if $\mathcal{A}_0$ is executed with these random bits. We show that a large polynomial LLL criterion---think of $p(ed)^{30}<1$--- holds. Thus, the \lovasz Local Lemma implies that there exist \emph{good random bits} such the \LCL{} constraints of $\Pi$ are satisfied for all nodes when using these bits in $\mathcal{A}_0$. In the \LOCAL model it is known how to solve an LLL problem with such a strong LLL criterion efficiently. We show that the same holds in the \CONGEST model, i.e.,  $O(\log \log n)$ \CONGEST rounds are sufficient to find a good assignment of random bits. We point out that we do not give a general LLL algorithm in the \CONGEST model but an algorithm that is tailored for the specific instances that we obtain. The constant time algorithm $\mathcal{A}_0$ executed with these random bits does not fail at any node. 

 We summarize the high level approach as follows: Given an \LCL{} problem $\Pi$ defined on trees and an $o(\log n)$-rounds randomized algorithm $\mathcal{A}$ for $\Pi$, we obtain a constant time algorithm $\mathcal{A}_0$ for $\Pi$ and a new problem $\Pi'$ of finding good random bits for $\mathcal{A}_0$. Problem $\Pi'$ is defined on the same graph as problem $\Pi$. 
 The algorithm $\mathcal{A}_0$ and the problem $\Pi'$ only depend  on $\Pi$ and $\mathcal{A}$. We show that $\Pi'$ is also an \LCL{} problem. In problem $\Pi'$ each node of the tree needs to output a bit string such that  if the constant time algorithm $\mathcal{A}_0$ is run with the computed random bits, the problem $\Pi$ is solved. We will show that $\Pi'$ can be solved in $O(\log \log n)$ rounds, and note that once $\Pi'$ is solved, one can run $\mathcal{A}(n_0)$ for $t_0=t_{\mathcal{A}_0}=O(1)$ rounds to solve $\Pi$.

\begin{lemma}[informal version of \Cref{lem:PiPrime}]
	The problem $\Pi'$, that is, the problem of finding a good assignment of random bits that allows us to solve $\Pi$ in constant time, can be solved in  $O(\log \log n)$ rounds in the \CONGEST model.
\end{lemma}

Problem $\Pi'$ is defined on the same tree as $\Pi$ but problem $\Pi'$ has checking radius $r+t_0$. Thus, the \emph{dependency graph} of the LLL instance is a power graph of the tree, or in other words the LLL instance is \emph{tree structured}. Hence, in the \LOCAL model, $\Pi'$ can be solved in $O(\log \log n)$ rounds by using a  $O(\log\log n)$ randomized \LOCAL algorithm for tree structured LLL instances \cite{CHLPU18}. We cannot do the same here, as it is not immediate whether this algorithm works in the \CONGEST model.

\paragraph{The shattering framework.}
Our main contribution is showing how to solve $\Pi'$ in $O(\log\log n)$ rounds in a bandwidth-efficient manner.
To design an $O(\log\log n)$-round algorithm for $\Pi'$, we apply the shattering framework for LLL of~\cite{FGLLL17}, that works as follows. After a precomputation phase of $O(\logstar n)$ rounds, the shattering process uses $\poly\Delta=O(1)$ rounds to determine the random bits of some of the nodes. The crucial property is that all nodes with unset random bits form \emph{small connected components} of size $N=\poly(\Delta)\cdot \log n=O(\log n)$; in fact, even all nodes that are close to nodes with unset random bits form small connected components $C_1,\ldots,C_k$. Furthermore, each connected component can be solved (independently) with an LLL procedure as well, with a slightly worse polynomial criterion (e.g., $p(ed)^{15}<1$).
Note that, in order to solve these smaller instances, we need to use a deterministic algorithm.This is because, if we try to recursively apply a randomized LLL algorithm on the smaller instances (e.g.\ by using the randomized \LOCAL algorithm of~\cite{CPS17}) we get that each component can be solved independently in $O(\log N)$ rounds, but with failure probability $1/\poly N\gg 1/n$.

\paragraph{Our main contribution: solving the small instances in a bandwidth efficient manner.}
Since it seems that we cannot directly use an LLL algorithm to solve the small remaining instances  $C_1,\ldots,C_k$, we follow a different route. We devise a deterministic \CONGEST algorithm that we can apply on each of the components in parallel: In \Cref{thm:gapCONGESTrakeCompress-treePolyGap} we prove that, on trees, any randomized algorithm running in $n^{o(1)}$ rounds (subpolynomial in the number of nodes) in the \LOCAL model can be converted into a deterministic algorithm running in $O(\log n)$ in the \CONGEST model.  We  use this result here, to show that, the mere existence of the randomized \LOCAL algorithm of~\cite{CPS17}, that fails with probability at most $1 / \poly N$ and runs in $O(\log N)$ rounds in the \LOCAL model, which fits the runtime requirement of \Cref{thm:gapCONGESTrakeCompress-treePolyGap}, implies that $\Pi'$ can be solved in $O(\log N) = O(\log \log n)$ \CONGEST rounds deterministically on the components induced by unset bits.  To apply \Cref{thm:gapCONGESTrakeCompress-treePolyGap} that only holds for \LCL{} problems, we  express the problem of completing the partial random bit assignment as  a problem $\Pi''$, that intuitively is almost the same problem as $\Pi'$, but allows some nodes to already receive bit strings as their input. We show that $\Pi''$ is a proper \LCL{}.

Formally, there are  several technicalities that we need to take care of.  In particular, we do not want to provide any promises on the inputs of $\Pi''$, as  \Cref{thm:gapCONGESTrakeCompress-treePolyGap} does not hold for \LCL{}s with promises on the input. For example, we cannot guarantee in the \LCL{} definition that the provided input, that is, the partial assignment of random bits, can actually be completed into a full assignment that is good for solving $\Pi$. On the other hand, if we just defined $\Pi''$ as the problem of completing a partial given bit string assignment, it might be unsolvable for some given inputs, and this would imply that an $n^{o(1)}$ time algorithm for this problem cannot exist to begin with, thus there would be no way to use \Cref{thm:gapCONGESTrakeCompress-treePolyGap}.

In order to solve this issue, we define $\Pi''$ such that it can be solved fast even if the input is not \emph{nice} (for some technical definition of nice). In particular, we make sure, in the definition of $\Pi''$, that if the input is nice then the only way to solve $\Pi''$ is to actually complete the partial assignment, while if the input is not nice, and only in this case, nodes are allowed to output \emph{wildcards} \wild; the constraints of nodes that see wildcards in their checkability radius are automatically satisfied. This way the problem is always solvable.
For an efficient algorithm, we make sure that nodes can verify in constant time if a given input assignment is nice or not.  We also show that inputs produced for $\Pi''$ in the shattering framework are always nice. The definition of $\Pi''$ and the provided partial assignment allows us to split the instance of $\Pi'$ into many independent instances of $\Pi''$ of size $N = O(\log n)$. By applying \Cref{thm:gapCONGESTrakeCompress-treePolyGap} we get that \cite{CPS17} implies the existence of a deterministic \CONGEST algorithm $\mathcal{B}$ for $\Pi''$ with complexity $O(\log N) = O(\log \log n)$.

\begin{lemma}[informal version of \Cref{lem:Pi3Congest}]
	There is a deterministic \CONGEST algorithm to solve $\Pi''$ on any tree with at most $N$ nodes in $O(\log N)$ rounds, regardless of the predetermined  input.
\end{lemma}

We apply algorithm $\mathcal{B}$ on each of the components  $C_1,\ldots,C_k$ in parallel, and the solution of $\Pi''$ on each small components together with the random bit strings from the shattering phase yield a solution for  $\Pi'$. This can then be transformed into a solution for $\Pi$ on the whole tree by running the constant time algorithm $\mathcal{A}_0$ with the computed random bits, which completes our task and proves \Cref{thm:treesrandSublogGap}.

\paragraph{Road map through the section.} We follow the same top-down approach that we used to explain the high level idea. In \Cref{ssec:thm:treesrandSublogGap} we define the problem $\Pi'$ and prove \Cref{thm:treesrandSublogGap}, assuming that we already know how to solve problem $\Pi'$ efficiently (\Cref{lem:PiPrime}). Then, in \Cref{ssec:solvingPiPrime} we define $\Pi''$ and prove \Cref{lem:PiPrime}, assuming that we know how to solve $\Pi''$ efficiently (\Cref{lem:Pi3Congest}). Lastly, in \Cref{ssec:solvingPi3Prime} we prove \Cref{lem:Pi3Congest}.

\subsection{Proof of Theorem~\ref{thm:treesrandSublogGap}: Solving \texorpdfstring{\boldmath $\Pi$}{\textPi}}
\label{ssec:thm:treesrandSublogGap}
We now start by proving that, if we are given an algorithm $\mathcal{A}$ running in $o(\log n)$ rounds, then we can make it faster at the cost of substituting its global  failure probability of $1/n$ with a much larger (even constant) local failure probability. Afterwards, in the rest of the section, we find a random bit assignment such that if we run the latter algorithm with this specific assignment of random bits then it does not fail at any node.  Recall, that $\mathcal{A}(n_0)$ denotes algorithm $\mathcal{A}$ run with parameter $n_0$ as an upper bound on the nodes (see \Cref{sec:definitions} for the formal definition of $\mathcal{A}(n_0)$).
The next statement was first observed in \cite{CP19} and has been used in many results afterwards. We prove it for completeness as it only occurs in the internals of proofs in these works and has not been stated in such an explicit form.
\begin{lemma}\label{lem:fakeN}
	Let $c$ be a constant $\ge 1$ and $\Pi$ an \LCL{} with maximum degree $\Delta$ and with checking radius $r$.
	Assume that we are given a randomized \LOCAL algorithm $\mathcal{A}$ solving $\Pi$ in $f(\Delta) + \epsilon \log_\Delta n$ rounds for some function $f$ and $\epsilon < 1$ with failure probability at most $1/n^c$ on any graph with at most $n$ nodes.
	There exists a constant $\hat{n}$, such that, for any $n_0 \ge \hat{n}$, if we run $\mathcal{A}(n_0)$ on a graph $G$ of size at most $n$, for each node $v$ of $G$ it holds that the outputs produced by $\mathcal{A}(n_0)$ on the $r$-radius neighborhood of $v$ is incorrect, according to the \LCL{} constraints of $\Pi$, with probability at most $1/n_0^c$.
	\end{lemma}
The local error probability of the constant time algorithm $\mathcal{A}(n_0)$ in \Cref{lem:fakeN} is $1/n_0^c$. 
\begin{proof}
	Assume that there is an algorithm $\mathcal{A}$ that takes an input $n$, an upper bound on the size of the graph, runs in $t_{\mathcal{A}}(n)$ on graphs with no IDs, and produces an output with local failure probability at most $1/n^c$, meaning that for each node $v$ of the graph, with probability at least $1-1/n^c$, the neighborhood of radius $r$ of $v$ is contained in $C$. 
	Let $\hat{n}$ be such that
	\begin{align}
		\label{eqn:IamNZero}
		\Delta^{1+t_{\mathcal{A}}(\hat{n}) + r}\leq \hat{n}
	\end{align}
	holds. Such a choice of $\hat{n}$ exists because due to $t_{\mathcal{A}}(n) = f(\Delta) + \epsilon \log_\Delta n$, we have $\Delta^{1+t_{\mathcal{A}}(n) + r}\leq \Delta^{1+f(\Delta) + \epsilon \log_\Delta n + r}=\Delta^{1+f(\Delta)+r} \cdot n^{\epsilon}=o(n)$. Fix an arbitrary $n_0\geq \hat{n}$ and note that Inequality~\ref{eqn:IamNZero} also holds for  $n_0$. 	By the definition of $\mathcal{A}$, we obtain the following fact.
	
	\smallskip
	\textbf{Fact:} Algorithm $\mathcal{A}(n_0)$ produces a feasible solution (according to $\Pi$) with probability $1/n_0^c$ on any graph with at most $n_0$ nodes.
	\smallskip
	
	To prove the statement of the lemma we show that for any $n$, the constant time algorithm $\mathcal{A}(n_0)$ produces a feasible solution (according to $\Pi$) with probability $1/n_0^c$ on any graph $G$ with at most $n$ vertices. Assume for a contradiction that it does not, i.e., that there is some $n$, a graph $G=(V,E)$ on at most $n$ nodes and a node $v\in V(G)$ such that the probability that the \LCL{} constraint of $v$ in its $r$-radius neighborhood is violated is strictly larger than $1/n_0^c$.
	Let $H_v=(V^H, E^H)$ be the graph consisting of the radius-$(t_{\mathcal{A}}(n_0)+r)$ neighborhood of $v$ in $G$ and let $\phi: V^H\rightarrow V$ be the map that identifies each vertex of $H_v$ with its counterpart in $G$. An upper bound on the number of nodes in a radius-$(t_{\mathcal{A}}(n_0)+r)$ neighborhood is given by $\Delta^{1+t_{\mathcal{A}}(n_0)+r}$, and thus by Inequality (\ref{eqn:IamNZero}) the graph $H_v$ has at most  $n_0$ vertices.
	Further, each node $v'\in V^H$ in the $r$-radius neighborhood around $\phi^{-1}(v)$ has the same radius-$t_{\mathcal{A}}(n_0)$ view in $H_v$ and as its counterpart $\phi(v')$ in $G$. Thus, when running $\mathcal{A}(n_0)$ on $H_v$ or on $G$, the joint probability distribution of the output labeling of all nodes in $N^r_{H(v)}=N^r_{H_v}(\phi^{-1}(v))$ and $N^r_{G}(v)$ are identical. As $H_v$ is a graph of at most $n_0$ nodes the aforementioned fact implies that the joint output labeling distribution of  $N^r_{H(v)}$  violates the constraints of $\Pi$ with probability at most $1/n_0^c$, a contradiction to the assumption that the same output labeling distribution on $N_G^r(v)$ violates the constraint with probability strictly larger than $1/n_0^c$.
\end{proof}

We now define the \LCL{} problem $\Pi'(\mathcal{A}, n_0)$ that is the central problem of our algorithm.   We only use the definition for a sublogarithmic algorithm $\mathcal{A}$ and in the proof of \Cref{thm:treesrandSublogGap} we determine a value of $n_0$ such that $\Pi'(\mathcal{A}, n_0)$ is  solvable.

\begin{definition}[\LCL{} problem $\Pi'$]\label{def:variousLCLs}
	Let $\Pi$ be an \LCL{} with checking radius $r$, let $\mathcal{A}$ be a randomized algorithm for $\Pi$ and let $n_0\in \NN$ and $0\leq x\leq 1$ be parameters.
Let $t_0=t_{\mathcal{A}}(n_0)$, $h_0=h_{\mathcal{A}}(n_0)$ be upper bounds on the runtime and on the number of random bits used by a node when running $\mathcal{A}$ on graphs with at most $n_0$ nodes.

\textbf{Problem $\Pi'(\mathcal{A}, n_0)$: }
				The input alphabet of  $\Pi'(\mathcal{A},n_0)$  is the same as in $\Pi$. The output of each node is a string of $h_0$ bits. A ball of radius $r + t_0 $ centered at $v$ is feasible if, by running $\mathcal{A}(n_0)$ for $t_0$ rounds on each node in the radius-$r$ neighborhood of $v$ using the bits given by the output of $\Pi'$, we obtain a neighborhood that is correct according to $\Pi$. The checkability radius is $r + t_0$.
\end{definition}

The next lemma is the main technical contribution in this section and proven in \Cref{ssec:solvingPiPrime}.
 \begin{restatable}[Solving $\Pi'$]{lemma}{lemPiPrime}
	\label{lem:PiPrime}
Let $c\geq 1$ be a constant. Let $\Pi$ be an \LCL{} on trees with checking radius $r$ and maximum degree $\Delta$. Let $\mathcal{A}$ be a \LOCAL algorithm and let $n_0$ be such that $\mathcal{A}(n_0)$ solves $\Pi$ with local failure probability at most $1/n_0$ and such that $n_0>(\Delta^{2(t_{\mathcal{A}}(n_0)+r)}\cdot e)^{26+2c}$ holds. Then $\Pi'(\mathcal{A},n_0)$ can be solved in the \CONGEST model in $O(\log \log n)$ rounds with failure probability at most $1/n^c$.
\end{restatable}
Given \Cref{lem:PiPrime} we can immediately prove \Cref{thm:treesrandSublogGap}.

\begin{proof}[Proof of \Cref{thm:treesrandSublogGap}]
Let $\Pi$ be an \LCL{} problem with checkability radius $r$ that is solvable with a randomized algorithm $\mathcal{A}$ in $t(n)=o(\log n)$ with failure probability at most $1/n$, in which each node uses at most $h(n)$ random bits, in any graph with at most $n$ nodes. As $t(n)$ is in $o(\log n)$, the function $g(n)=(\Delta^{2(t(n)+r)}\cdot e)^{(26+2c)}$ is in $o(n)$. Thus, we can fix a constant $n_0^A$ such that $n_0^A>g(n_0^A$) holds.
Let $n_0^B$ be the constant $\hat{n}$ of \Cref{lem:fakeN}. We define $n_0 = \max\{n_0^A,n_0^B\}$. Let $t_0=t(n_0)$. We obtain that
\begin{align}
	n_0>(\Delta^{2(t_0+r)}\cdot e)^{(26+2c)} \label{eqn:nzeroChoice},
\end{align}
and that $\mathcal{A}(n_0)$ fails on each neighborhood with probability at most $1 / n_0$. We obtain constants $n_0$, $t_0$ and $h_0=h(n_0)$ that do not depend on any input graph but solely on the properties of $\mathcal{A}$.
For this choice of $n_0$ we consider the \LCL{} problem $\Pi' = \Pi'(\mathcal{A}, n_0)$ which is formally defined in \Cref{def:variousLCLs}. It is the problem of finding $h_0$ random bits for each node of the input graph such that if we run $\mathcal{A}(n_0)$ (that requires $t_0$ rounds) by using the random bits given by the output for $\Pi'$, we solve $\Pi$. We satisfy all requirements of \Cref{lem:PiPrime} to solve $\Pi'(\mathcal{A},n_0)$ in $O(\log \log n)$ rounds. Given a solution for $\Pi'$, we solve the original problem $\Pi$ on the whole tree as follows. We run the $t_0=O(1)$-round \LOCAL algorithm $\mathcal{A}(n_0)$ in $O(1)$ \CONGEST rounds, using the bits provided by the solution of $\Pi'$ as random bits for $\mathcal{A}(n_0)$. Given the solution to $\Pi'$, this process is deterministic and never fails.
 Thus, the overall randomized algorithm can only fail with probability $1/n^c$, the failure probability of solving $\Pi'$ with  \Cref{lem:PiPrime}.
\end{proof}

\subsection{Proof of Lemma~\ref{lem:PiPrime}: Solving \texorpdfstring{\boldmath $\Pi'$}{\textPi'}}
\label{ssec:solvingPiPrime}

In this section we prove the following lemma.

\lemPiPrime*

We first explain how the problem can be solved in the \LOCAL model and then present our much more involved solution in the \CONGEST model. The proof of \Cref{lem:PiPrime} that puts all arguments of the section together appears at the end of the section. For both approaches we require that $\Pi'$ can be solved with a randomized process that is a \lovasz Local Lemma (LLL) instance and whose dependency graph $H$ is in close relation to the communication network $G$. We next explain the LLL problem and the distributed LLL problem. 

\paragraph{Distributed \lovasz Local Lemma (LLL).}  Consider a set $\mathcal{V}$ of independent random variables, and a family $\mathcal{X}$ of $n$
(bad) events $\mathcal{E}_1, \ldots, \mathcal{E}_n$ on these variables. Each event $\mathcal{E}_i \in  \mathcal{X}$ depends on some subset $vbl(\mathcal{E}_i) \subseteq \mathcal{V}$ of variables. Define the dependency graph $H_{\mathcal{X}} = (\mathcal{X},\{(\mathcal{E},\mathcal{E}') \mid vbl(\mathcal{E}) \cap vbl(\mathcal{E}')\neq \emptyset \})$ that connects any two events
which share at least one variable. Let $d=\Delta_H$ be the maximum degree in this graph, i.e., each event $\mathcal{E}\in\mathcal{X}$ shares variables with at most $\Delta_H$ other events $\mathcal{E}'\in \mathcal{X}$. Finally, define $p=\max_{\mathcal{E}\in\mathcal{X}} \Pr(\mathcal{E})$. The \lovasz Local Lemma \cite{LLL73}  shows that $Pr(\cap_{\mathcal{E}\in \mathcal{X}\mathcal{E}})>0$ if $epd<1$. In the \emph{distributed LLL problem} each bad event and each variable  is simulated by a vertex of the communication network and the objective is to compute an assignment for the variables such that no bad event occurs.

\paragraph{\boldmath \LCL{} Problem $\Pi'$ can be solved by LLL.}
Given an instance of $\Pi'(\mathcal{A},n_0)$ on a tree $G$, let each vertex $v\in V$ select its string $R(v)$ of $h_0=h(n_0)=O(1)$ random bits uniformly at random. Associate a \emph{bad event} $\mathcal{E}_v$ to each vertex $v\in V$. The event $\mathcal{E}_v$ occurs if, by executing $\mathcal{A}(n_0)$ with the random bits assigned by $R$, the constraints of $\Pi$ are violated in the radius-$r$ neighborhood of $v$. The associated dependency graph $H$ has one vertex for each bad event $\mathcal{E}_v$ and there is an edge between two vertices $\mathcal{E}_u$ and $\mathcal{E}_v$ if the two events depend on each other, that is, if $u$ and $v$ are in distance at most $2(t_{\mathcal{A}}(n_0)+ r)$ in $G$.
Note that in general $H$ is not a tree.

\begin{claim}[$\Pi'$ can be solved by an LLL]\label{claim:PiPrimeLLL}
	Under the conditions of \Cref{lem:PiPrime} on $\Pi$, $\mathcal{A}$ and $n_0$, the described random process is an instance of LLL with dependency graph $H$ and LLL criterion $p ( \Delta_H \cdot e)^{26+2c}$. Its solution provides a solution to $\Pi'(\mathcal{A},n_0)$. Further, one round of \CONGEST or \LOCAL communication in $H$ can be simulated in $O(1)$ rounds in the same model.
\end{claim}
\begin{proof}
	Denote $t_0=t_{\mathcal{A}}(n_0)$.
The maximum degree $\Delta_H$ of $H$, i.e., the dependency degree of the LLL, is upper bounded by $\Delta^{2(t_0+r)}$.   Due to the assumption on $n_0$, $t_0$ and the failure probability of $1/n_0$ for $\mathcal{A}(n_0)$ in the statement of \Cref{lem:PiPrime}, the LLL criterion is satisfied with exponent $(26+2c)$, that is,
\begin{align}
	(1/n_0) \cdot (\Delta_H\cdot e)^{26+2c}< 1.
\end{align}

Since $\Delta_H$ is constant, one round of communication of a \LOCAL or \CONGEST algorithm on $H$ can be run in $G$ in $O(1)$ rounds in the same model. 
\end{proof}

\paragraph{\boldmath Solving $\Pi'$ in the \LOCAL model \cite{CP19,CHLPU18}.}
In the \LOCAL model, it would be possible to solve $\Pi'$ by using a standard LLL algorithm \cite{FGLLL17}. This algorithm consists of two phases, called preshattering and postshattering. In the preshattering phase, the algorithm fixes some variables (strings of random bits $R(v), v\in V$) such that all vertices/events of $H$ that have an unset variable/random bit in their vicinity induce small connected components in $H$, of size $N = O(\log n)$. Moreover, this partial assignment satisfies the following  property: completing this partial assignment is also an LLL problem with a slightly worse criterion. Then, to complete the assignment, in the postshattering phase the small components are solved with a deterministic $O(\log N) = O(\log\log n)$ round algorithm for so called tree-structured LLL instances \cite{CHLPU18}.

\paragraph{\boldmath Solving $\Pi'$ in the \CONGEST model.}

We use the rest of the section to devise an efficient algorithm for $\Pi'$ in the \CONGEST model, that is, we prove \Cref{lem:PiPrime}. In order to solve $\Pi'$ in the \CONGEST model, we use the same preshattering algorithm of \cite{FGLLL17}, as it directly works in the \CONGEST model, but we replace the postshattering phase and provide a bandwidth-efficient way to solve the small components. We next explain our whole \CONGEST algorithm for $\Pi'$; the proof of \Cref{lem:PiPrime} follows at the end of the section.

\paragraph{\boldmath Preshattering for $\Pi'$.}  In the following, let $x=\sqrt{1/n_0}$. We sketch the preshattering phase of \cite{FGLLL17} to point out that it can be implemented in \CONGEST. For more details we refer to \cite{FGLLL17}. First, each vertex obtains an ID chosen uniformly at random from the ID space $[n^{c+3}]$. The rest of the proof of \Cref{lem:PiPrime} is conditioned on the event that the chosen IDs are unique which happens with probability at least $1-1/n^{c+1}$.
Then, we compute a distance-$2$ coloring of $H$ with $\Delta_H^2+1$ colors (this can be done in $O(\poly \Delta_H+\logstar (n^{c+3}))=O(\logstar n)$ rounds in CONGEST, e.g., with Linial's algorithm \cite{linial87}). During the execution of the preshattering phase we have \emph{set} and \emph{unset} random bits as well as \emph{frozen} and \emph{unfrozen} random bits. At the beginning all random bits are unset, and no random bit is frozen.
Then, we iterate through the $\Delta_H^2+1=O(1)$ color classes. In iteration $i\in[\Delta_H^2+1]$ all nodes  $v$ with color $i$ in parallel sample all of their not frozen random bits uniformly at random. Then, a node $v$ checks whether the probability of $\mathcal{E}_u$ for any $u\in (\{v\}\cup N_H(v))$ has increased to $\geq x$, if so, the random bits of $v$ are unset and all random bits that $\mathcal{E}_u$ depends on are frozen.
All of these steps can be executed in $\poly\Delta_H=O(1)$ rounds in the \CONGEST model on the communication graph $G$.

\begin{observation}[\mbox{\cite[arxiv Observation 3.4]{FGLLL17}}]
	\label{obs:shatteringProbability}
	After the preshattering process, for each event $\mathcal{E}_v$ the probability of $\mathcal{E}_v$ having at least one unset variable, i.e., an unset random bit string stored at an $H$-neighbor $\mathcal{E}_u$ of $\mathcal{E}_v$, is at most $(\Delta_H + 1) \sqrt{1/n_0}$. Furthermore, this is independent of events that are further than $2$ hops (in $H$) from $\mathcal{E}_v$.
\end{observation}
We call vertices/events $\mathcal{E}_v$ \emph{undecided} if one of the random bit strings that the event depends on is unset, i.e., if one of the random bit strings in $v$'s $t_0+r$ neighborhood in $G$ is unset, all other nodes are called \emph{decided}. Note that, in general, even if a node is undecided, it may be the case that its bit string is already set.

As used in \cite{FGLLL17}, \Cref{obs:shatteringProbability} and \Cref{lem:shattering} (the shattering Lemma from \cite{FGLLL17})  imply that only small (independent) connected components remain undecided, as stated in the next claim. Recall that $c$ is the arbitrary constant in \Cref{thm:treesrandSublogGap} and \Cref{lem:PiPrime}, respectively.
\begin{claim}[\mbox{\cite{FGLLL17}}]
	\label{claim:smallComponents}
	After the preshattering phase the undecided vertices of $H$ induce small components $C^H_1,\ldots, C^H_k$, each of size at most $\poly(\Delta_H)\cdot \log n=O(\log n)$ with probability $1-1/n^{c+1}$.	
\end{claim}
\begin{proof}
	The goal is to apply \Cref{lem:shattering} with $c_1=12+c$, $c_2=2$ and $c_3=c+1$ (satisfying $c_1>c_3+ 4c_2 + 2$) and with the preshattering process where the set $B$ in the lemma statement consists of undecided nodes.
	By \Cref{obs:shatteringProbability} and using $\Delta_H+1\leq e\Delta_H$ for $\Delta_H\geq 2$ we obtain
	\begin{align}
		Pr(v\in B)\leq (\Delta_H+1)\sqrt{1/n_0}\leq  (\Delta_H+1) \cdot (e\Delta_H)^{-(13+c)}\leq (1/\Delta_H)^{12+c}=(1/\Delta_H)^{c_1}~.
	\end{align}
Thus, by \Cref{lem:shattering} the maximal components of undecided nodes are of size $O(\Delta_H^{2c_2}\log n)=O(\Delta_H^{4}\log n)=O(\log n)$ with probability at least $1-1/n^{c+1}$.
\end{proof}
In \Cref{claim:smallComponents} two distinct components $C^H_i$ and $C^H_j$ are not connected by an edge in $H$, that is, they have distance at least $2(t_0+r)+1$ in $G$.
We do not bound the number of components $k$.

\paragraph{Postshattering.}
When dealing with small components we need to deal with vertices that have already decided on their strings of random bits (these were set in the preshattering phase) and vertices that are still undecided about their random bits.
Next, we define the problem that we solve on the small components. The precise instances of the small components are described later. While the definition of $\Pi''(\mathcal{A},n_0,x)$ is more general, we only use it for sublogarithmic algorithm $\mathcal{A}$, the choice of $n_0$ from the statement of \Cref{lem:PiPrime}, and for $x=\sqrt{1/n_0}$.

\begin{definition}[\LCL{} problem $\Pi''$]\label{def:variousLCLs2}
	Let $\Pi$ be an \LCL{} with checking radius $r$, $\mathcal{A}$ a randomized algorithm for $\Pi$ and $n_0\in \NN$ and $0\leq x\leq 1$ parameters.
	Let $t_0=t_{\mathcal{A}}(n_0)$, $h_0=h_{\mathcal{A}}(n_0)$ be upper bounds on the runtime and on the number of random bits used by a node when running $\mathcal{A}$ on graphs with at most $n_0$ nodes.
	
	\textbf{Problem $\Pi''(\mathcal{A}, n_0,x)$:}
		Each node has in input a pair, where the first element is a possible input for $\Pi$, and the second element can either be an empty string, or a bit string of length $h_0$. The output of each node is either a bit string of $h_0$ bits, or a wildcard \wild. A neighborhood of radius $r + t_0 $ centered around some node $v$ is feasible if and only if the following holds. If $v$ outputs \wild, then it must hold that by running $\mathcal{A}(n_0)$ on nodes in the $r$-radius neighborhood of $v$ for $t_0$ rounds we obtain a radius-$r$ neighborhood that with probability strictly greater than $x$ violates the constraints of $\Pi$, where all random bits that are not predetermined by the input are chosen uniformly at random. If $v$ does not output a wildcard, and an input string of $h_0$ bits is provided in input, then the same bit string must be given in output. Also, a neighborhood of a node not having an input string of $h_0$ bits is always valid if there is a node different from $v$ in the neighborhood that is outputting a wildcard. Finally, if no node in a neighborhood is outputting a wildcard, the output is valid if and only if by running $\mathcal{A}(n_0)$ for $t_0$ rounds on each node in the radius-$r$ neighborhood of $v$ using the bits given by the output of $\Pi'$, we obtain a neighborhood that is correct according to $\Pi$. The checkability radius is $t_0 + r $.	
\end{definition}
In problem $\Pi''(\mathcal{A}, n_0, x)$ a vertex needs to output a random bit string or a wildcard \wild. The inputs of vertices and the value $x$ determines which vertices are allowed to output wildcards. The wildcards in the definition of the problem ensure that the  problem is always solvable for a suitable choice of $x$. The following technical statement is standalone but its motivation will only become clear once we use it thereafter to solve problem $\Pi'$. The proof of \Cref{lem:Pi3Congest} is postponed to \Cref{ssec:solvingPi3Prime}.

\begin{restatable}[Solving $\Pi''$]{lemma}{lemPiThreeCongest}
\label{lem:Pi3Congest}
	Let $\Pi$ be an \LCL{} problem with checking radius $r$ and maximum degree $\Delta$ and let $\mathcal{A}$ be a randomized algorithm for $\Pi$.
	Let $n_0\in \NN$ and $0\leq x \leq 1$ be constant values that satisfy $e\cdot x\cdot \Delta^{4(t_{\mathcal{A}}(n_0)+r)}<1$.

	Then there is a deterministic \CONGEST algorithm to solve $\Pi''(\mathcal{A},n_0,x)$ on trees with at most $n$ vertices and given unique IDs (or a large enough $O(1)$-distance coloring) from an ID space $\mathcal{S}$ in  $O(\log n)+O(\logstar |\mathcal{S}|)$ rounds. In particular, if the ID space is at most exponential in $n$, then the runtime is $O(\log n)$.
\end{restatable}

 We show that the components $C^H_1,\ldots, C^H_k$ induced by undecided vertices (together with some surrounding vertices to be defined soon) form independent instances of $\Pi''(n_0,x)$. First, recall from \Cref{def:variousLCLs2} that $\Pi''=\Pi''(\mathcal{A},n_0,x)$ is the same problem as $\Pi'$, but where some nodes might already have their random bits fixed, and a node $v$ is allowed to output \wild if the predetermined random bits in its $(t_0+r)$-radius neighborhood are such that $\mathcal{A}(n_0)$ run for $t_0$ rounds on the nodes in the $r$-radius neighborhood of $v$ has a probability of failure of at least $\geq x$, where undetermined random bits are chosen uniformly at random.

Let $\mathcal{A}$ and $n_0$ be as in the statement of \Cref{lem:PiPrime} and recall that $x=\sqrt{1/n_0}$.

\begin{figure}[h]
	\centering
	\includegraphics[page=8,width=\textwidth]{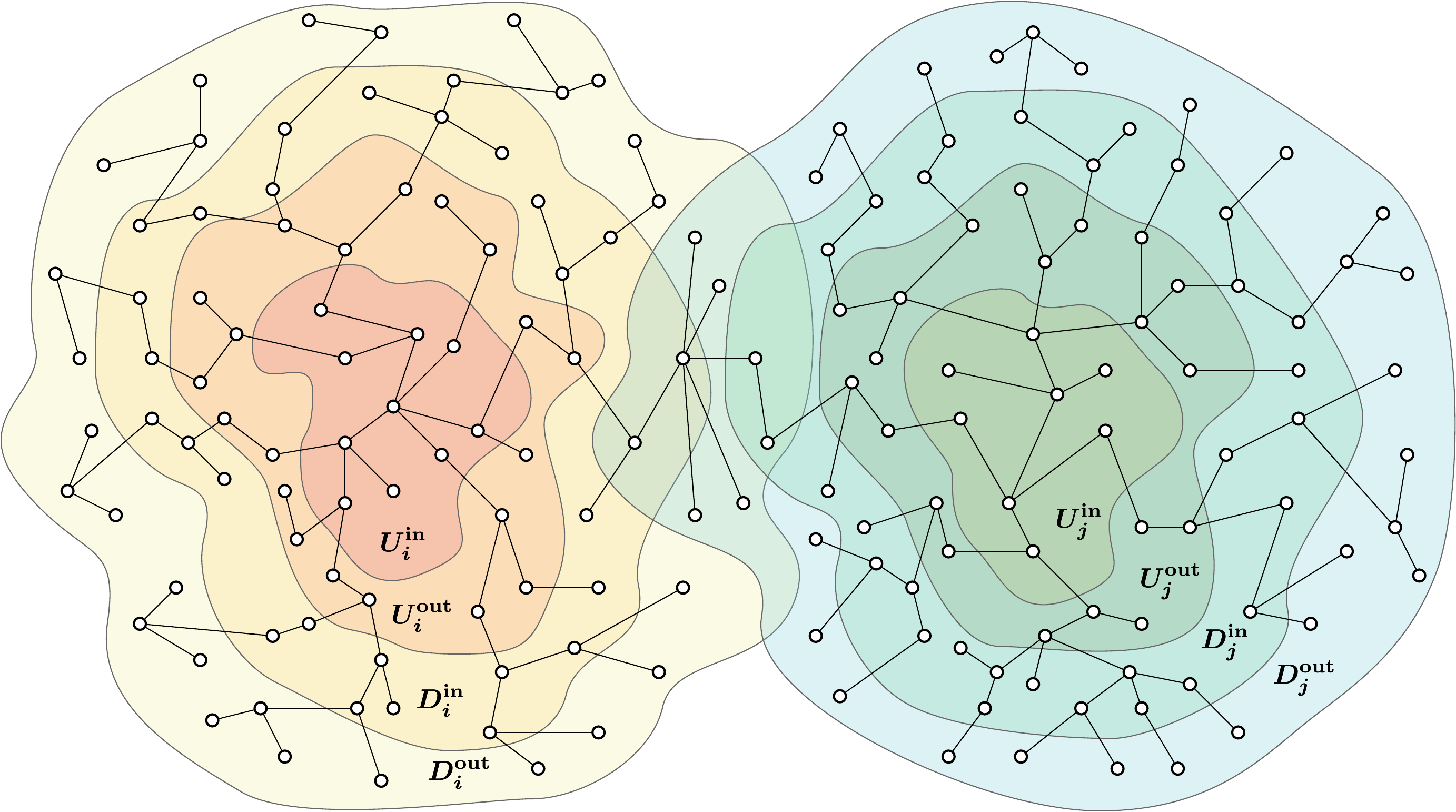}
	\caption{In this illustration we use the values $t_0=1$ and $r=1$. Nodes in $\uin_i$ have their bit strings unfixed, nodes in $\uout_i$ see a node with unfixed bit strings in their checking radius ($t_0+r=2$) but have their own bit strings fixed; nodes in $D_{i}=\din_i\cup\dout_i$ are decided, that is, all bit strings in their $(t_0+r=2)$ neighborhood are fixed and they could run $\mathcal{A}(n_0)$, using these bit strings. }
	\label{fig:smallComp}
\end{figure}

For each component $C^H_i$ we define the four sets of nodes $\uin_i,\uout_i,\din_i,$ and $\dout_i$. For an illustration of their definitions see \Cref{fig:smallComp}.
 Let $\uin_i$ be the nodes of $C^H_i$ that do not have their bit string fixed.  Let $\uout_i$ be the nodes of $C^H_i$ that have their bit string fixed. We obtain $C_i^H=U_i=\uin_i\cup\uout_i$. Let $\din_i$ be nodes that are not in $U_i$ and have at least one node of $\uout_i$ at distance at most $t_0 + r$. Let $\dout_i$ be nodes that are not in $U_i\cup \din_i$ and have a vertex in $\din_i$ at distance at most $t_0 + r$.  Note that $\uin_i$ and $\uout_i$ only contain undecided nodes, and we later prove that $\din_i$ and $\dout_i$ only contain decided nodes.

For each component $C^H_i, i\in[k]$, we obtain an instance of $\Pi''=\Pi''(\mathcal{A}, n_0, x)$ on $C_i^G\subseteq V(G)$ as follows: $C^G_i$ consists of the graph induced by all vertices in $\uin_i \cup \uout_i \cup \din_i \cup \dout_i$. Inputs for $\Pi''$ are inherited from the inputs for $\Pi'$, where, additionally, nodes such that the preshattering phase fixed their bit string are provided with that bit string as input. Note that the defined instances of $\Pi''$ are solvable, because the provided partial assignment either defines an LLL instances, or nodes are allowed to output \wild.

For $1\leq i\leq k$, define $D_i=\din_i \cup \dout_i$ and recall that $U_i=\uin_i\cup \uout_i=C_i^H$ holds by definition.
\begin{observation}\label{obs:independentInstances}
	For any $1\leq i \neq j\leq k$, we have  $D_i \cap U_j= \emptyset$, and hence
$D_i$ only contains  decided nodes.
\end{observation}
\begin{proof}
	As $U_i=C_i^H$ and $U_j=C_j^H$ are maximal connected components in $H$, there is no edge in $H$ between any $u\in U_i$ and $u'\in U_j$. Thus, nodes in $U_i$ are at distance (measured in hops in $G$) strictly larger than $2(t_0 +r)$ from nodes in $U_j$. But by the definition of $D_i$, each node in $D_i$ is in distance at most $2(t_0+r)$ from $U_i$, that is, $D_i$ and $U_j$ must have an empty intersection. Since a node cannot be decided and undecided at the same time,  and since $D_i\cap U_i=\emptyset$ and $D_i\cap U_j=\emptyset$ holds for $j\neq i$, the set $D_i$ only contains decided nodes.
\end{proof}
The size of each instance $C_i^G$ is upper bounded by $N := \Delta^{2(t_0+r)}\cdot |C^H_i|=O(\log n)$. We do not need to bound the total number of instances.  We use \Cref{lem:Pi3Congest} in parallel on all $C_1^G, \ldots, C_k^G$ to solve the $k$ instances of $\Pi''$ in $O(\log N)=O(\log\log n)$ rounds, deterministically, in \CONGEST (IDs are given from the preshattering phase); then we prove that these solutions together with already determined bit strings solve $\Pi'$ on $G$ (the precise lifting of solutions is defined later).
First note that the condition on $x$ required to apply \Cref{lem:Pi3Congest} is satisfied as the assumption on $n_0$ in \Cref{lem:PiPrime} and the choice of $x=\sqrt{1/n_0}$ imply
\begin{align}
	\label{eqn:xisGood}
e\cdot x\cdot \Delta^{4(t_{\mathcal{A}}(n_0)+r)}=e\cdot \sqrt{1/n_0}\cdot \Delta^{4(t_{\mathcal{A}}(n_0)+r)}\leq e \left(\Delta^{2(t_{\mathcal{A}}+r)}e\right)^{-(13+c)}\cdot \Delta^{4(t_{\mathcal{A}}(n_0)+r)}<1~.
\end{align}
Inequality~\ref{eqn:xisGood} has significant slack as the conditions on $n_0$ needed to prove \Cref{claim:smallComponents} are much stricter.

We obtain a solution for $\Pi'$ as follows: for each $i$, each node in $\uin_i$ uses the bit string obtained from the solution of $\Pi''$ in the instance $C^G_i$ (this is well defined, since $\uin_i \cap \uin_j = \emptyset$ for $i\neq j$), while each other node keeps the bit string obtained in the preshattering phase. Due to \Cref{obs:independentInstances}, only decided nodes (that is, nodes in $\din$ and $\dout$) can participate in multiple instances, but as their output comes from the preshattering phase, they do not obtain conflicting outputs from different instances. Further, there is no bandwidth issue as they participate in at most $\Delta^{O(t_0+r)}=O(1)$ instances.

We need to prove that nodes in $\uin_i$ do not output \wild, and that, by combining the solutions of $\Pi''$ obtained for $\uin$ nodes with the solution already provided in input for $\uout$,$\din$, and $\dout$ nodes, we obtain a valid solution for $\Pi'$.

\begin{claim}\label{claim:nowild}
All nodes in $\uin_i \cup \uout_i \cup \din_i$ do not output wildcards in the instance of $\Pi''$ on $C^G_i$.
\end{claim}
\begin{proof}
	A node is allowed to output \wild only if by running $\mathcal{A}(n_0)$ on nodes in the $r$-radius neighborhood of $v$ for $t_0$ rounds we obtain a radius-$r$ neighborhood that with probability strictly greater than $x$ violates the constraints of $\Pi$, where all random bits that are not predetermined by the input are chosen uniformly at random.
	Note that the view of nodes $\uin_i \cup \uout_i \cup \din_i$ at distance at most $t_0 + r$ is the same in $G$ and in $C^G_i$, since all nodes at distance at most $t_0 + r$ are either in $\uin_i \cup \uout_i \cup \din_i$ itself, or in $\dout_i$. As the preshattering process freezes all random bits involved in an event $\mathcal{E}_v$ in a state where the probability of the event is $<x$, it is guaranteed that no node $v\in \uin_i \cup \uout_i \cup \din_i$ can output a wildcard \wild in a solution to $\Pi''$.
\end{proof}
We emphasize that nodes in $\dout_i$ might output wildcards as their $(t_0+r)$-hop view in $C_{i}^G$ might be different from their view in $G$.

\begin{claim}\label{claim:validSolution}
	The obtained solution is valid for $\Pi'$.
\end{claim}
\begin{proof}
	The \LCL{} constraints of $\Pi''$ of a node $u\in U_i$ is satisfied in the instance of $\Pi''$ on $C^G_i$, and $u$ has the same distance-$(t_0+r)$ topological view in $C_i^G$ and in $G$. Let $F$ be the set of vertices in this view.  Due to \Cref{obs:independentInstances} each vertex in $F$ is either already decided or contained in $C_i^G$. Also, by \Cref{claim:nowild}, no node in $F$ outputs \wild in $\Pi''$ on $C_i^G$. Hence, in the lifted solution for $\Pi'$ on $G$ every vertex in $F$ has the same output as in the solution to $\Pi''$ on $C_i^G$ and the \LCL{} constraint of $\Pi'$ of $u$ is satisfied.
	
	Nodes in $\din_i$ and $\dout_i$ have the same view as in the preshattering phase, since at distance $t_0+r$ they only see nodes that have already fixed their output even before solving $\Pi''$, and hence their output is valid.
\end{proof}

\begin{proof}[Proof of \Cref{lem:PiPrime}]
	During this proof we use the notation that has been introduced in this section.
	Under the condition of \Cref{lem:PiPrime},  \Cref{claim:PiPrimeLLL} states that the instance of $\Pi'$ is an LLL instance with dependency graph $H$. Note that $H$ is not a tree. After randomly choosing unique IDs and computing a distance-$2$ coloring of $H$ in $O(\logstar n)$ rounds, the preshattering phase uses $O(\Delta^H)=O(1)$ rounds and sets the bit strings of some vertices. Recall that any vertex who has an unset random bit string within distance $t_0+r$ is called undecided. Due to \Cref{claim:smallComponents}, the undecided vertices of $H$ induce small connected components $C_1^H,\ldots,C_k^H$, each of size at most $O(\log n)$. Now, for each $C_i^H$ we form an instance of $\Pi''(\mathcal{A}, n_0, x)$ on $C_i^G\supseteq C_i^H$ with $x=\sqrt{1/n_0}$.  We use \Cref{lem:Pi3Congest} to solve each of the instances independently. The condition on $x$ is satisfied due to Inequality~\ref{eqn:xisGood} and due to \Cref{obs:independentInstances} simulating one round of the \CONGEST algorithm from \Cref{obs:independentInstances} on all components in parallel can be done in $O(1)$ rounds. Thus, we obtain a solution for each instance of $\Pi''$ on $C_i^G$ in $O(\log \log n)$ rounds.
	
	We obtain a solution for $\Pi'$ as follows: for each $i$, each node in $C_i^H$ with an unset bit string---with the notation of the section these are the nodes  in $\uin_i$---uses the bit string obtained from the solution of $\Pi''$ in the instance $C^G_i$ (this is well defined, since $\uin_i \cap \uin_j = \emptyset$ for $i\neq j$), while each other node keeps the bit string obtained in the preshattering phase.
	Due to \Cref{claim:validSolution} we obtain a valid solution for $\Pi'$.
	
	\textbf{Failure Probability.}  The process can fail either when the computed IDs are not unique which happens with probability $1/n^{c+1}$ or if the components created during the preshattering phase are actually bigger than $O(\log n)$. Due to \Cref{claim:smallComponents}, this happens with probability at most $1/n^{c+1}$. The total failure probability is bounded by $1/n^{c}$.
\end{proof}

\subsection{Proof of Lemma~\ref{lem:Pi3Congest}: Solving \texorpdfstring{\boldmath $\Pi''$}{\textPi''}}
\label{ssec:solvingPi3Prime}
The only remaining thing in this section is to prove \Cref{lem:Pi3Congest}. Its technical statement is standalone.
The core idea to prove the \Cref{lem:Pi3Congest} is that for certain choices of  $x$ and $n_0$ problem $\Pi''$ can be solved in $O(\log n)$ rounds with the randomized LLL algorithm of \cite{CPS17} (see below). We can then apply \Cref{thm:gapCONGESTrakeCompress-treePolyGap} to show that $\Pi''$ can be solved with a deterministic \CONGEST algorithm also running in $O(\log n)$ rounds.

First, we state the result by \cite{CPS17} on fast randomized algorithms for the LLL problem.  It assumes that the communication network is identical to the dependency graph.
\begin{lemma}[LLL on general graphs, \cite{CPS17}]
	\label{thm:CPSLLL}
	There is a randomized \LOCAL algorithm to solve any LLL instance on general dependency graphs in $O(\log n)$ rounds with failure probability $1/n^2$ on graphs of at most $n$ nodes if the instance satisfies the polynomial LLL criterion $epd^2<1$.
\end{lemma}
 The actual result in \cite{CPS17} is more general than the statement of \Cref{thm:CPSLLL}. We restate \Cref{lem:Pi3Congest} and prove it.
\lemPiThreeCongest*
\begin{proof}
	We prove the lemma by designing an algorithm for $\Pi''=\Pi''(\mathcal{A},n_0,x)$. Let $t_0=t_{\mathcal{A}}(n_0)$. At the beginning, each node $v$ gathers its $r + t_0 $ radius neighborhood and checks whether the inputs in its neighborhood are not nice, that is, if \wild would be a correct output according to $\Pi''$. In that case $v$ outputs \wild. Then, any node for which a bit string is provided in its input outputs the same bit string.  We then define an LLL process to solve the problem on nodes that did not output already. The random variables are the unset random bit strings, and we associate a bad event $\mathcal{E}_v$ with each node $v$ as follows: If a node $v$ sees a wildcard in distance $t_0+r$ its \emph{bad event} $\mathcal{E}_v$ is \emph{empty}, i.e., never occurs, otherwise its \emph{bad event} occurs if $\mathcal{A}(n_0)$ run for $t_0$ rounds in the radius-$r$ neighborhood of $v$ with the random bits assigned in the $t_0+r$ neighborhood of $v$ violates the constraints of $\Pi$ at $v$.
	The dependency graph $H$ of the LLL has one vertex for each bad event $\mathcal{E}_v$ and an edge between two events $\mathcal{E}_v$, $\mathcal{E}_u$ if $u$ and $v$ are in distance $\leq 2(t_0+r)$ in $G$.  Let $\Delta_H\leq \Delta^{2(t_0+r)}$ denote the maximum degree of $H$. A \LOCAL algorithm on $H$ can be simulated in $G$ with an $O(t_0)$ multiplicative overhead. 
	
	As we only consider bad events $\mathcal{E}_v$ for vertices who are not allowed to output a wildcard \wild, we obtain  $\Pr(\mathcal{E}_v)\leq x$ and in particular the LLL instance satisfies the following polynomial LLL criterion by the assumption on $n_0$ and $x$.
	\begin{align}\label{eqn:LLLcrit}
		e\cdot Pr(\mathcal{E}_v)\Delta_H^2\leq e\cdot x\cdot  \big(\Delta^{2(t_0+r)}\big)^{2}<1.
	\end{align}
	Due to the polynomial  LLL criterion (\ref{eqn:LLLcrit}), we can use the algorithm of \Cref{thm:CPSLLL} to solve the problem with probability $1/n^2$ in $O(\log n)$ rounds.\footnote{Alternatively we could use any other $\poly\log n$ round randomized LLL algorithm in the \LOCAL model, e.g, the celebrated algorithm by Moser and Tardos \cite{MT10}, even though the expected runtime of the original paper would have to be turned into a with high probability guarantee. The algorithm from \cite{CPS17} in \Cref{thm:CPSLLL} is extremely simple: First, nodes obtain IDs from a range of size $\poly n$, which fails only with probability $1/\poly n$. Then, for $O(\log n)$ iterations,  local ID-minima among the violated bad events resample all their variables.} Hence, the randomized \LOCAL complexity of $\Pi''$ on trees is at most $O(\log n)$.
	
	The mere existence of such a randomized algorithm, that in particular is an $n^{o(1)}$-rounds algorithm that fails with probability at most $1/n$, implies that we can apply \Cref{thm:gapCONGESTrakeCompress-treePolyGap} to obtain a deterministic \CONGEST algorithm that solves $\Pi''$ in $O(\log n) + O(\logstar |\mathcal{S}|)$ rounds if nodes are provided with IDs from an ID space $\mathcal{S}$. This algorithm works even if nodes are provided with just an $l$-distance coloring for a large enough constant distance $l$.
\end{proof}

\section{Separation of CONGEST and LOCAL for General Graphs}\label{sec:separation}
In this section we define an \LCL{} problem $\Pi$ on general bounded-degree graphs, and show that, while $\Pi$ can be solved deterministically in $O(\log n)$ rounds in the \LOCAL model, any randomized \CONGEST algorithm requires $\Omega(\sqrt{n}/\log^2 n)$ rounds. On a high level, the section is structured as follows.
We start by formally defining a family $\mathcal{G}$ of graphs of interest, and then we present a set $\mathcal{C}^{\lproof}$ of local constraints, satisfying that a graph $G$ is in $\mathcal{G}$ if and only if it can be labeled with labels from a constant-size set, such that all nodes satisfy the constraints in $\mathcal{C}^\lproof$. We then define our \LCL{} $\Pi$ in the following way:
\begin{itemize}[noitemsep]
	\item there is a problem $\Pi^{\mathsf{real}}$ such that, on any correctly labeled graph $G\in\mathcal{G}$, nodes must solve $\Pi^{\mathsf{real}}$;
	\item on any labeled graph $G\notin\mathcal{G}$, nodes can either output a locally checkable proof that shows that there exists a node in $G$ that does not satisfy some constraint in $\mathcal{C}^\lproof$ (by solving some \LCL problem that we call $\Pi^\lerr$), or solve $\Pi^{\sf{real}}$ (if possible).
\end{itemize}
Finally, we show lower and upper bounds for $\Pi$ in the \CONGEST and \LOCAL model respectively. The challenging part is to express all these requirements as a proper \LCL, while preventing nodes from ``cheating'', that is, on all graphs $G \in \mathcal{G}$, it must not be possible for nodes to provide a locally checkable proof showing that $G$ is not in $\mathcal{G}$, while for any graph $G \not\in \mathcal{G}$ it should be possible to produce such a proof within the required running time.

Informally, the graphs contained in the family $\mathcal{G}$ of graphs look like the following: we start from a $2$-dimensional grid; we build a binary tree-like structure on top of each column $i$, and let $r_i$ be the root of the tree-like structure in top of column $i$; we use another grid to connect all left-most nodes of these trees; finally, we build on top of these $r_i$ nodes another binary tree-like structure, where $r_i$ nodes are its leaves (see Figure \ref{fig:pyramid-like} for an example). Note that the graphs in $\mathcal{G}$ are the bounded-degree variant of the lower bound family of graphs of Das Sarma et al. \cite{Sarma2012}, where we also add some edges that are necessary in order to make the construction locally checkable.

More formally, we prove the following theorem.
\begin{theorem}\label{thm:separation}
	There exists an \LCL{} problem $\Pi$ that can be solved in $O(\log n)$ deterministic rounds in the \LOCAL model, that requires $\Omega(\sqrt{n} / \log^2 n)$ rounds in the \CONGEST model, even for randomized algorithms.
\end{theorem}

\subsection{The Graph Family \texorpdfstring{\boldmath $\mathcal{G}$}{\mathcal{G}}}

We start by defining the graph family $\mathcal{G}$ in a global way. In order to do so, we first focus on defining a grid structure, then we formally define a tree-like structure, and finally we show how these structures are ``glued'' together.

\begin{figure}
	\centering
	\includegraphics[page=4,width=0.6\textwidth]{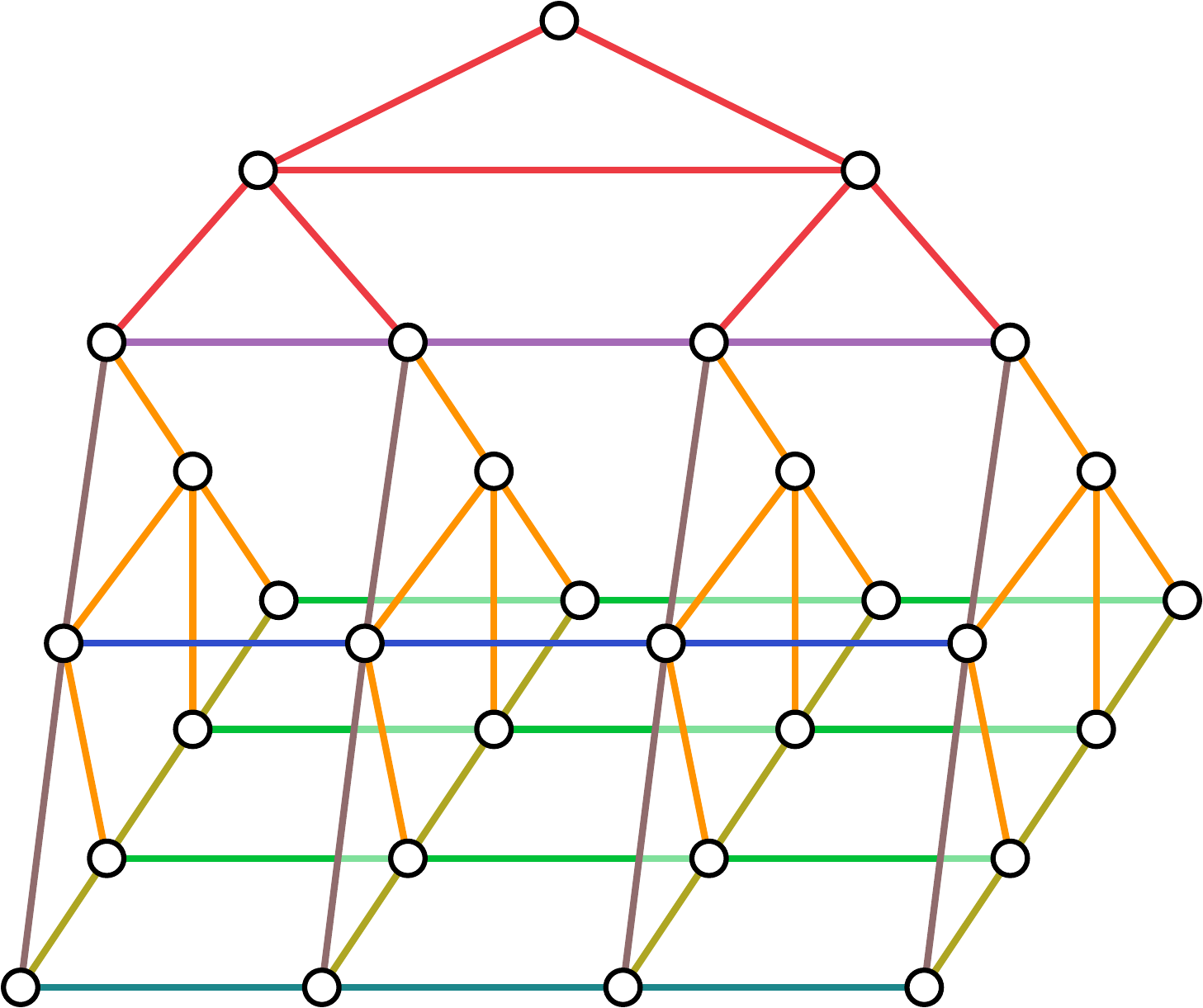}
	\caption{An example of a graph in the family $\mathcal{G}$ of graphs. Edges colored with red, orange, green, and blue, belong exclusively to the structures top-tree, column-tree, bottom-grid, and side-grid, respectively. The other edges belong to exactly two different type of structures and are given the right mixed color.}\label{fig:pyramid-like}
\end{figure}

\begin{definition}[Grid structure]\label{def:gridstructure}
	A graph $G$ is a \emph{grid structure} of size $h \times w$ if it is possible to assign coordinates to its nodes satisfying the following. Let $(x_u, y_u)$ be the coordinates of a node $u$, where $0\le x_u<w$, $0\le y_u < h$. Let $u$ and $v$ be two grid-nodes with coordinates $(x_u, y_u)$ and $(x_v, y_v)$ respectively, such that $x_v \le x_u$, and $y_v \le y_u$. There is an edge connecting $u$ and $v$ if and only if $(x_u, y_u)=(x_v+1, y_v)$, or $(x_u, y_u)=(x_v, y_v+1)$.
\end{definition}

\begin{definition}[Tree-like structure]\label{def:treestructure}
	A graph $G$ is a \emph{tree-like structure} of height $\ell$ if it is possible to assign coordinates to its nodes satisfying the following.
	Let $(l_u, k_u)$ be the coordinates of a node $u$, where  $0\le l_u < \ell$ describes the depth of $u$ in the tree, and $0\le k_u < 2^{l_u}$ the position of $u$ in layer $l_u$ according to some order. Let $u$ and $v$ be two nodes with coordinates $(l_u, k_u)$ and $(l_v, k_v)$ respectively, such that $l_v \le l_u$, and $k_v \le k_u$. There is an edge between $u$ and $v$ if and only if $(l_v,k_v)=(l_u-1,\lfloor \frac{k_u}{2} \rfloor)$, or $(l_u,k_u)=(l_v,k_v+1)$.
\end{definition}

We are now ready to formally define our family $\mathcal{G}$ of graphs, by showing how to glue together grid structures and tree-like structures.
\begin{definition}[The graph family $\mathcal{G}$]
	A graph $G$ is in $\mathcal{G}$ if and only if it can be obtained by the following process.
	Consider one $h\times w$ grid $G_1$, where $h = 2^\ell$ and $w = 2^{\ell'}$, one $\ell \times w$ grid $G_2$, $w$ many tree-like structures $T_i$ ($i\in\{0,1,\dotsc, w-1\}$) of height $\ell$, and one tree-like structure $T$ with height $\ell'$. 
	Each tree-like structure $T_i$ is put on top of a column of the grid $G_1$, that is, each node $u$ of $T_i$ with coordinates $(\ell-1,k_u)$ is identified with the node $v$ of the grid with coordinates $(i,k_u)$ (i.e., node $(\ell-1,k_u)$ of $T_i$ is the same as the grid-node $(i,k_u)$). 
	The grid $G_2$ is put on the side of the trees $T_i$, that is, node $(x,y)$ of $G_2$ is identified with node $(\ell-y-1,0)$ of $T_{x}$.
	Finally, the tree $T$ is put on top of this construction, that is, the root $r_{i}$ of the tree $T_{i}$ is identified with the $i$th leaf of $T$. More precisely, the root node $r_{i}$ is the same as node $(\ell'-1,i)$ of $T$.
\end{definition}

\subsection{Local Checkability}\label{subsec:localCheckability}

In this section, we characterize our graph family $\mathcal{G}$ from a \emph{local} perspective. We define an \LCL problem $\Pi^\lproof$ that is solvable on a graph $G$ if and only if $G\in\mathcal{G}$. We will not require nodes to solve $\Pi^\lproof$, but we will provide to nodes a solution to $\Pi^\lproof$ to prove them that $G\in\mathcal{G}$. In particular, we define a constant-size set of labels and a set of local constraints $\mathcal{C}^\lproof$ on these labels. With ``local'' constraints we mean two things: firstly, each node can check whether it satisfies the constraints in constant time, and secondly, the constraints are globally satisfied if they are satisfied locally at all nodes. The set of constraints $\mathcal{C}^\lproof$ is defined such that the following holds: (i) any graph $G$ that satisfies $\mathcal{C}^\lproof$ at all nodes is in $\mathcal{G}$, and (ii) any graph $G\in\mathcal{G}$ can be labeled such that it satisfies $\mathcal{C}^\lproof$ at all nodes. Intuitively, the set of local constraints $\mathcal{C}^\lproof$ that we define is tailored for graphs in the family $\mathcal{G}$. In other words, $\mathcal{C}^\lproof$ is such that it is trivial to label a graph $G\in\mathcal{G}$ in such a way that all nodes in $G$ satisfy the constraints in $\mathcal{C}^\lproof$. The hard part is to ensure that graphs not in $\mathcal{G}$ cannot be labeled such that $\mathcal{C}^\lproof$ is satisfied by all nodes. For simplicity, we first provide labels and constraints for the two structures that compose a graph in $\mathcal{G}$ separately (i.e., for grid structures and tree-like structures), and then we provide the set $\mathcal{C}^\lproof$ of constraints that determine how these structures can be glued together and result in a graph in $\mathcal{G}$.

Before diving into the details, we provide some useful definitions. Let $L_u(e)$ be the label assigned to the half-edge $(u,e)$, where $e=\{u,v\}$. Also, let $L_1,L_2,\dotsc,L_k$ be $k$ labels, and let $f_u(L_1,L_2,\dotsc,L_k)$ be a function that takes in input a node $u$ and a sequence of labels $L_1,L_2,\dotsc,L_k$. Informally, $f_u(L_1,L_2,\dotsc,L_k)$ returns the node reachable from $u$ by following half-edges labeled $L_1,L_2,\dotsc,L_k$ if and only if this node exists and is unique. More formally, consider a path $P$ that starts on $u = v_1$ and, for $1\le i \le k$, continues on node $v_{i+1}$ reachable from $v_{i}$ by following the edge $e_i = \{v_{i},v_{i+1}\}$ such that $L_{v_{i}}(e_i)=L_i$. If $P$ exists and is unique, then $f_u(L_1,L_2,\dotsc,L_k)$ returns the node $v_{k+1}$ reached by $P$, otherwise $f_u(L_1,L_2,\dotsc,L_k)$ returns $\bot$.

\paragraph{Grid structure.}
Let $\mathcal{E^\lgrid} = \{\lup,\ldown, \lleft, \lright\}$ (each label stands for ``up'', ``down'', ``left'', and  ``right'', respectively) be the possible labels that can be assigned to half-edges.
Assume that each incident half-edge of each node is labeled with exactly one label in $\mathcal{E^\lgrid}$.
The local constraints $\mathcal{C}^{\lgrid}$ are defined as follows:
\begin{enumerate}[label=(\alph*), noitemsep]	
	\item for any two edges $e,e'$ incident to a node $u$, it must hold that $L_u(e)\neq L_u(e')$;\label{cons-grid:differentEdgeLabels}
	
	\item for each edge $e=\{u,v\}$, if $L_u(e)=\lleft$, then $L_v(e)=\lright$, and vice versa; \label{cons-grid:left-right}
	
	\item for each edge $e=\{u,v\}$, if $L_u(e)=\lup$, then $L_v(e)=\ldown$, and vice versa;\label{cons-grid:up-down}
	
	\item if a node $u$ has two incident edges labeled with $\lright$ and $\lup$ respectively, then it must hold that $f_u(\lright,\lup,\lleft,\ldown)=u$. \label{cons-grid:square}
	
	\item if $f_u(\lright)$ exists, then
	$u$ has an incident edge labeled with $\ldown$ (resp. $\lup$) if and only if  $f_u(\lright)$ has an incident edge labeled with $\ldown$  (resp. $\lup$). \label{cons-grid:down-propagates}
	
	\item if $f_u(\lup)$ exists, then
	$u$ has an incident edge labeled with $\lleft$ (resp. $\lright$) if and only if  $f_u(\lup)$ has an incident edge labeled with $\lleft$  (resp. $\lright$). \label{cons-grid:left-propagates}
\end{enumerate}
While the above constraints characterize a grid structure from a local perspective, there are graphs that can be labeled such that the constraints in $\mathcal{C}^{\lgrid}$ are satisfied, but they are not grid structures. Such graphs can be, for examples, torii. Nevertheless, we can prove the following.
\begin{lemma}\label{lem:grid}
	Let $G$ be a graph that is labeled with labels in $\mathcal{E^\lgrid}$ such that $\mathcal{C}^{\lgrid}$ is satisfied for all nodes. Moreover, assume that there exists at least one node that has no incident half-edge labeled $\ldown$ (or $\lup$), and that there exists at least one node that has no incident half-edge labeled $\lleft$ (or $\lright$). Then, $G$ is a grid structure.
\end{lemma}
\begin{proof}
	We prove the statement assuming that there exists at least one node that has no incident half-edge labeled $\ldown$ and at least one node that has no incident half-edge labeled $\lleft$, since the other cases are symmetric.
	By constraint \ref{cons-grid:differentEdgeLabels}, the label of each incident edge of a node in $G$ is different. Constraints \ref{cons-grid:left-right} and \ref{cons-grid:up-down} ensure that edge ``directions'' are given in a consistent manner. By constraint \ref{cons-grid:square}, we have that $G$ locally looks like a grid structure. Constraints \ref{cons-grid:left-right} and \ref{cons-grid:up-down} also guarantee that the number of half-edges labeled $\lright$ (resp. $\lup$) is the same as the number of half-edges labeled $\lleft$ (resp. $\ldown$). Hence, since by assumption there is at least one node with no incident half-edge labeled $\ldown$ and at least one node with no incident half-edge labeled $\lleft$, the graph contains at least one node with no incident edge labeled $\lup$ and a node with no incident edge labeled $\lright$.
	By constraints \ref{cons-grid:down-propagates} and \ref{cons-grid:left-propagates} the boundaries of the grid propagate correctly. In other words, the existence of nodes with no incident half-edges labeled $\lup$, $\ldown$, $\lleft$, and $\lright$, imply that we have a graph $G$ that locally looks like a grid structure everywhere, and that has bottom, left, right, and up boundaries. Hence $G$ is a grid structure.
\end{proof}

Each grid structure can be labeled such that the constraints $\mathcal{C}^{\lgrid}$ are satisfied. Consider two neighboring nodes $v = (x_v,y_v)$ and $u = (x_u,y_u)$. The half-edge $(u,e = \{u,v\})$ is labeled as follows:
\begin{itemize}[noitemsep]
	\item $L_u(e)=\lup$ if $(x_v, y_v)=(x_u, y_u+1)$.
	
	\item $L_u(e)=\ldown$ if $(x_v, y_v)=(x_u, y_u-1)$.
	
	\item $L_u(e)=\lleft$ if $(x_v, y_v)=(x_u-1, y_u)$.
	
	\item $L_u(e)=\lright$ if $(x_v, y_v)=(x_u+1, y_u)$.
\end{itemize}
It is easy to check that this labeling satisfies $\mathcal{C}^{\lgrid}$.

\paragraph{Tree-like structure.}
Let $\mathcal{E^\ltreelike} = \{\lleft,\lright, \lparent,\llch,\lrch\}$ (each label stands for ``left'', ``right'', ``parent'',  ``left child'', and ``right child'', respectively) be the possible labels that can be assigned to half-edges.
Assume that each incident half-edge of each node is labeled with exactly one label in $\mathcal{E^\ltreelike}$.
The local constraints $\mathcal{C}^{\ltreelike}$ are defined as follows:
\begin{enumerate}[label=(\alph*),noitemsep]
	\item for any two edges $e,e'$ incident to a node $u$, it must hold that $L_u(e)\neq L_u(e')$;\label{cons-tree:differentEdgeLabels}
	
	\item for each edge $e=\{u,v\}$, if $L_u(e)=\lleft$, then $L_v(e)=\lright$, and vice versa; \label{cons-tree:left-right}
	
	\item for each edge $e=\{u,v\}$, if $L_u(e)=\lparent$, then $L_v(e)\in\{\llch,\lrch\}$, and vice versa;\label{cons-tree:parent-child}
	
	\item if a node $u$ has an incident edge $e=\{u,v\}$ with label $L_u(e)=\lparent$ such that $L_v(e)=\llch$, then $f_u(\lparent,\lrch,\lleft)=u$;\label{cons-tree:triangle}
	
	\item if a node $u$ has an incident edge $e=\{u,v\}$ with label $L_u(e)=\lparent$ such that $L_v(e)=\lrch$, if $u$ has an incident edge labeled $\lright$, then $f_u(\lparent,\lright,\llch,\lleft)=u$.\label{cons-tree:square}
	
	\item if a node has an incident half-edge labeled $\llch$, then it must also have an incident half-edge labeled $\lrch$, and vice versa;\label{cons-tree:2children}
	
	\item node $u$ does not have an incident half-edge labeled $\lparent$ if and only if it has no incident half-edges labeled $\lleft$ or $\lright$; \label{cons-tree:root}
	
	\item if a node $u$ does not have an incident edge $e$ with label $L_u(e)\in\{\llch, \lrch\}$, then neither do nodes $f_u(\lleft)$ and $f_u(\lright)$ (if they exist);\label{cons-tree:boundarychildren}
	
	\item if a node $u$ has an incident edge $e=\{u,v\}$ with label $L_u(e)=\lparent$ such that $L_v(e)=\lrch$ (resp. $L_v(e)=\llch$), then $u$ has an incident edge labeled $\lright$ (resp. $\lleft$) if and only if $f_u(\lparent)$ has an incident edge labeled $\lright$ (resp. $\lleft$);\label{cons-tree:boundarylr}
	
\end{enumerate}
The above constraints give an exact characterization of tree-like structures. Hence, we prove the following.
\begin{lemma}\label{lem:treelike}
	Let $G$ be a graph that is labeled with labels in $\mathcal{E^\ltreelike}$ such that $\mathcal{C}^{\ltreelike}$ is satisfied for all nodes. Then, $G$ is a tree-like structure.
\end{lemma}
\begin{proof}
	By constraint \ref{cons-tree:differentEdgeLabels} the label of each incident edge of a node in $G$ is different. Constraints \ref{cons-tree:left-right} and \ref{cons-tree:parent-child} ensure that edge ``directions'' are given in a consistent manner. By constraints \ref{cons-tree:triangle} and \ref{cons-tree:square}, $G$ locally looks like a tree-like structure. Constraint \ref{cons-grid:up-down} guarantees that the number of half-edges labeled $\lparent$ is exactly two times the number of half-edges labeled $\llch$ or $\lrch$, and hence, by a counting argument, that there exists at least one node with no half-edges labeled $\llch$ or $\lrch$. By constraint \ref{cons-tree:2children} such a node has neither half-edges labeled $\llch$ nor half-edges labeled $\lrch$. By constraint \ref{cons-tree:boundarychildren} we ensure that the bottom boundary propagates correctly. Also, among nodes that do not have incident half-edges with a label in $\{\llch,\lrch\}$, there must exist a node that does not have an incident half-edge labeled $\lleft$. If this were not the case, then, by constraint \ref{cons-tree:left-right} all such nodes must have a $\lleft$ and a $\lright$ incident half-edge. Since $G$ locally looks like a tree-like structure, it means that these nodes with no incident half-edges with a label in $\{\llch,\lrch\}$ must form a cycle. By constraint \ref{cons-tree:boundarylr}, this propagates on the above layers forming cycles the length of which halves each time we go from a level to the one above. This reaches a contradiction at the top of the structure, since we would have a single node with a self-loop formed by half-edges with labels $\lleft$ and $\lright$. Hence, in the bottom boundary, there must exist a node that does not have an incident half-edge labeled $\lleft$, and by constraint \ref{cons-tree:left-right} (and by a counting argument) there must also exist a node that does not have an incident half-edge labeled $\lright$. Hence, the bottom boundary is a path, as required, and by constraint \ref{cons-tree:boundarylr}, the left and right boundaries are correctly propagated. Thus, each layer contains a path the length of which halves each time we go from a lever to the one above, ensuring that $G$ has a single node $u$ that does not have incident half-edges with a label in $\{\lleft,\lright\}$. By constraint \ref{cons-tree:root}, this node $u$ is the root of $G$. Hence, $G$ locally looks like a valid tree-like structure for all nodes, and it has valid boundaries at the bottom and on the sides. This means that $G$ is a tree-like structure.
\end{proof}

Each tree-like structure can be labeled such that the constraints $\mathcal{C}^{\ltreelike}$ are satisfied. Consider two neighboring nodes $v = (l_v,k_v)$ and $u = (l_u,k_u)$. The half-edge $(u,e = \{u,v\})$ is labeled as follows:
\begin{itemize}[noitemsep]
	\item $L_u(e)=\lright$ if $(l_v, k_v)=(l_u, k_u+1)$.
	\item $L_u(e)=\lleft$ if $(l_v, k_v)=(l_u, k_u-1)$.
	\item $L_u(e)=\lparent$ if $(l_v, k_v)=(l_u-1, \lfloor \frac{k_u}{2} \rfloor)$.
	\item $L_u(e)=\llch$ if $(l_v, k_v)=(l_u+1, 2k_u)$.
	\item $L_u(e)=\lrch$ if $(l_v, k_v)=(l_u+1, 2k_u + 1)$.
\end{itemize}
It is easy to check that this labeling satisfies $\mathcal{C}^{\ltreelike}$. Note that tree-like structures have already been used in the context of \LCL{}s in \cite{BBOS20paddedLCL}, and that we use the same proof ideas in order to locally check these structures.

\paragraph{\boldmath Graphs in the family $\mathcal{G}$.}
There are four ingredients that compose a graph in the family $\mathcal{G}$ of graphs: the grid on the bottom, some column-trees, the grid on the side of the column-trees, and a top-tree. The half-edges of each of these structures have a labeling that satisfies some constraints, as described in the above paragraphs. Informally, we tag each edge with a label that represents in which structure the edge belongs to. Also, to each half-edge we assign a label that is used to check the validity of the structure the edge belongs to. Each edge can be part of multiple structures. See \Cref{fig:pyramid-like-labels} for an example of a correctly labeled graph $G\in\mathcal{G}$. More formally, let $\mathcal{E^\lstruct} = \{\lgridbot, \lgridside, \lcolumntrees, \ltoptree\}$. Each edge must be labeled with a subset of $\mathcal{E^\lstruct}$. Let $s(e)$ be the labels from $\mathcal{E^\lstruct}$ given to edge $e$. Also, for a node $u$, let $s(u) = \bigcup_{e : u \in e} s(e)$. Labels in $\mathcal{E^\lstruct}$ highlight in which structure each edge belongs to (the bottom-grid, the side-grid, a column-tree, or the top-tree), and for an edge of some structure the following holds (note that an edge can be part of more than one structure). 

\begin{itemize}[noitemsep]
	\item For every edge labeled $\lgridbot$ both its half-edges must have labels from $\{\lgridbot\} \times \mathcal{E^\lgrid}$.
	\item For every edge labeled $\lgridside$ both its half-edges must have labels from $\{\lgridside\} \times \mathcal{E^\lgrid}$.
	\item For every edge labeled $\lcolumntrees$ both its half-edges must have labels from $\{\lcolumntrees\} \times \mathcal{E^\ltreelike}$.
	\item For every edge labeled $\ltoptree$ both its half-edges must have labels from $\{\ltoptree\} \times \mathcal{E^\ltreelike}$.
\end{itemize}    
Moreover, given a half-edge $b=(u,e)$, let $l^{\lgridbot}(b)$ be the function that maps $b$ into the label $L$ satisfying that the half-edge $b$ is labeled $(\lgridbot,L)$. We define analogously $l^{\lgridside}$, $l^{\lcolumntrees}$ and $l^{\ltoptree}$.
The local constraints $\mathcal{C}^{\lproof}$ are defined as follows.
\begin{enumerate}[noitemsep]
	\item The following constraints guarantee that each graph induced by edges of the same type satisfies the constraints of the grid and tree-like strictures.
	\begin{enumerate}[noitemsep]
		\item\label{cons-fam:grid} Consider the graph induced by edges labeled $\lgridbot$ (resp.\ $\lgridside$) and the labeling given by $l^{\lgridbot}$ (resp.\ $l^{\lgridside}$). This labeling must satisfy the constraints $\mathcal{C}^{\lgrid}$.
		
		\item\label{cons-fam:tree} Consider the graph induced by edges labeled $\lcolumntrees$ (resp.\ $\ltoptree$) and the labeling given by $l^{\lcolumntrees}$ (resp.\ $l^{\ltoptree}$). This labeling must satisfy the constraints $\mathcal{C}^{\lcolumntrees}$.
	\end{enumerate}
	\item\label{cons-fam:nodetypes} All nodes $u$ must satisfy $|s(u)| = 1$, unless specified differently. Also, no node is allowed to satisfy $s(u) \in \{\{\lgridbot\},\{\lgridside\}\}$.
	\item The following constraints guarantee a good connection between the top-tree, the column-trees, and the side-grid.
	\begin{enumerate}[noitemsep]
		\item\label{cons-fam:top} Each node $u$ such that $s(u) = \{\ltoptree\}$ must have incident half-edges that are labeled $(\ltoptree,\llch)$ and $(\ltoptree,\lrch)$.
		
		\item\label{cons-fam:midtop} A node $u$ can satisfy $s(u) = \{\lcolumntrees, \ltoptree, \lgridside\}$ if its half-edges are labeled either $(L_1,L_2,L_3,L_4,L_5)$, $(L_1,L_2,L_3,L_4)$, or $(L_1,L_2,L_3,L_5)$, where
		\begin{align*}
		L_1 &= \{(\ltoptree,\lparent)\},\\
		L_2 &= \{(\lcolumntrees, \lrch)\},\\
		L_3 &= \{(\lcolumntrees, \llch),(\lgridside,\ldown)\},\\
		L_4 &= \{(\ltoptree,\lleft),(\lgridside,\lleft)\},\\
		L_5 &= \{(\ltoptree,\lright),(\lgridside,\lright)\}.
		\end{align*}
		\item\label{cons-fam:mid} Each node $u$ such that $s(u) = \{\lcolumntrees\}$ must have an incident half-edge labeled $(\lcolumntrees,\lparent)$.
		\item\label{cons-fam:midside} A node $u$ can satisfy $s(u) = \{\lcolumntrees, \lgridside\}$ if its half-edges are labeled either $(L_1,L_2,L_3,L_4,L_5,L_6)$, $(L_1,L_2,L_3,L_4,L_5)$, or $(L_1,L_2,L_3,L_4,L_6)$, where
		\begin{align*}
		L_1 &= \{(\lcolumntrees,\lparent),(\lgridside,\lup)\},\\
		L_2 &= \{(\lcolumntrees, \lrch)\},\\
		L_3 &= \{(\lcolumntrees, \llch),(\lgridside,\ldown)\},\\
		L_4 &= \{(\lcolumntrees,\lright)\},\\
		L_5 &= \{(\lgridside,\lleft)\},\\
		L_6 &= \{(\lgridside,\lright)\}.
		\end{align*}
	\end{enumerate}
	\item The following constraints guarantee a good connection between the column-trees, the bottom-grid, and the side-grid.
	\begin{enumerate}[noitemsep]
		\item\label{cons-fam:midbot} Each node $u$ such that $s(u) = \{\lcolumntrees\}$ must have incident edges labeled $(\lcolumntrees,\llch)$ and $(\lcolumntrees,\lrch)$.
		\item\label{cons-fam:bot} A node $u$ can satisfy  $s(u) = \{\lgridbot, \lcolumntrees\}$ if its half-edges are labeled either $(L_1,L_2,L_3,L_4,L_5)$, $(L_1,L_2,L_3,L_4)$, $(L_1,L_2,L_3,L_5)$, $(L_1,L_2,L_4,L_5)$, $(L_1,L_2,L_3)$, or $(L_1,L_2,L_4)$, where
		\begin{align*}
		L_1 &= \{(\lcolumntrees,\lparent)\},\\
		L_2 &= \{(\lcolumntrees,\lleft),(\lgridbot,\ldown)\},\\
		L_3 &= \{(\lgridbot,\lleft)\},\\
		L_4 &= \{(\lgridbot,\lright)\},\\
		L_5 &= \{(\lcolumntrees,\lright),(\lgridbot,\lup)\}.
		\end{align*}
		\item\label{cons-fam:botside} A node $u$ can satisfy $s(u) = \{\lgridbot, \lcolumntrees, \lgridside\}$ if its half-edges are labeled either $(L_1,L_2,L_3,L_4)$, $(L_1,L_2,L_3)$, or $(L_1,L_2,L_4)$, where
		\begin{align*}
		L_1 &= \{(\lcolumntrees,\lparent),(\lgridside,\lup)\},\\
		L_2 &= \{(\lcolumntrees,\lright),(\lgridbot,\lup)\},\\
		L_3 &= \{(\lgridbot, \lleft),(\lgridside,\lleft)\},\\
		L_4 &= \{(\lgridbot,\lright),(\lgridside,\lright)\}.
		\end{align*}
	\end{enumerate}
\end{enumerate}

The constraints in $\mathcal{C}^\lproof$ are tailored for graphs in $\mathcal{G}$. Hence, while it is not hard to see that any graph in $\mathcal{G}$ can be labeled such that the set $\mathcal{C}^\lproof$ is satisfied at all nodes, the converse is not easy to see. In the following lemma, we prove that all graphs that satisfy the constraints in $\mathcal{C}^\lproof$ must be in the graph family $\mathcal{G}$. This is a property that will be useful when proving an upper bound for the \LCL problem $\Pi$ that we define later in \Cref{sec:LCL-Pi}.

\begin{lemma}
	Let $G$ be a labeled graph such that the set $\mathcal{C}^{\lproof}$ of constraints is satisfied at all nodes. Then $G \in \mathcal{G}$.
\end{lemma}
\begin{proof} We prove the lemma by first showing that there must be a top-tree structure in $G$. Then, we show that this implies that $G$ must contain column-trees that are correctly connected to a side-grid, and that the leaves of these column-trees form a bottom-grid. Finally we put all this together to show that $G\in\mathcal{G}$.
	
	\paragraph{Top-tree structure.} By constraint \ref{cons-fam:nodetypes}, there exists at least an edge in $G$ with label $\lcolumntrees$ or $\ltoptree$. Starting by this, we want to show that $G$ must have an edge labeled $\ltoptree$. Hence, suppose we are in the case where there exists at least an edge labeled $\lcolumntrees$. Constraint \ref{cons-fam:tree}, combined with  \Cref{lem:treelike}, implies that the tree-like structure of each column-tree is correct, and hence that each column-tree has a root. By constraints \ref{cons-fam:mid}, \ref{cons-fam:midside}, \ref{cons-fam:bot}, and \ref{cons-fam:botside}, nodes that do not have incident edges of type $\ltoptree$ must have a parent in a column-tree, that is, an edge labeled $(\lcolumntrees,\lparent)$. These nodes cannot be the root of a column-tree, and the only possible assignment of labels to incident half-edges that allows to not have a parent in a column-tree is the one allowed by constraint \ref{cons-fam:midtop}, which implies that in $G$ there must exist at least an edge labeled $\ltoptree$. By constraint \ref{cons-fam:tree}, combined with \Cref{lem:treelike}, $G$ contains a top-tree with a correct tree-like structure. By constraint \ref{cons-fam:nodetypes} and \ref{cons-fam:top} each non-leaf node of the top-tree has only incident edges of type $\ltoptree$.
	
	\paragraph{Column-tree and side-grid structures.} By constraints \ref{cons-fam:top} and \ref{cons-fam:midtop}, the leaves of the top-tree must be roots of column-trees, that by constraint \ref{cons-fam:tree}, combined with  \Cref{lem:treelike}, implies that all these column-trees have a correct tree-like structure. By constraint \ref{cons-fam:nodetypes} and \ref{cons-fam:midside}, the only additional type of edges incident to each non-leaf node of a column-tree must be $\lgridside$, and by constraint \ref{cons-fam:midside} only the left-most nodes of each layer of the column-trees can be incident to the side-grid.
	
	\paragraph{Bottom-grid structure.} By constraint \ref{cons-fam:midbot} leaves of column-trees must have incident edges of some additional type different from $\lcolumntrees$, that by constraints \ref{cons-fam:nodetypes}, \ref{cons-fam:bot}, and \ref{cons-fam:botside}, this additional type must be $\lgridbot$. By constraint \ref{cons-fam:nodetypes}, \ref{cons-fam:bot}, and \ref{cons-fam:botside}, leaves of column-trees must be connected to the bottom-grid in the desired way. Also, constraints \ref{cons-fam:bot} and \ref{cons-fam:botside} guarantee that the side-grid terminates on the ``down'' boundary of the bottom-grid.
	
	\paragraph{Putting things together.} The correctness of the top-tree implies that there is a left-most leaf in the top-tree. This node is also a side-grid node with no edges labeled $\lleft$ or $\lup$. Moreover, by constraint \ref{cons-fam:midtop}, nodes that are part of the top-tree are also part of the \emph{same} side-grid. By \Cref{lem:grid}, the side-grid is a valid grid structure, that by constraint \ref{cons-fam:midside} is connected correctly on the left-most nodes of the column-trees. This implies that the $i$th leaf of the top-tree is connected to a column-tree, that in turn is connected to the $i$th column of the bottom-grid. By constraint \ref{cons-fam:botside} this implies that there is at least a node of the bottom-grid with no incident edges labeled $\lleft$ or $\ldown$, and by \Cref{lem:grid} this implies that the bottom-grid is a valid grid structure. By constraint \ref{cons-fam:nodetypes} no other edges are allowed to be connected to the bottom-grid.
\end{proof}

\begin{figure}
	\centering
	\includegraphics[page=5,width=0.8\textwidth]{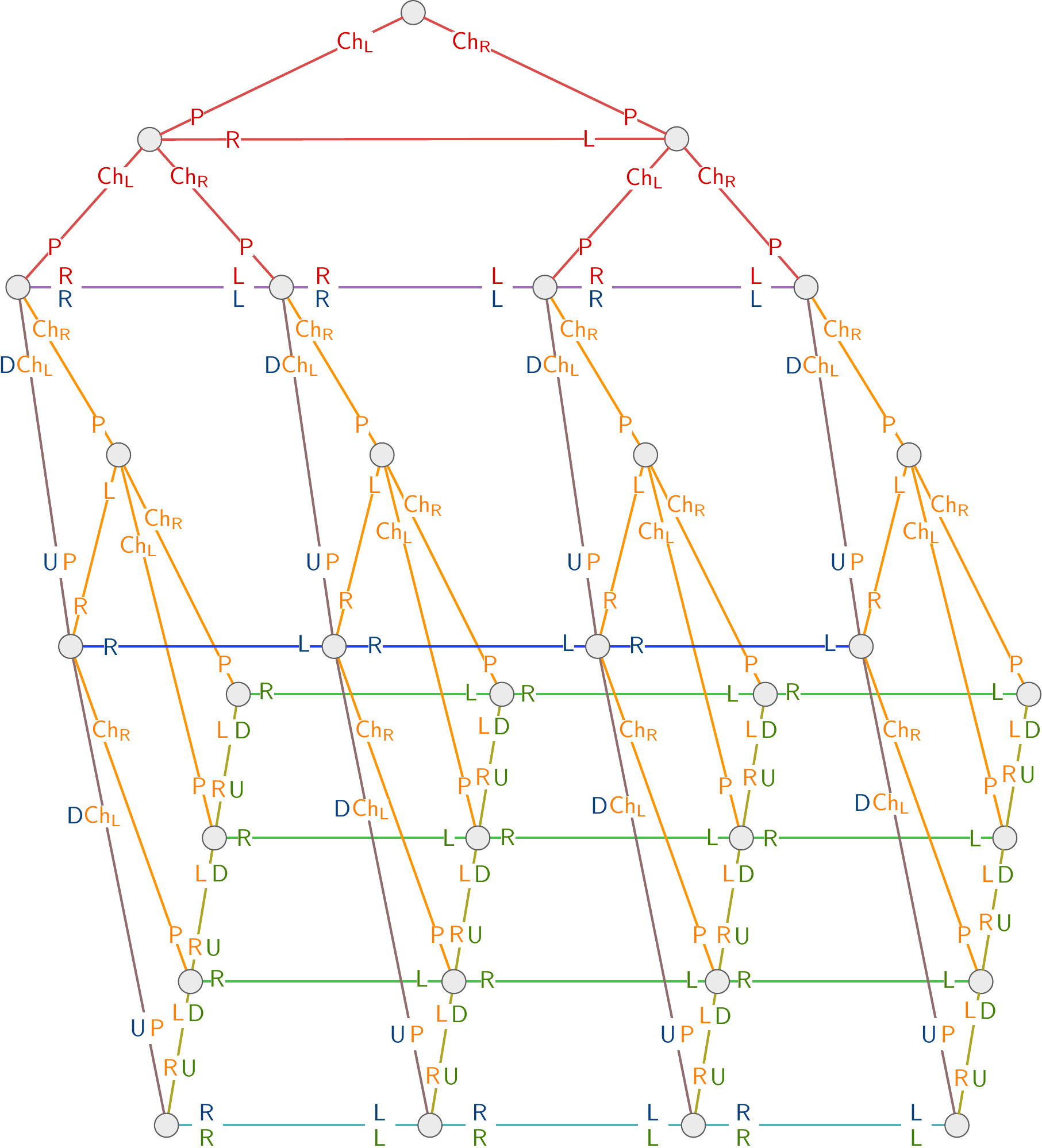}
	\caption{An example of a labeled graph $G\in\mathcal{G}$ where each node locally satisfies the constraints in $\mathcal{C}^\lproof$. Each type of structure has its own color. Edges belonging to two different structures are colored with the sum of the colors of the structures they belong to.}\label{fig:pyramid-like-labels}
\end{figure}

Since the constraints in $\mathcal{C}^{\lproof}$ are tailored for the target family of graphs, it is easy to see that any graph $G \in \mathcal{G}$ can be labeled such that the set $\mathcal{C}^{\lproof}$ is satisfied (see \Cref{fig:pyramid-like-labels} for an example). In fact, when describing the constraints for the grid-structure and the tree-like structure in \Cref{subsec:localCheckability}, we showed a way to label each of them such that the respective local constraints were satisfied. Hence we obtain the following lemma, that will be useful when proving a lower bound for the \LCL problem $\Pi$ that we define later in \Cref{sec:LCL-Pi}.

\begin{lemma}\label{lem:validinstances}
	Any graph $G \in \mathcal{G}$ can be labeled such that each node locally satisfies the constraints in $\mathcal{C}^{\lproof}$.
\end{lemma}

We say that a (labeled) graph $G$ is a \emph{valid} instance if and only if all its nodes locally satisfy the constraints $\mathcal{C}^{\lproof}$. In other words, $G$ is a valid instance if $G\in\mathcal{G}$ and $G$ is correctly labeled such that it locally satisfies everywhere the constraints $\mathcal{C}^{\lproof}$. Otherwise, we say that $G$ is invalid.

\subsection{Proving That a Graph Is Invalid}\label{subsec:Pi-err}

In this section, we define an \LCL $\Pi^{\lerr}$ where, informally, in valid instances, nodes \emph{must} output an empty output, while, in invalid instances, nodes \emph{can} prove that the graph is invalid (or just produce an empty output). While the complexity of this problem is clearly $O(1)$, we will prove that, in invalid instances, nodes can produce a valid non-empty output, that is, a locally checkable proof that the instance is invalid, in $O(\log n)$ rounds in the \LOCAL model.

In \Cref{subsec:localCheckability}, for the sake of readability, we assigned sets of labels on edges and half-edges, while the definition of \LCL{}s in \Cref{sec:definitions} only allows labels on half-edges. Note that it is trivial to convert this kind of labels and constraints such that they satisfy the requirements of a formal definition of an \LCL. Hence, assume that $\mathcal{C}^{\lproof}$ is the set of constraints previously defined, but modified such that it only refers to half-edges. Also, let $\Sigma^{\lproof}_{\lout}$ be the possible half-edge labels appearing in $\mathcal{C}^{\lproof}$. The input labels $\Sigma^{\lerr}_{\lin}$ for $\Pi^{\lerr}$ are the same as $\Sigma^{\lproof}_{\lout}$. The possible output labels $\Sigma^{\lerr}_{\lout}$ of $\Pi^{\lerr}$, with an intuitive explanation of their purpose, are the following:

\begin{itemize}[noitemsep]
	\item $\lerror$, used by nodes that do not satisfy $\mathcal{C}^{\lproof}$;
	\item $(\lpointer,c,p, t)$, where $c \in \{1,2,3\}$ is a counter, $p \in \{\lparent,\lrch,\lleft,\lright\}$ is a direction, $t \in \{\lcolumntrees,\ltoptree\}$ is an edge type, used by nodes to produce pointer chains that point to errors. These kind of labels are referred as \emph{pointers};
	\item $\bot$, an empty output, used by all nodes when the graph is valid, or to accept pointers.
\end{itemize}
The constraints $\mathcal{C}^{\lerr}$ are defined as follows.

\begin{enumerate}[noitemsep]
	\item A node $u$ is allowed to output $\lerror$ only if it does not satisfy $\mathcal{C}^{\lproof}$.\label{con-err:error}
	
	\item A node $u$ can output $(\lpointer,c,p, t)$ on a half-edge $b=(u,e)$ if $t \in s(e)$ and $l^t(b)=p$, that is, $(u,e)$ contains the label $(t, p)$. \label{con-err:preserve-struct}
	
	\item There can be at most one pointer for each edge, that is, nodes cannot point to each other. \label{con-err:consistent}
	
	\item If a node $u$ has an incident half-edge $(u,e)$ labeled $(\lpointer,c,p,t)$, then $v = f_u(p)$ must either output $\lerror$ on all incident half-edges, or output $(\lpointer,c',p',t')$ on at least one incident half-edge $(v,e')$, such that if $t = t'$ then $c = c'$, while if $t\neq t'$, then $c' < c$. \label{con-err:counters}
	
	\item  If a node $u$ has an incident half-edge $(u,e)$ labeled $(\lpointer,c,p,t)$ and $v = f_u(p)$ has an incident half-edge $(v,e')$ labeled $(\lpointer,c,p',t)$, then only the following values of $p$ and $p'$ are allowed.
	\begin{itemize}[noitemsep]
		\item if $p = \lleft$ then $p' = \lleft$;
		
		\item if $p = \lright$ then $p' = \lright$;
		
		\item if $p =\lparent$ then $p' \in \{\lparent,\lleft,\lright\}$;
		
		\item if $p = \lrch$ then $p' \in \{\lrch, \lleft, \lright\}$.
		
	\end{itemize} \label{con-err:nocycles}
\end{enumerate}

In the following lemma, we prove that, for all $G\in\mathcal{G}$ there exists an input labeling such that the only valid output for the nodes is $\bot$.
\begin{lemma}\label{lem:nocheat}
	Let $G$ be a graph in $\mathcal{G}$. There exists an input labeling such that, in order to satisfy $\mathcal{C}^{\lerr}$, all nodes must output $\bot$ on all incident edges.
\end{lemma}
\begin{proof}
	Assume that $G$ is labeled in such a way that all nodes satisfy $\mathcal{C}^\lproof$, which is possible by \Cref{lem:validinstances}. Note that by constraint \ref{con-err:error} no node can output $\lerror$. What remains to be proven is that no node can output pointers. By constraint \ref{con-err:counters}, if a node outputs a pointer, then the pointed node must also output a pointer, and by constraint \ref{con-err:consistent} nodes cannot point to each other. Hence, pointers must produce pointer chains. We now prove that pointer chains cannot close cycles, implying that pointer chains must form paths. These (oriented) paths can only terminate on nodes that output $\lerror$ (constraint \ref{con-err:counters}), and since this is not possible, then nodes cannot output pointers at all. 
	
	By constraint \ref{con-err:preserve-struct} and \ref{con-err:counters}, each time the pointer chain propagates on a different structure, a counter must decrease, but on the same structure counters must be equal. Hence, the only way to produce a cycle is to put the whole cycle on the same tree-like structure. Note that, a necessary condition for creating a cycle in a valid tree-like structure, is to use, on the same pointer chain, both labels $\lparent$ and $\lrch$. The claim then follows since once a pointer chain starts using an $\lleft$ (resp.\ $\lright)$, then it must continue using the same direction. Also, a chain using $\lparent$ (resp.\ $\lrch$) can either continue using it, or switch to $\lleft$ or $\lright$.
\end{proof}

Before showing that it is always possible to prove efficiently that an invalid instance is indeed invalid, we first prove that, on any invalid tree-like structure, it is possible to efficiently produce an error pointer chain that satisfies the requirements. In fact, we prove something stronger, that is, even if the structure is valid, it is possible to produce a pointer chain that ends on a chosen ``marked'' node.
\begin{lemma}\label{lem:pointers}
	Let $G$ be a graph. If $G$ is a valid tree-like structure, assume that at least one node is marked. If $G$ is an invalid tree-like structure, assume that all nodes that do not satisfy $\mathcal{C}^{\ltreelike}$ are marked. Then, there is a deterministic $O(\log n)$-round algorithm for the \LOCAL model that produces consistent pointer chains (according to constraint \ref{con-err:nocycles}) that end in marked nodes, such that each non-marked node has at least an incident edge that is part of a pointer chain.
\end{lemma}

\begin{proof}
	Note that any valid tree-like structure has diameter $T=O(\log n)$.
	Every node $u$ starts by spending $T$ rounds to gather its $T$-radius ball. If $G$ is a valid tree-like structure, then $u$ gathered the whole graph (and sees at least one marked node). If $G$ is an invalid tree-like structure, then $u$ must see at least one marked node, since there must be a node at distance $O(\log n)$ that notices an error locally. In the following, we use regular expressions to shortly denote different kinds of pointer chains. For example, $\lparent^* (\lleft^* | \lright^*)$ denotes a pointer chain composed of some $\lparent$ labels, followed by either some $\lleft$ labels or some $\lright$ labels. A node $u$ chooses the first possible among the available following actions.
	\begin{enumerate}[noitemsep]
		\item If there is a path of the form $\lright^*$ connecting $u$ to a marked node, then $u$ outputs $\lright$.
		\item If there is a path of the form $\lleft^*$ connecting $u$ to a marked node, then $u$ outputs $\lleft$.
		\item If there is a path of the form $\lparent^* (\lleft^* | \lright^*)$ connecting $u$ to a marked node, then $u$ outputs $\lparent$.
		\item If there is a path of the form $\lrch^* (\lleft^* | \lright^*)$ connecting $u$ to a marked node, then $u$ outputs $\lrch$.
	\end{enumerate}
	Note that, since all nodes follow the same priority list, then outputs are consistent. We need to prove that at least one of the four cases always applies. Suppose that, for any marked node $w$, for any path that connects $u$ to $w$, $u$ cannot reach $w$ by following the above rules. This means that, by following any of the above rules, we cannot reach nodes that notice invalidities in the tree-like structure (since, otherwise, they would have been marked). This implies that $u$ also cannot reach any marked node by following paths of the form $\llch^* (\lleft^* | \lright^*)$, since they could be replaced by paths of the form $\lrch^* (\lleft^* | \lright^*)$. Hence, node $u$ cannot reach any marked node by following paths of the form $\lleft^* | \lright^* |  (\lparent^*|\llch^*|\lrch^*) (\lleft^* | \lright^*)$.
	The claim follows by noticing that, in a valid tree-like structure, any node can be reached by following only such paths.
\end{proof}

We are now ready to prove the main result of \Cref{subsec:Pi-err}.
\begin{lemma}\label{lem:solve-pi-err}
	Let $G$ be a graph not contained in $\mathcal{G}$. For any input labeling, it is possible to produce, in $O(\log n)$ rounds in the \LOCAL model, a solution of $\Pi^{\lerr}$ where all nodes have at least one incident half-edge not labeled $\bot$.
\end{lemma}
\begin{proof}
	We first provide an $O(\log n)$ \LOCAL algorithm and then we argue about its correctness and runtime. The algorithm is the following (see \Cref{fig:error-pointers} for an example of execution of the algorithm).
	\begin{enumerate}[noitemsep]
		\item Mark all nodes that do not satisfy $\mathcal{C}^{\lproof}$. These nodes output $\lerror$ on all incident half-edges.
		
		\item For $i = 1, 2, 3$, repeat the following on nodes that still need to produce an output.
		\begin{enumerate}
			\item For each graph induced by edges of the same type $t\in \{\lcolumntrees,\ltoptree\}$, each unmarked node $u$ checks, in $O(\log n)$ rounds, if there is at least one marked node.
			
			\item Apply the algorithm of \Cref{lem:pointers} inside each graph induced by edges of the same type $t$ containing at least one marked node, and output $(\lpointer, i, p, t)$, where $t\in \{\lcolumntrees,\ltoptree\}$ is the type of the subgraph and $p$ is the output obtained by \Cref{lem:pointers}.
			
			\item Mark all nodes that produced at least one pointer.
		\end{enumerate}
	\end{enumerate}
	Clearly, the total runtime of the above algorithm is $O(\log n)$. We need to prove that we obtain a correct output. First of all, note that in any graph $G\in\mathcal{G}$, every node must be contained in at least one tree-like structure. Hence, any node that is not contained in any valid tree-like structure outputs an error. In the first iteration, all nodes that are part of invalid tree-like structures produce valid pointers (this is ensured by \Cref{lem:pointers}), and the algorithm marks all of them at the end of the first iteration. In particular, this implies that, for every column-tree it holds that, either its nodes produced correct pointer chains because it is invalid, or it is valid and it has a root connected to a top-tree. Moreover, if a top-tree is invalid, all its nodes produce correct pointer chains and get marked at the end of the first iteration. Hence, if a valid column-tree is connected to an invalid top-tree, it produces correct pointer chains during the second iteration (since the root of the column-tree is marked after the first iteration, being it also part of a top-tree). Finally, if a column-tree is valid, and the top-tree it is connected to is also valid, in order for the graph to not be contained in $\mathcal{G}$, it means that there is another column-tree $T$ connected to the top-tree, such that $T$ is invalid. In this case, nodes of $T$ produce correct pointer chains, and all of these nodes get marked, during the first iteration. Then, nodes of the top-tree produce correct pointer chains during the second iteration, and nodes of all the other column-trees produce correct pointer chains during the third iteration. Hence, all nodes output at least one pointer label. Inside each structure, pointers use the same counter, and by \Cref{lem:pointers}, they satisfy $\mathcal{C}^{\lerr}$. On the other hand, pointers produced during later iterations have larger counters, and hence they also satisfy $\mathcal{C}^{\lerr}$.
\end{proof}

\begin{figure}
	\centering
	\includegraphics[page=7,width=0.9\textwidth]{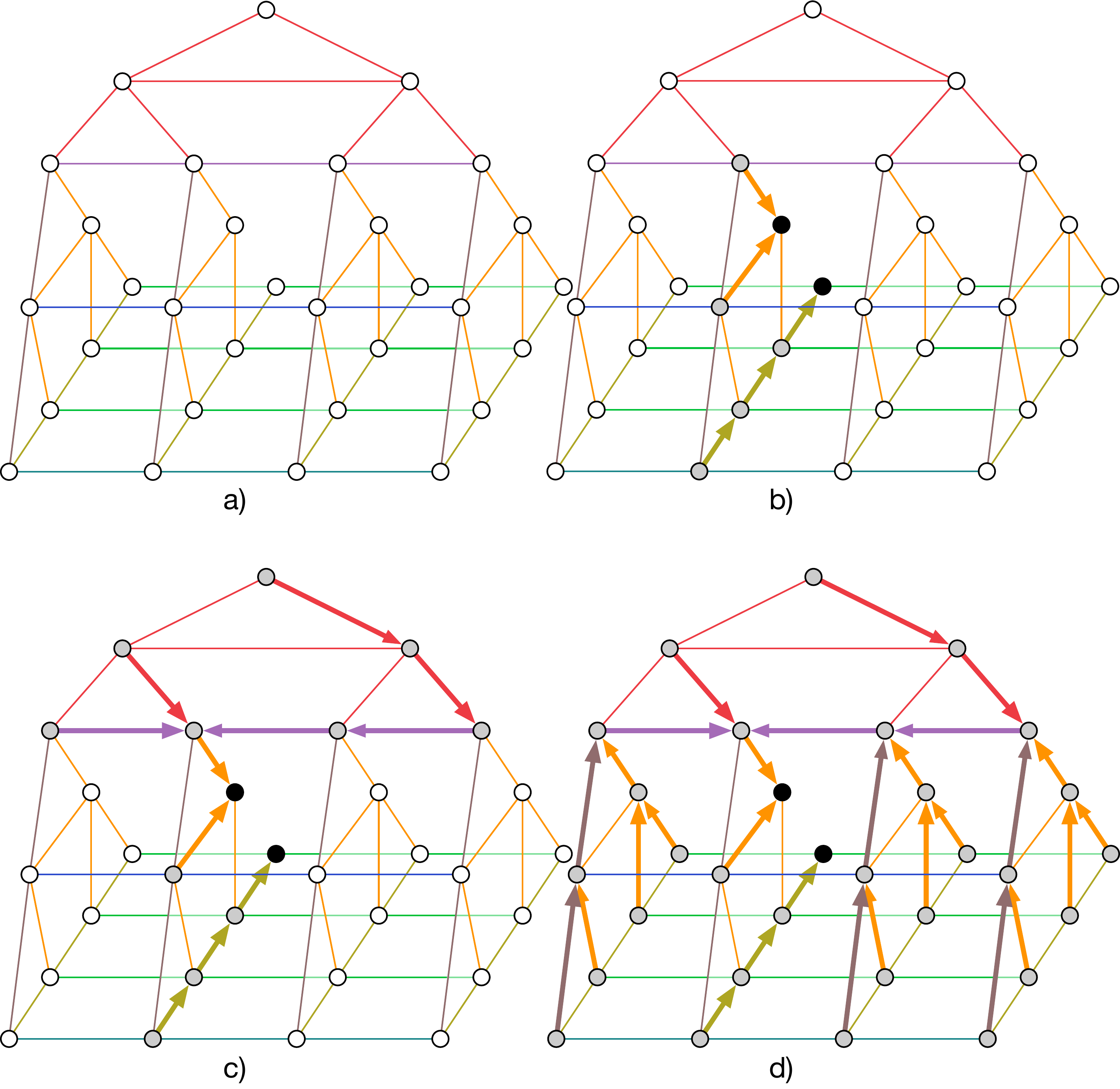}
	\caption{An example of execution of the algorithm described in \Cref{lem:solve-pi-err}. Figure a) shows an invalid graph due to a missing edge between the nodes that are black in Figure b). Figures b), c), and d) show the status after the first, second, and third iteration respectively.}\label{fig:error-pointers}
\end{figure}

\subsection{The \LCL{} Problem \texorpdfstring{\boldmath $\Pi$}{\textPi}}\label{sec:LCL-Pi}

In this section we formally define our \LCL{} problem $\Pi$. Informally, the \LCL{} $\Pi$ is defined as follows. Hard instances $\mathcal{H}$ for the problem are graphs that belong to the family $\mathcal{G}$ and labeled accordingly, where we additionally give to each left-most node of the bottom grid an input from $\{0,1\}$. On this kind of instances, the problem is defined such that every row of the bottom-grid must be labeled with the same value given to its left-most node. This problem corresponds to the problem $\Pi^{\mathsf{real}}$ mentioned in the introduction. 
 While we do not formally define $\Pi^{\mathsf{real}}$, since the notion of \LCL{}s defined in \Cref{sec:definitions} does not allow \LCL{}s with promises, we directly embed this problem in the definition of $\Pi$. On all instances not in $\mathcal{H}$, nodes are allowed to produce a non-empty output for $\Pi^{\lerr}$, and hence prove that the graph is not a valid input instance.

Formally, the set of input labels $\Sigma^{\Pi}_{\lin}$ of our \LCL $\Pi$ is defined as $\Sigma^{\lerr}_{\lin} \times \{0,1,\eps\}$. The set of output labels $\Sigma^{\Pi}_{\lout}$ is defined as $\Sigma^{\lerr}_{\lout} \times \{0,1,\eps\}$. Let $(i_\lerr,i_\Pi)$ be the input of some half-edge $b=(u,e)$. We define $\alpha^{\lerr}(b) = i_\lerr$ and  $\alpha^{\Pi}(b) = i_\Pi$. Also, let $(o_\lerr,o_\Pi)$ be the output on some half-edge $b=(u,e)$. We define $\beta^{\lerr}(b) = o_\lerr$ and $\beta^{\Pi}(b) = o_\Pi$.

The constraints $\mathcal{C}^{\Pi}$ are defined as follows.
\begin{enumerate}[noitemsep]
	\item The input induced by $\alpha^{\lerr}$, combined with the output induced by $\beta^{\lerr}$, must be correct according to $\mathcal{C}^{\lerr}$.\label{cons-pi:nocheat}
	
	\item Any node with at least one incident half-edge $b$ satisfying $\beta^{\lerr}(b) \neq \bot$ is exempt from satisfying the next constraints.\label{cons-pi:solvepi}
	
	\item Every node $u$ such that $\lgridbot \in s(u)$ that does not have an incident half-edge labeled $(\lgridbot,\lleft)$ must satisfy that, if all incident half-edges $b$ have the same input according to $\alpha^\Pi$, then $\beta^{\Pi}(b) = \alpha^{\Pi}(b)$.\label{cons-pi:inout}
	
	\item Every node $u$ such that $\lgridbot \in s(u)$ that does have an incident half-edge $(u,e=\{u,v\})$ labeled $(\lgridbot,\lleft)$ must satisfy that all incident half-edges are labeled with the same label that node $v$ outputs on its half-edge labeled $(\lgridbot,\lright)$.\label{cons-pi:propagate}
\end{enumerate}
Intuitively, the first constraint guarantees that nodes cannot cheat and prove that the graph is invalid on valid instances. The second constraint allows nodes to be exempt on solving the actual problem if the graph is invalid. The third and fourth constraints require nodes in the rows of the grids to propagate the input of the first nodes. See \Cref{fig:LCL-real} for an example of a solution to the \LCL problem $\Pi$ on a graph $G\in\mathcal{G}$.

\begin{figure}
	\centering
	\includegraphics[page=6,width=0.6\textwidth]{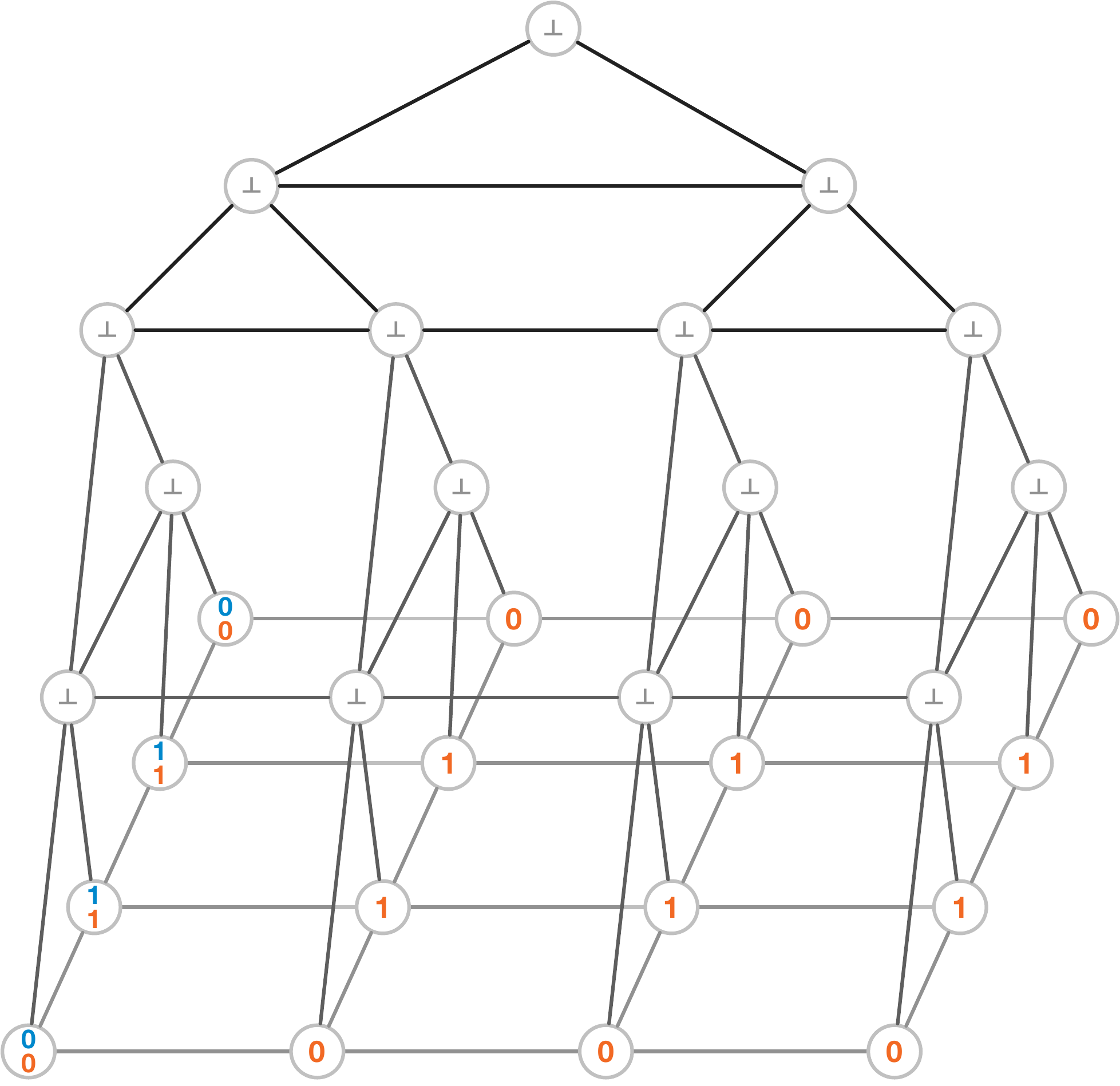}
	\caption{An example of a solution to $\Pi$ on a graph $G\in\mathcal{G}$. The left-most nodes of the bottom-grid have in input either $0$ or $1$ (shown in blue), and give in output the value they have in input. The other nodes of the bottom-grid propagate along the row the output of the let-most nodes. The other nodes output $\bot$.}\label{fig:LCL-real}
\end{figure}

\paragraph{\boldmath Upper bound for $\Pi$ in \LOCAL.}
We now prove that $\Pi$ can be solved efficiently in the \LOCAL model.
\begin{lemma}
	The \LCL problem $\Pi$ can be solved in $O(\log n)$ deterministic rounds in the \LOCAL model.
\end{lemma}
\begin{proof}
	By \Cref{lem:solve-pi-err} nodes can spend $O(\log n)$ rounds to solve $\Pi^\lerr$, such that if $G \not\in \mathcal{G}$, then every node outputs a value different from $\bot$ on at least one incident edge. This implies that if $G \not\in \mathcal{G}$, then nodes can solve $\Pi$ in $O(\log n)$ rounds as well.
	On the other hand, if $G \in \mathcal{G}$, then the diameter of the graph is $O(\log n)$, and hence nodes can spend $O(\log n)$ rounds to gather the whole graph and produce a valid solution for $\Pi$.
\end{proof}

\paragraph{\boldmath Lower bound for $\Pi$ in \CONGEST.}
We now prove that $\Pi$ is a hard problem in the \CONGEST model.
\begin{lemma}
	\label{lemma:general-congest-LB}
	The \LCL problem $\Pi$ requires $\Omega(\sqrt{n}/\log^2{n})$ rounds in the \CONGEST model, even for a randomized algorithm.
\end{lemma}
\begin{proof}
	Consider any graph $G \in \mathcal{G}$, and note that for any $n$ there is a graph of size at least $n$ in the family, such that the bottom-grid has dimensions $s \times s$, for some $s = \Theta(\sqrt{n})$, that is, the grid is a square.
	Moreover, assume that a value in $\{0,1\}$ is provided to each left-most node of the bottom-grid. By \Cref{lem:nocheat}, the only way to solve $\Pi^{\lerr}$ is to output $\bot$ on all nodes, and by constraint \ref{cons-pi:nocheat}  the same holds for $\Pi$ as well. Hence, by constraint \ref{cons-pi:solvepi}, all nodes must satisfy constraints \ref{cons-pi:inout} and \ref{cons-pi:propagate}. Constraint \ref{cons-pi:inout} ensures that, if a value in $\{0,1\}$ is given in input to a left-most node $u$ of the bottom-grid, then it must produce the same output. Constraint \ref{cons-pi:propagate} ensures that all nodes on the same row of $u$ must produce the same output of $u$. Hence, each right-most node in the grid must know the input in $\{0,1\}$ given to the left-most node in the same row. We prove that there is no bandwidth-efficient way to solve this problem in the \CONGEST model. 
	
	The proof follows the lines of the lower bounds of Das Sarma et al.~\cite{Sarma2012}. Since the lower bound graph is slightly different here, we include a proof here for completeness.
	
	We reduce from the 2-party communication problem $P$ in which one of two players, Alice, has an input string $str=\{0,1\}^k$, and the other player, Bob, needs to output that string. A straightforward information theoretic argument gives that the communication complexity of this task is exactly $k$ bits.
	
	Consider a \CONGEST algorithm $A$ for computing $\Pi$ on the graph $G$ with $n=\Theta(k^2)$ in $T$ rounds. If $T\geq k-1$ then $T=\Omega(\sqrt{n})$ and the lemma follows. Otherwise, we show that Alice and Bob can simulate $A$ and produce the output for $P$ by exchanging a number of bits that is $\tilde{O}(T)$, which implies that $T$ has to be at least $\tilde{\Omega}(k)=\tilde{\Omega}(\sqrt{n})$.
	
	The setup for the simulation is as follows. Alice and Bob construct the graph $G$ with $n=\Theta(k^2)$ nodes such that $s=k$ (recall that the grid size is $s\times s$). Alice assigns each node in position $(0,y)$ of the bottom-grid with the input $str[y]$ from her string, for every $0\leq y\leq k-1$. Thus, Alice can simulate the entire distributed algorithm $A$ locally without communication with Bob. On the other hand, while Bob knows the topology of the graph, he does not know the inputs to the nodes $(0,y)$ and therefore cannot simulate the entire distributed algorithm $A$ locally without communication with Alice. However, communicating with Alice in order to simulate all the nodes throughout the algorithm is too expensive, and nullifies the lower bound. Hence, Bob's goal is to simulate only the nodes $(k-1,y)$ for every $0\leq y\leq k-1$ and thus after simulating them in $A$, Bob knows Alice's input as desired. The simulation works by the \emph{moving cut} approach of Das Sarma et al.~\cite{Sarma2012}, in which Bob starts by simulating all nodes except those whose inputs he does not know, and gradually drops additional nodes from the simulation when simulating them becomes too expensive, in a way which still allows the simulation of the remaining nodes to proceed with little communication with Alice.
	
	Formally, we need to define some sets of nodes, as follows. Fix $0\leq j\leq k-1$, and let $C_j= \{ (j,y) \mid 0\leq y\leq k-1\}$ be the set of all nodes in column $j$ of the bottom-grid. Define $W_j$ to be the set of all the nodes in $C_j$ and all the nodes in the levels above that are reachable from nodes in $C_j$ by moving only up to a parent (either in the column-tree, or in the top-tree). That is, let $U_{0,j}=C_j$, and for every $0<t<\ell$ (where $\ell=\Theta(\log n)$ is the number of levels), let $U_{t,j} = \{ u \in V\mid \exists v \in U_{t-1,j}, u=\lparent(v) \}$, where $\lparent(v)$ denotes the parent of $v$. Then $W_j$ is defined to be $\bigcup_{0\leq t < \ell} {U_{t,j}}$.
	
	We now define $B_r$ for $0\leq r < k-1$ as $B_r=\bigcup_{r+1\leq j \leq k-1} {W_j}$.
	The goal is that for every $0\leq r < k-1$, Bob knows the state of all nodes in $B_r$ at the end of round $r$. In particular, at the end of round $k-2$, Bob knows the state of all nodes in $B_{k-2}$, which is equal to the set $W_{k-1}$, which includes all nodes in $C_{k-1}$. Since the assumption is that algorithm $A$ completes in $T\leq k-2$ rounds, Bob knows the outputs of the nodes in $C_{k-1}$ and thus he knows Alice's input. It remains to bound the communication that is needed for Bob to know the state of all nodes in $B_r$ at the end of round $r$, for every $0\leq r < k-1$. We prove by induction on $r$ that by delivering $O(\log^2 n)$ bits from Alice to Bob it is possible for Bob to obtain this information.
	
	The base case is for $r=0$. Notice that $B_0$ is equal to the set $V\setminus \{ (0,y) \mid 0\leq y\leq k-1\}$, which is the set of all nodes except those whose inputs Bob does not know. It holds that Bob knows the initial state of all nodes in $B_0$ at the start of the simulation, or in other words, at the end of round $r=0$.
	
	The induction hypothesis is that the claim holds for some value of $r$, and we show that it holds for $r+1$. That is, we need to show that Bob knows the state of all nodes in $B_{r+1}$ at the end of round $r+1$. First, notice that by definition, it holds that  $B_{r+1}\subseteq B_r$, and therefore Bob knows the state of all nodes in $B_{r+1}$ at the end of round $r$. To obtain their state at the end of round $r+1$, Bob needs to know the messages that are sent to them in round $r+1$. For any node $u$ in $B_{r+1}$ whose neighbors are all in $B_r$, the messages that are sent to $u$ in round $r+1$ can be computed locally by Bob without obtaining any information from Alice, due to the induction hypothesis. It remains to bound the number of messages that need to be sent to a node $u \in B_{r+1}$ in round $r+1$ from neighbors which are not in $B_r$, and to sum this over all nodes $u \in B_{r+1}$.
	
	We consider several cases for a node $u \in B_{r+1}$. The first case is that $u$ is a bottom-grid node (the nodes on the side-grid will be treated as column-tree nodes). In this case all of the neighbors of $u$ are in $B_r$, because $\lparent(u)$ is in $B_{r+1}$ and hence in $B_r$, and if $u=(x,y)$ then its neighbor $(x-1,y)$ is in $B_r$ and its possible neighbors $(x+1,y)$, $(x,y-1)$, and $(x,y+1)$ are in $B_{r+1}$ and hence in $B_r$. Note that this argument is the reason for which we defined $B_{r}$ to drop the grid nodes of the $r$-th column.
	
	The second case is that $u$ is in the column-tree $T_x$ for some column $x$. Then all of its neighbors are either in $T_x$ or in $T_{x-1}$, or in $T_{x+1}$, and hence in $B_r$, except if it is the root of the column-tree (and therefore is in the top-tree).
	
	The remaining case is if $u$ is in the top-tree, in which case the parent of $u$, $\lparent(u)$, its neighbor to the right, $\lright(u)$, and its right child $\lrch(u)$ are also in $B_{r+1}$ and hence in $B_r$. Its remaining neighbors, the neighbor to its left, $\lleft(u)$, and its left child, $\llch(u)$, may be out of $B_r$, but this can happen only for a single node in each level of the top tree. This gives a total of $2\log{n}$ messages that Alice needs to send to Bob in order for him to successfully simulate the nodes of $B_{r+1}$ at the end of round $r+1$, for a total of $O(\log^2{n})$ bits, as claimed.
	
	To conclude, we obtain that if the complexity $T$ of the distributed algorithm $A$ is $T\geq k-1$ then the lower bound directly holds since $k=\Theta(\sqrt{n})$, and otherwise, Bob can simulate $A$ on the nodes $(k-1,y)$ for all $0\leq y \leq k-1$, by receiving $O(T\log^2{n})$ bits from Alice, after which Bob knows Alice's input $str$. Since this number of bits must be at least $k$, we obtain that $T\geq \Omega(k/\log^2 n)$, and the lemma follows.
\end{proof}

\section*{Acknowledgments}
This project was partially supported by the European Union's Horizon 2020 Research and  Innovation Programme under grant agreement no. 755839 (Keren Censor-Hillel, Yannic Maus).
\bibliographystyle{alpha}
\bibliography{references}

\appendix

\section{Shattering}
The following lemma is proven in \cite{FGLLL17} along the lines of a similar proof in \cite{BEPSv3}. While the proof of (P1) requires a careful reasoning, (P2) immediately follows from (P1): If a connected component with more than $O(\log_{\Delta} n \cdot \Delta^{2c_2})$ vertices existed in $H[B]$ one could greedily pick $> \log_{\Delta}n$ vertices from the component to obtain a connected component $U$ in $Z[B]$, that violates (P1).

\begin{lemma}[The Shattering Lemma \cite{FGLLL17} cf. \cite{BEPSv3}]
	\label{lem:shattering}
	Let $H = (V,E)$ be a graph with maximum degree $\Delta_H$. Consider
	a process which generates a random subset $B\subseteq V$ where $Pr(v \in B) \leq \Delta^{-c_1}$, for some constant $c_1 \geq 1$,
	and that the random variables $1(v \in B)$ depend only on the randomness of nodes within at most $c_2$
	hops from $v$, for all $v \in V$ , for some constant $c_2 \geq 1$. Moreover, let $Z = H[2c_2+1,4c_2+2]$ be the graph
	which contains an edge between $u$ and $v$ iff their distance in $H$ is between $2c_2 + 1$ and $4c_2 + 2$.  Then
	with probability at least $1 - n^{-c_3}$, for any constant $c_3$ satisfying $c_1>c_3+ 4c_2 + 2$, we have the following two
	properties:
	\begin{description}
		\item [(P1)] $Z[B]$ has no connected component $U$ with $|U| \geq \log_{\Delta}n$.
		\item[(P2)] Each connected component of $H[B]$ has size at most $O(\log_{\Delta} n \cdot \Delta^{2c_2})$.
	\end{description}
\end{lemma}

\end{document}